\documentclass[a4paper,pra,amsmath,amssymb,floatfix,longbibliography,aps,twocolumn,superscriptaddress,accepted=2024-08-20]{quantumarticle}
\pdfoutput=1
\usepackage{amsfonts,color,physics,enumerate}
\usepackage{dsfont}
\usepackage{latexsym} 
\usepackage{mathtools}
\usepackage{tabularx}

\usepackage[utf8]{inputenc}
\usepackage[english]{babel}
\usepackage[T1]{fontenc}

\usepackage[numbers,sort&compress]{natbib}

\usepackage{amsmath,amsfonts,amssymb,amsthm,graphics,graphicx,epsfig,bbm}
\usepackage{graphicx}
\usepackage{lipsum}
\usepackage{subfigure}
\usepackage{amsmath}
\usepackage{float} 
\usepackage{epsfig}
\usepackage{dcolumn}
\usepackage{bm}
\usepackage{color}
\usepackage{enumerate}
\usepackage{epstopdf}
\usepackage{amssymb}
\usepackage{amstext}
\usepackage{latexsym}
\usepackage{hyperref}
\usepackage{psfrag}
\usepackage{soul,xcolor}
\usepackage[normalem]{ulem}
\usepackage{physics}
\usepackage{footnote}
\usepackage{multirow}
\usepackage{appendix}
\usepackage{mathtools}
\usepackage{xspace}

\begin{document}
\newtheorem{Proposition}{Proposition}[section]

\title{Unraveling the emergence of quantum state designs in systems with symmetry}

\author{Naga Dileep Varikuti}
\email{vndileep@physics.iitm.ac.in}
\affiliation{Department of Physics, Indian Institute of Technology Madras, Chennai, India, 600036}
\affiliation{Center for Quantum Information, Communication and Computing (CQuICC),
Indian Institute of Technology Madras, Chennai, India 600036}
\author{Soumik Bandyopadhyay}
\email{soumik.bandyopadhyay@unitn.it}
\affiliation{Pitaevskii BEC Center, CNR-INO and Dipartimento di Fisica, Universit\`a di Trento, Via Sommarive 14, Trento, I-38123, Italy}	
\affiliation{INFN-TIFPA, Trento Institute for Fundamental Physics and Applications, Via Sommarive 14, Trento, I-38123, Italy
}

\begin{abstract}
Quantum state designs, by enabling an efficient sampling of random quantum states, play a quintessential role in devising and benchmarking various quantum protocols with broad applications ranging from circuit designs to black hole physics. Symmetries, on the other hand, are expected to reduce the randomness of a state. Despite being ubiquitous, the effects of symmetry on quantum state designs remain an outstanding question. The recently introduced projected ensemble framework generates efficient approximate state $t$-designs by hinging on projective measurements and many-body quantum chaos. In this work, we examine the emergence of state designs from the random generator states exhibiting symmetries. Leveraging on translation symmetry, we analytically establish a sufficient condition for the measurement basis leading to the state $t$-designs. Then, by making use of the trace distance measure, we numerically investigate the convergence to the designs. Subsequently, we inspect the violation of the sufficient condition to identify bases that fail to converge. We further demonstrate the emergence of state designs in a physical system by studying the dynamics of a chaotic tilted field Ising chain with translation symmetry. We find faster convergence of the trace distance during the early time evolution {in comparison to the cases when the symmetry is broken.} To delineate the general applicability of our results, we extend our analysis to other symmetries. We expect our findings to pave the way for further exploration of deep thermalization and equilibration of closed and open quantum many-body systems.

\end{abstract}

\maketitle

\newtheorem{theorem}{Result}[section]
\newtheorem{corollary}{Corollary}[theorem]
\newtheorem{lemma}[theorem]{Lemma}
\def\endproof{\hfill$\blacksquare$}

\section{Introduction}
Preparing random quantum states and operators is an essential ingredient to explore a variety of quantum protocols, such as randomized benchmarking \cite{benchmarking1, knill2008randomized, benchmarking2}, randomized measurements \cite{vermersch2019probing, elben2023randomized}, circuit designs \cite{harrow2009random, brown2010convergence}, quantum state tomography \cite{smith2013quantum, merkel2010random}, etc., and has vast applications ranging from quantum gravity \cite{sekino2008fast}, information scrambling \cite{styliaris2021information, pawan}, quantum chaos \cite{haake1991quantum}, information recovery \cite{hayden2007black, yoshida2017efficient}, machine-learning \cite{huang2020predicting, huang2022quantum, holmes2021barren} to quantum algorithms \cite{tilly2022variational}. Quantum state $t$-designs were introduced to answer the pertinent question: \textit{How can one efficiently sample a Haar random state from the given Hilbert space?}
To this end, state designs correspond to finite ensembles of pure states uniformly distributed over the Hilbert space, replicating the behavior of Haar random states to a certain degree \cite{renes2004symmetric, klappenecker2005mutually, benchmarking2}. 
However, generating such states in experiments is a challenging task since it requires precise control over the targeted degrees of freedom with fine-tuned resolution \cite{morvan2021qutrit, proctor2022scalable, boixo2018characterizing}.

Motivated by the recent advances in quantum technologies \cite{gross2017quantum, blatt2012quantum, browaeys2016experimental, gambetta2017building}, the `projected ensemble' framework has been introduced as a natural avenue for the emergence of state designs from quantum chaotic dynamics \cite{cotler2023emergent, choi2023preparing}. Under this framework, one employs projective measurements on the larger subsystem (bath) of a single bi-partite state undergoing quantum chaotic evolution, which generates a set of pure states on the smaller subsystem. These states, together with the Born probabilities, referred to as the projected ensemble, remarkably converge to a state design when the measured part of the system is sufficiently large. This phenomenon, dubbed as \textit{emergent state designs}, has been closely tied to a stringent generalization of regular quantum thermalization. Under the usual framework, predominately characterized through the Eigenstate Thermalization Hypothesis (ETH) \cite{deutsch1991quantum, srednicki1994chaos, ETH_ansatz_expt, d2016quantum, deutsch_18, eth_nonherm}, the bath degrees of freedom are traced out, and the thermalization is retrieved at the level of local observables. Whereas the projected ensemble retains the memory of the bath through the measurements such that thermalization is explored for the sub-system wavefunctions. This generalization has been referred to as \textit{deep thermalization} \cite{ho2022exact, ippoliti2022solvable, ippoliti2023dynamical}.
The emergence of higher-order state designs has been explicitly studied in recent years under various physical settings \cite{cotler2023emergent, ho2022exact, ippoliti2022solvable, lucas2023generalized, shrotriya2023nonlocality, versini2023efficient}, including dual unitary circuits \cite{claeys2022emergent, ippoliti2023dynamical} and constrained physical models \cite{bhore2023deep} with applications to classical shadow tomography \cite{mcginley2023shadow} and benchmarking quantum devices \cite{, choi2023preparing}. 
In the case of chaotic systems without symmetries, arbitrary measurement bases can be considered to witness the emergence of state designs. The presence of symmetries is expected to influence this property.

Symmetries in quantum systems are associated with discrete or continuous group structures. Their presence causes the decomposition of the system into charge-conserving subspaces. This results in constraining the dynamical \cite{ope4, ope5, friedman2019spectral} and equilibrium properties \cite{yunger2016microcanonical} of many-body systems \cite{nakata2023black, bhattacharya2017syk, balachandran2021eigenstate, kudler2022information, chen2020many, paviglianiti2023absence, agarwal2023charge, varikuti2022out}. When a generic system displays symmetry, ETH is known to be satisfied within each invariant subspace \cite{deutsch1991quantum, srednicki1994chaos, d2016quantum}. Deep thermalization, on the other hand, depends non-trivially on the specific measurement basis \cite{cotler2023emergent, bhore2023deep}. Motivated by this, here we ask the intriguing question: \textit{What's the general choice of measurement basis for the emergence of $t$-designs when the generator state abides by a symmetry?} In order to address this question, we first adhere our analysis to generator states with translation symmetry. In particular, we consider the ensembles of the random translation invariant (or shortly T-invariant) states and investigate the emergent state designs within the projected ensemble framework. We then elucidate the generality of our findings by extending its applicability to other discrete symmetries.

This paper is structured as follows. In Sec. \ref{Frame}, we briefly review the projected ensemble framework, outline the central question we are trying to address, and summarize our key results. In Sec. \ref{T-invariant}, we consider the ensembles of translation symmetric states and provide an analytical expression for their moments. In Sec. \ref{first-design} and \ref{second-design}, we study the emergence of first and higher-order state designs, respectively, and outline a sufficient condition on the measurement basis for achieving these designs. This is followed by an analysis of the violation of the condition shown by various measurement bases in Sec. \ref{violation}. We then consider a chaotic Ising chain with periodic boundary conditions in Sec. \ref{sising} and examine deep thermalization in a state evolved under this Hamiltonian. In Sec. \ref{genn}, we generalize the results to other discrete symmetries such as $Z_2$ and reflection symmetries. 
Finally, we conclude this paper in Sec. \ref{discussion}.

\section{Framework and results}\label{Frame}
Here, we briefly outline the projected ensemble framework \cite{ho2022exact, cotler2023emergent} and summarize our main results.
A $t$-th order quantum state design ($t$-design) is an ensemble of pure quantum states that reproduces the average behavior of any quantum state polynomial of order $t$ or less over all possible pure states, represented by the Haar average. An ensemble $\mathcal{E}\equiv\{p_i, |\psi_{i}\rangle\}$ is an exact $t$-design if and only if its moments match those of the Haar ensemble up to order $t$, i.e., 
\begin{eqnarray}
\sum_{i=1}^{|\mathcal{E}|}p_i\left(|\psi\rangle\langle\psi |\right)^{\otimes t}=\int_{|\psi\rangle}d\psi \left(|\psi\rangle\langle\psi|\right)^{\otimes t} 
\end{eqnarray}
The projected ensemble framework aims to generate quantum state designs from a single chaotic or random many-body quantum state. The protocol involves performing local projective measurements on part of the system. First, consider a generator quantum state $|\psi\rangle\in\mathcal{H}^{\otimes N}$, where $\mathcal{H}$ denotes the local Hilbert space of dimension $d$ and $N=N_A+N_B$ denotes the size of the system constituting subsystems-$A$ and $B$. Then, projectively measuring the subsystem-$B$ gives a statistical mixture of pure states (or projected ensemble) corresponding to the subsystem-$A$. To be more precise, projective measurement of $|\psi\rangle$ in a basis $\mathcal{B}\equiv\{|b\rangle\}$ supported over the subsystem-$B$ yielding the state 
\begin{eqnarray}\label{unnor}
|\Tilde{\psi}(b)\rangle =\left(\mathbb{I}_{2^{N_A}}\otimes|b\rangle\langle b|\right)|\psi\rangle \quad \text{ (unnormalized) },  
\end{eqnarray}
with the probability $p_b=\langle\Tilde{\psi}(b)|\Tilde{\psi}(b)\rangle=\langle\psi |b\rangle\langle b|\psi\rangle$. Since the projective measurements disentangle the subsystems, we can safely disregard the subsystem-$B$ and focus on the quantum state of subsystem-$A$. After normalizing the post-measurement state, we obtain 
\begin{eqnarray}
|\phi(b)\rangle=\dfrac{|\Tilde{\psi}(b)\rangle}{\sqrt{p_b}}=\dfrac{\left(\mathbb{I}_{2^{N_A}}\otimes\langle b|\right)|\psi\rangle}{\sqrt{\langle\psi |b\rangle\langle b|\psi\rangle}}.    
\end{eqnarray}
Then, the projected ensemble on $A$ given by $\mathcal{E}(|\psi\rangle, \mathcal{B})\equiv\left\{p_b, |\phi(b)\rangle\right\}$ approximate higher-order quantum state designs if $|\psi\rangle$ is generated by a chaotic evolution \cite{ho2022exact, cotler2023emergent}, i.e., for $t\geq 1$, 
\begin{eqnarray}\label{Deltat}
\Delta^{(t)}_{\mathcal{E}}&\equiv &\left\|\sum_{|b\rangle\in\mathcal{B}}\dfrac{\left(\langle b|\psi\rangle\langle\psi |b\rangle\right)^{\otimes t}}{\left(\langle\psi |b\rangle\langle b|\psi\rangle\right)^{t-1}}-\int\limits_{|\phi\rangle\in\mathcal{E}^{A}_{\text{Haar}}}d\phi \left(|\phi\rangle\langle\phi|\right)^{\otimes t}\right\|_1\nonumber\\
&\leq &\varepsilon.
\end{eqnarray}
In the above equation, the trace norm (or Schatten $1$-norm) of an operator $W$, denoted as $\|W\|_{1}$, is defined as $\|W\|_{1}=\Tr(\sqrt{W^{\dagger}W})$, which is equivalent to the sum of singular values of the operator. The two terms inside the trace norm are the $t$-th moments of the projected ensemble and the ensemble of Haar random states supported over $A$, respectively. It is important to note that \cite{renes2004symmetric} 
\begin{eqnarray}
\int_{|\phi\rangle\in\mathcal{E}^{A}_{\text{Haar}}}d\phi \left(|\phi\rangle\langle\phi|\right)^{\otimes t}= \dfrac{\bm{\Pi}^{A}_{t}}{\mathcal{D}_{A}},    
\end{eqnarray}
where the Haar measure over $\mathcal{E}^{A}_{\text{Haar}}$ is denoted by $d\phi$, $\mathcal{D}_A=d^{N_A}(d^{N_A}+1)\cdots (d^{N_A}+t-1)$, and 
$\bm{\Pi}^{A}_{t}$ is the projector onto the permutation symmetric subspace of the Hilbert space $\mathcal{H}^{\otimes N_{A}}\otimes\mathcal{H}^{\otimes N_{A}}\otimes\cdots t$-times, i.e., $\bm{\Pi}^{A}_{t}=\sum_{\pi_i\in S_t}\pi_i$. The permutation group is represented by $S_t$. The permutation operators $\pi_i$'s act on $t$-replicas of the same Hilbert space ($\mathcal{H}^{\otimes N_{A}}$) as follows:
\begin{eqnarray}
\pi_i\left(|\psi_1\rangle\otimes ...\otimes |\psi_t\rangle\right)=|\psi_{{\pi}^{-1}(1)}\rangle\otimes ...\otimes |\psi_{{\pi}^{-1}(t)}\rangle.
\end{eqnarray}
The trace distance $\Delta^{(t)}_{\mathcal{E}}$ in Eq. (\ref{Deltat}) vanishes if and only if the ensemble $\mathcal{E}(|\psi\rangle, \mathcal{B})$ forms an exact $t$-design \cite{renes2004symmetric, klappenecker2005mutually, benchmarking2}. If the generator state $|\psi\rangle$ is Haar random, the trace distance $\Delta^{(t)}_{\mathcal{E}}$ exponentially converges to zero with $N_B$ for any $t\geq 1$, as demonstrated in Ref. \cite{cotler2023emergent}. In this case, the measurement basis can be arbitrary, and the behavior is generic to the choice of basis. On the other hand, for the generator state abiding by a symmetry, the choice of measurement basis becomes crucial. In this work, we address this particular aspect.  

Our general result can be summarized as follows: Given a symmetry operator $Q$ and a measurement basis $\mathcal{B}$, the projected ensembles of generic $Q$-symmetric quantum states approximate state $t$-designs if for all $|b\rangle\in\mathcal{B}$, $\langle b|\mathbf{Q}_{k}|b\rangle=\mathbb{I}_{2^{N_A}}$, where $\mathbf{Q}_{k}$ represents the projector onto an invariant subspace with charge $k$. Here, it is to be noted that $\mathbf{Q}_{k}$ can be constructed by taking an appropriate linear combination of the elements of the corresponding symmetry group. We employ $\Delta^{t}_{\mathcal{E}}$ as a figure of merit for the state designs. We explicitly derive this condition for the ensembles of translation invariant states and provide arguments to show its generality to other symmetries. We further elucidate the emergence of state designs from translation symmetric states by explicitly considering a physical model. 

\begin{figure}[ht!]
\includegraphics[width=\linewidth, height=7.cm]{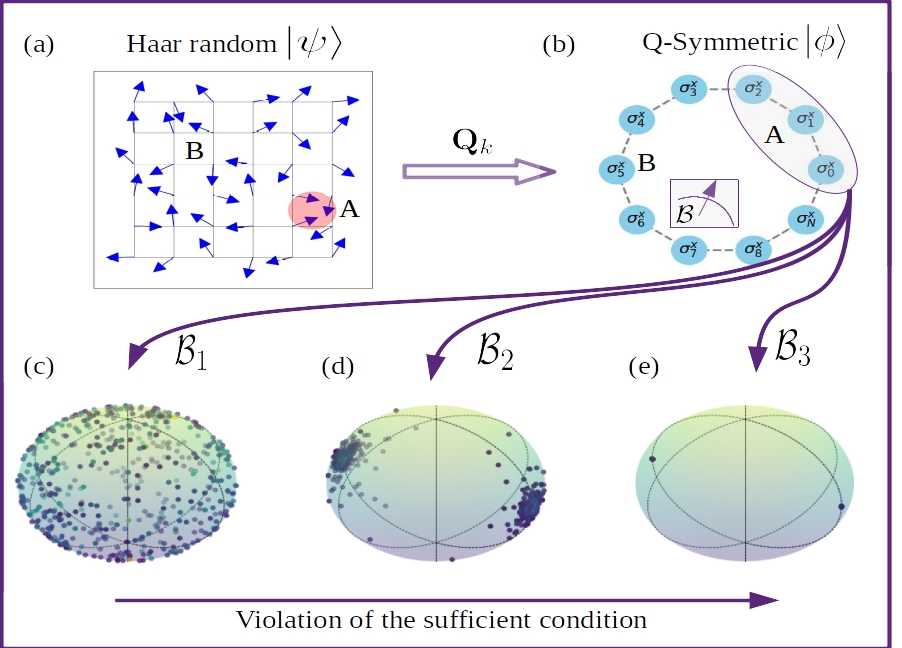}
\caption{\label{fig:sch} Schematic representation of the projected ensemble framework for a $Q$-symmetric state (an eigenstate of a symmetry operator $Q$), showcasing the interplay between measurement bases and symmetry. (a) Haar random state, wherein the spins are randomly aligned and entangled. An application of the subspace projector $\mathbf{Q}_k$ takes the Haar random state to an invariant subspace with charge $k$, which is depicted in (b). For simplicity, we are showing a $Z_2$-symmetric state (see Sec. \ref{z2s}) in this schematic. The labels $\sigma^x_j$ indicate that the state is invariant under spin-flip operations. The resultant state is subjected to projective measurements over the subsystem-$B$. These measurements result in the so-called projected ensembles supported over $A$ (see Sec. \ref{Frame}). (c)-(e) Distributions of the resultant projected ensembles for three different measurement bases. While in (c), the measurements in $\mathcal{B}_1$ yield a uniformly distributed projected ensemble, $\mathcal{B}_2$ and $\mathcal{B}_3$ result in ensembles of pure states localized near $|\pm\rangle$ states. We understand that the latter cases largely violate a sufficient condition we derive in this text for the emergence of state designs. } \end{figure}

\section{T-invariant quantum states}\label{T-invariant}
In this section, we construct the ensembles of translation symmetric (T-invariant) quantum states from the Haar random states and calculate the moments associated with those ensembles. The T-invariant states are well studied in condensed matter physics and quantum field theory for the ground state properties, such as entanglement and phase transitions. Due to the non-onsite nature of the symmetry, the generic T-invariant states (excluding the product states of the form $|\psi\rangle^{\otimes N}$) are necessarily long-range entangled \footnote{An operational definition of the long-range entanglement is as follows: A state $|\phi\rangle\in \mathcal{H}^{\otimes N}$ is long-range entangled if it can not be realized by the application of a finite-depth local circuit on a trivial product state $|0\rangle^{\otimes N}$.}\cite{gioia2022nonzero}. Thus, in a generic T-invariant state, the information is uniformly spread across all the sites, somewhat mimicking the Haar random states. This makes them ideal candidates for generating state designs besides Haar random states. Moreover, in the context of ETH, generic T-invariant systems have been shown to thermalize local observables \cite{santos2010localization, mori2018thermalization, sugimoto2023eigenstate}. Hence, studying deep thermalization in these systems is of profound interest.

Let $T=e^{iP}$ denote the lattice translation operator on a system with a total of $N$ sites, each having a local Hilbert space dimension of $d$, where $P$ is the lattice momentum operator. Then, $T$ can be defined by its action on the computational basis vectors as follows:
\begin{equation}
T|i_1, i_2, ..., i_N\rangle =|i_N, i_1, ..., i_{N-1}\rangle, 
\end{equation}
with $N$-th roots of unity as eigenvalues. A system is considered T-invariant if its Hamiltonian $H$ commutes with $T$, i.e., $[H, T]=0$. A T-invariant state $|\psi\rangle$ is an eigenstate of $T$ with an eigenvalue $e^{-2\pi i k/N}$, i.e., $T|\psi\rangle =e^{-2\pi ik/N}|\psi\rangle$, where $k\in\mathbb{Z}_{\geq 0}$ and $0\leq k\leq N-1$, characterizes the lattice momentum charge. We then represent the set of all pure T-invariant states having the momentum charge $k$ with $\mathcal{E}^{k}_{\text{TI}}$. To compute their moments, we first outline their construction from the Haar random states, followed by the Haar average. 

The translation operator, $T$, generates the translation group given by $\{T^j\}_{j=0}^{N-1}$. Then, a Haar random state $|\psi\rangle\in\mathcal{H}^{\otimes N}$ can be projected onto a translation symmetric subspace by taking a uniform superposition of the states $\{e^{2\pi ijk/N}T^j|\psi\rangle\}_{j=0}^{N-1}$:
\begin{eqnarray}\label{Tsym}
|\psi\rangle\rightarrow |\phi\rangle=\mathbb{T}_{k}(|\psi\rangle)=\dfrac{1}{\mathcal{N}}\mathbf{T}_{k}|\psi\rangle, 
\end{eqnarray}
where
\begin{equation}\label{Tstate}
\mathbf{T}_{k}=\sum_{j=0}^{N-1}e^{2\pi ijk/N}T^j\hspace{-0.1cm}\quad\text{and }
\mathcal{N}=\sqrt{\langle\psi|\mathbf{T}^{\dagger}_{k}\mathbf{T}_{k}|\psi\rangle}.
\end{equation}
Here, the mapping from $|\psi\rangle$ to the T-invariant state $|\phi\rangle$ is denoted with $\mathbb{T}_{k}(|\psi\rangle)$. 
The Hermitian operator $\mathbf{T}_{k}$ projects any generic state onto an invariant momentum sector with the charge $k$. Therefore, $\mathbf{T}_{k}\mathbf{T}_{k'}=N\mathbf{T}_{k}\delta_{k, k'}$, yielding the normalizing factor $\mathcal{N}=\sqrt{N\langle\psi|\mathbf{T}_{k}|\psi\rangle}$. The resultant state $|\phi\rangle$ is an eigenstate of $T$ with the eigenvalue $e^{-2\pi ik/ N}$, i.e., $T|\phi\rangle =e^{-2\pi ik/N}|\phi\rangle$. In this way, we can project the set of Haar random states to an $N$-number of disjoint sets of random T-invariant states, each characterized by the momentum charge $k$. Introducing the translation invariance causes partial de-randomization of the Haar random states. This is because a generic quantum state with support over $N$ sites can be described using $n_{\text{Haar}}\approx d^{N}$ independent complex parameters \cite{linden1998multi}. The translation invariance, however, reduces this number by a factor of $N$, i.e., $n_{\text{TI}}\approx d^N/N$. As we shall see, this results in more structure of the moments of the T-invariant states,  $\mathbb{E}_{\phi\in\mathcal{E}^{k}_{\text{TI}}}\left[|\phi\rangle\langle\phi|^{\otimes t}\right]$.

Before evaluating the moments of $\mathcal{E}^{k}_{\text{TI}}$, it is useful to state the following result: 
\begin{theorem}
Let $U_{\text{TI}}(d^N)$ be a subset of the unitary group $U(d^N)$, which contains all the unitaries that are T-invariant, i.e., $[v, T]=0$ for all $v\in U_{\text{TI}}(d^N)$, then $U_{\text{TI}}(d^N)$ is a compact subgroup of $U(d^N)$.  
\end{theorem}
\begin{proof}
Consider the subset $U_{\text{TI}}(d^N)$ of $U(d^N)$ containing all the T-invariant unitaries --- for every $v\in U_{\text{TI}}(d^N)$, we have $T^{\dagger}vT=v$. Clearly, $U_{\text{TI}}(d^N)$ is a subgroup of $U(d^N)$. As the operator norm of any unitary matrix is bounded, $U_{\text{TI}}(d^N)$ is bounded. Moreover, we can define $U_{\text{TI}}(d^N)$ as the preimage of the null matrix $\mathbf{0}$ under the operation $v-T^{\dagger}vT$ for $v\in U(d^N)$, hence it is necessarily closed. This implies that $U_{\text{TI}}(d^N)$ is a compact subgroup of $U(d^N)$. 
\end{proof}

A method for constructing random translation invariant unitaries using the polar decomposition is outlined in Appendix \ref{polar}. 
An immediate consequence of the above result is that there exists a natural Haar measure on the subgroup $U_{\text{TI}}(d^N)$. It is to be noted that projecting Haar random states onto T-invariant subspaces creates uniformly random states within those subspaces. This means that the distribution of states in $\mathcal{E}^{k}_{\text{TI}}$ is invariant under the action of $U_{\text{TI}}(2^N)$.
To see this, sample $|\phi\rangle$ and $|\phi'\rangle$ from $\mathcal{E}^{k}_{\text{TI}}$ such that they are related to each other via the left invariance of the Haar measure over $U(d^N)$, i.e., $|\phi\rangle=\mathbf{T}_{k}u|0\rangle/\mathcal{N}$ and $|\phi'\rangle=\mathbf{T}_{k}vu|0\rangle/\mathcal{N}$, where $u\in U(d^N)$ and $v\in U_{\text{TI}}(d^N)$. Since $v$ and $\mathbf{T}_{k}$ commute, we can write $|\phi'\rangle=v|\phi\rangle$. 
Let $v$ be sampled according to the Haar measure in $U_{\text{TI}}(d^N)$, then, the state $v|\phi\rangle$ must be uniformly random in $\mathcal{E}^{k}_{\text{TI}}$. Now, $|\phi\rangle$ and $|\phi'\rangle$ being sampled through the projection 
of $\mathbf{T}_{k}$, we can conclude that all the states similarly projected to $\mathcal{E}^{k}_{\text{TI}}$ will be uniformly random. 
We use this result to derive the moments of $\mathcal{E}^{k}_{\text{TI}}$. 

\begin{theorem}\label{Tran_designs}
Let $|\psi\rangle$ be a pure quantum state drawn uniformly at random from the Haar ensemble, then for the mapping $|\psi\rangle\rightarrow |\phi\rangle=\mathbb{T}_{k}(|\psi\rangle)$, it holds that 
\begin{eqnarray}\label{t-moments-res}
\mathbb{E}_{|\phi\rangle\in\mathcal{E}^{k}_{\text{TI}}}\left[\left[|\phi\rangle\langle\phi|\right]^{\otimes t}\right]=\dfrac{1}{\alpha_{t}}\mathbf{T}^{\otimes t}_{k}\bm{\Pi}_{t}, 
\end{eqnarray}
where $\alpha_t$ denotes the normalizing constant, which is given by
\begin{eqnarray}\label{Tmoments}
\alpha_t=\Tr\left(\mathbf{T}^{\otimes t}_{k}\bm{\Pi}_{t}\right).    
\end{eqnarray}    
\end{theorem}
The proof of this result is given in Appendix \ref{trans_designs}. From Eq. (\ref{t-moments-res}), we notice that the T-invariance imposes an additional structure to the moments through the product of $\mathbf{T}^{\otimes t}_{k}$, where the Haar moments are uniform linear combinations of the permutation group elements. 
In the following section, we identify a sufficient condition for obtaining approximate state designs from the generic T-invariant generator states sampled from $\mathcal{E}^{k}_{\text{TI}}$. 

\section{Quantum state designs from T-invariant generator states}\label{emergent_designs}
In this section, we construct the projected ensembles from the T-invariant generator states sampled from $\mathcal{E}^{k}_{\text{TI}}$and verify their convergence to the quantum state designs. In particular, if $|\phi\rangle$ is drawn uniformly at random from $\mathcal{E}^{k}_{\text{TI}}$, we intend to verify the following identity:
\begin{eqnarray}\label{mean_t}
\mathbb{E}_{|\phi\rangle\in\mathcal{E}^{k}_{\text{TI}}} \left(\sum\limits_{|b\rangle\in\mathcal{B}}\dfrac{\left(\langle b|\phi\rangle\langle\phi|b\rangle\right)^{\otimes t}}{\left(\langle\phi|b\rangle\langle b|\phi\rangle\right)^{t-1}}\right)=\dfrac{\bm{\Pi}^{A}_{t}}{\mathcal{D}_{A, t}}, 
\end{eqnarray}
where $\mathcal{D}_{A, t}=2^{N_A}(2^{N_A}+1)...(2^{N_A}+t-1)$.  

The term on the left-hand side evaluates the $t$-th moment of the projected ensemble of a T-invariant state $|\phi\rangle$, with an average taken over all such states in $\mathcal{E}^{k}_{\text{TI}}$. It has been shown that for the Haar random states, the $t$-th moments of the projected ensembles are Lipshitz continuous functions with a Lipschitz constant $\eta=2(2t-1)$ \cite{cotler2023emergent}. Being $\mathcal{E}^{k}_{\text{TI}}\subset\mathcal{E}_{\text{Haar}}$, $\eta$ is also a Lipschitz constant for the case of $\mathcal{E}^{k}_{\text{TI}}$. Note that the number of independent parameters of a T-invariant state with momentum $k$ is $l=2\rank(\mathbf{T}_{k})-1$. Since $\mathbf{T}_k$ is a subspace projector, its eigenvalues assume values of either $N$ or zero. Therefore, $\rank(\mathbf{T}_{k})=\Tr(\mathbf{T}_{k})/N$. Hence, a random T-invariant state can be regarded as a point on a hypersphere of dimension $l=2\Tr(\mathbf{T}_{k})/N-1$. Then, if the above relation in Eq. (\ref{mean_t}) holds, Levy's lemma guarantees that any typical state drawn from $\mathcal{E}^{k}_{\text{TI}}$ will form an approximate state design. In the following, we first verify the identity in Eq. \eqref{mean_t} for $t=1$, followed by the more general case of $t>1$. We then compare the results against the Haar random generator states.

\subsection{Emergence of first-order state designs ($t=1$)}
\label{first-design}
From Eq. (\ref{Tmoments}) of the last section, the first moment of $\mathcal{E}^{k}_{\text{TI}}$ is
$\mathbb{E}_{|\phi\rangle\in\mathcal{E}^{k}_{\text{TI}}}\left(|\phi\rangle\langle\phi|\right) = \mathbf{T}_{k}/\alpha_{1}$, where $\alpha_{1}$ typically scales exponentially with $N$. If $N$ is a prime number, we can write $\alpha^{1}_{k}$ explicitly as follows \cite{nakata2020generic}:
\begin{equation}\label{trace_T_k}
\alpha_{1}=\Tr(\mathbf{T}_{k})=
\begin{cases}
    d^N+d(N-1) & \quad \textrm{if } k=0 \\
    d^N-d & \textrm{otherwise. }
\end{cases}
\end{equation}
To construct the projected ensembles, we now perform the local projective measurements on the $B$-subsystem. For $t=1$, the measurement basis is irrelevant. Then, for some generator state $|\phi\rangle\in\mathcal{E}^{k}_{\text{TI}}$, the first moment of the projected ensemble is given by $\sum_{|b\rangle\in\mathcal{B}}\langle b|\phi\rangle\langle\phi|b\rangle=\Tr_{B}(|\phi\rangle\langle\phi|)$. Typically $\Tr_{B}(|\phi\rangle\langle\phi |)$ approximates the maximally mixed state in the reduced Hilbert space $\mathcal{H}^{\otimes N_A}$. We verify this by averaging the partial trace over the ensemble $\mathcal{E}^{k}_{\text{TI}}$:
\begin{align}\label{N_A_1stmoment_gen}
\mathbb{E}_{\phi\in\mathcal{E}^{k}_{\text{TI}}}\left[\Tr_{B}\left(|\phi\rangle\langle\phi |\right)\right]&=\Tr_{B}\left[\mathbb{E}_{|\phi\rangle\in\mathcal{E}^{k}_{\text{TI}}}\left( |\phi\rangle\langle\phi| \right)\right]\nonumber\\
&=\dfrac{\Tr_{B}(\mathbf{T}_{k})}{{\alpha_{1}}}.
\end{align}
To examine the closeness of $\Tr_{B}(\mathbf{T}_{k})/\alpha_{1}$ to the maximally mixed state, we first expand $\mathbf{T}_{k}$ given in Eq. (\ref{Tstate}) as
\begin{align}\label{Ptrace}
\Tr_{B}(\mathbf{T}_{k})=2^{N}\left[\dfrac{\mathbb{I}_{A}}{2^{N_A}}+\dfrac{1}{2^N}\sum_{j=1}^{N-1}e^{2\pi ijk/N}\Tr_{B}(T^j)\right].
\end{align}
For $j=1$, it can be shown that $\Tr_{B}(T)=T_{A}$, where $T_A$ is the translation operator acting exclusively on the subsystem-$A$. For $j\geq 2$, $\Tr_{B}(T^j)=\sum_{b}\langle b|T^j|b\rangle$ outputs a random permutation operator supported over $A$ whenever $N_A\geq \gcd{(N, j)}$. If $N_A<\gcd{(N, j)}$, the partial trace would result in a constant times identity operator $\mathbb{I}_{2^{N_A}}$. Interested readers can find more details in Appendix \ref{ptrace}. If $N$ is prime and $2\leq j<N$, we have $\gcd{(N, j)}=1$, which is less than $N_A$ whenever $N_A>2$. Then, we can explicitly show that the trace norm of the second term on the right side in Eq. (\ref{Ptrace}) is bounded from above as follows: 
\begin{eqnarray}\label{1-design}
\dfrac{1}{2^N}\left\| \sum_{j=1}^{N-1}e^{2\pi ijk/N}\Tr_{B}(T^j) \right\|_{1}&\leq &\sum_{j=1}^{N-1}\dfrac{\left\| \Tr_{B}(T^j) \right\|_{1}}{2^N}\nonumber\\
&=&\dfrac{(N-1)}{2^{N_B}},
\end{eqnarray}
implying that it converges to a null matrix exponentially with $N_B$. Hence, Eq. (\ref{N_A_1stmoment_gen}) converges exponentially with $N_B$ to the maximally mixed state ($\rho_{A}\propto\mathbb{I}$) in the reduced Hilbert space $\mathcal{H}^{\otimes N_A}$ --- for $N_B\gg\log_2(N)$, we have $\mathbb{E}_{\phi}[\sum_{b}\langle b|\phi\rangle\langle\phi|b\rangle]\approx\mathbb{I}_{A}/2^{N_A}$. Then, as mentioned before, we can invoke Levy's lemma to argue that a typical $|\phi\rangle\in\mathcal{E}^{k}$ approximately generates a state-$1$ design. 

While we initially assumed that $N$ is prime, the results also hold qualitatively for non-prime $N$. In the latter case, for $N_A<\gcd(N, j)$, partial traces may yield identity operators, i.e., $\sum_{b}\langle b|T^j|b\rangle\propto\mathbb{I}_{2^{N_A}}$. These instances introduce slight deviations from Eq. (\ref{1-design}) but have a minor impact on the overall result. Since the error is exponentially suppressed, it is natural to expect that $\Delta^{(1)}$ for these states typically shows identical behavior as that of the Haar random states, and we confirm this with the help of numerical simulations.

Figure \ref{fig:haar_vs_tran} illustrates the decay of the trace distance versus $N_B$ for the first three moments, considering three different bases (see the description of the Fig. \ref{fig:haar_vs_tran} and the following subsection) and two momentum sectors. The blue-colored curves correspond to the first moment. The blue curves, irrespective of the choice of basis and the momentum sector, always show exponential decay, i.e., $\overline{\Delta^{(1)}}\sim 2^{-N/2}$, where the overline indicates that the quantity is averaged over a few samples. We further benchmark these results against the case of Haar random generator states. For the Haar random states, the results are shown in dashed curves with the same color coding used for the first moment. We observe that the results are nearly identical in both cases, with minor fluctuations attributed to the averaging over a finite sample size.

\begin{figure}[ht!]
\includegraphics[scale=0.35]{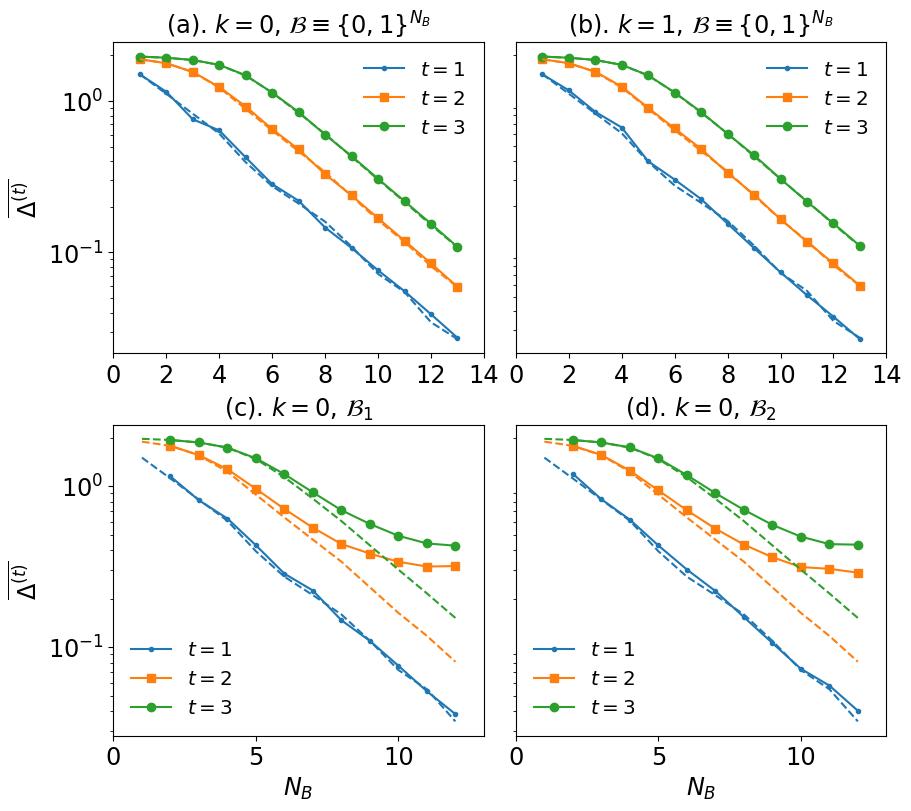}
\caption{\label{fig:haar_vs_tran} Illustration of $\overline{\Delta^{(t)}}$ versus $N_B$ for the projected ensembles of T-invariant generator states sampled uniformly at random from $\mathcal{E}^{k}_{\text{TI}}$. We fix $N_{A} = 3 $ and show the results for the first three moments. In the top panels (a) and (b), the measurements are performed in the computational basis: $\{\Pi_b = |b\rangle\langle b|$ for all $b\in \{0, 1\}^{N_B}\}$, where $\{0, 1\}^{N_B}$ denotes the set of all $N_B$-bit strings. While (a) represents the numerical computations in the $k=0$ momentum sector, $k=1$ is considered in panel (b). The results are averaged over $10$ initial generator states. The results appear qualitatively similar for both the momentum sectors. The trace distance falls to zero with an exponential scaling $\sim 2^{-N_B/2}$ for all the moments calculated. The calculations in (c) and (d) are performed in the momentum sector $k=0$ for two different measurement bases. In (c), we perform the measurements in the eigenbasis of the local translation operator supported over $B$ --- $T_B$. In (d), we break the local translation symmetry of $T_B$ by applying a local Haar random unitary. We perform measurements in the eigenbasis of the resulting operator. See the main text for more details. }   
\end{figure}

\subsection{Higher-order state designs ($t>1$)}
\label{second-design}
Here, we provide a condition for producing approximate higher-order state designs from the random T-invariant generator states. 
    
\begin{theorem}\label{sufficient}
{\textup{(Sufficient condition for the identity in Eq. (\ref{mean_t}))}} 
Given a measurement basis $\mathcal{B}$ having supported over the subsystem-$B$, then the identity in Eq. (\ref{mean_t}) holds if for all $|b\rangle\in\mathcal{B}$, $\langle b|\mathbf{T}_{k}|b\rangle =\mathbb{I}_{2^{N_A}}$. 
\end{theorem}
The proof is given in Appendix \ref{verification}. For a given basis vector $|b\rangle\in\mathcal{B}$, the condition is maximally violated if it can be extended to have support over the full system such that it becomes an eigenstate (with momentum charge $k$) of the translation operator. That is, if there exists an arbitrary $|a\rangle\in\mathcal{H}^{\otimes N_A}$ such that $T\left(|a\rangle\otimes |b\rangle\right) =e^{-2\pi ik/N}\left(|a\rangle\otimes |b\rangle\right)$. Then, the expectation of $\mathbf{T}_{k}$ in this state becomes $\langle ab|\mathbf{T}_{k}|ab\rangle =N$, which is in maximal violation of the sufficient condition \footnote{Note that the sufficient condition would require $\langle ab|\mathbf{T}_{k}|ab\rangle =1$}. For example, when $k=0$, the basis states $|0\rangle^{\otimes N_B}$ and $|1\rangle^{\otimes N_B}$ of the standard computational basis can be extended to $|0\rangle^{\otimes N}$ and $|1\rangle^{\otimes N}$ respectively. So, $\langle 0|^{\otimes N}\mathbf{T}_0|0\rangle^{\otimes N}=\langle 1|^{\otimes N}\mathbf{T}_0|1\rangle^{\otimes N}=N$, correspond to the maximal violation. On the other hand, most of the basis vectors of the computational basis satisfy the sufficient condition. It is also worth noting that if $|ab\rangle$ becomes a T-invariant state with a different momentum charge ($k'\neq k$), then $\langle ab| \mathbf{T}_{k} |ab\rangle=0$. Note that when the condition holds, we get $\langle ab|\mathbf{T}_{k}|ab\rangle =1$.    

The sufficient condition on the measurement basis for the emergent state designs has been deduced from the satisfiability of Eq. (\ref{mean_t}). Intuitively, the sufficient condition requires that the eigenbasis of $\mathbf{T}_{k}$ and the measurement basis ($\mathcal{B}\otimes \mathcal{A}$, where the basis states in $\mathcal{A}$ are supported over $\mathcal{H}^{\otimes N_A}$) should be nearly mutually unbiased. In other words, the measurement basis should have minimal correlation with the translation symmetric states with a given momentum charge --- for all $|b\rangle\in \mathcal{B}$ and any $|\phi\rangle\in\mathcal{E}^{k}_{\text{TI}}$, we have $\langle ab|\phi\rangle\langle\phi|ab\rangle\sim D^{-1}$, where $|a\rangle\in \mathcal{H}^{\otimes N_A}$ can be arbitrary and $D=d^{N}$ is the total Hilbert space dimension. In addition, when a random translation symmetric state is partially measured in such a basis, all the states in the projected ensemble are almost equally likely. Thus, when the condition is satisfied, all the outcomes on the unmeasured part approximate Haar random states in $\mathcal{H}^{\otimes N_A}$. Such bases can be referred to as ``symmetry-non-revealing," analogous to the ``energy-non-revealing" bases recently studied in Ref. \cite{mark2024maximum}. 

We now calculate the trace distance $\Delta^{(t)}$ and average it over a few sample states taken from $\mathcal{E}^{k}_{\text{TI}}$. We illustrate the results in Fig. \ref{fig:haar_vs_tran} for the second (orange color) and third (green color) moments by keeping $N_A$ and the measurement bases as before. 
Similar to the case of the first moment, we contrast the results with the case of Haar random generator states represented by the dashed curves. 
In the panels \ref{fig:haar_vs_tran}a and \ref{fig:haar_vs_tran}b corresponding to generator states with $k=0$ and $1$,  projective measurements in the computational basis show an exponential decay of the average trace distance with $N_B$ for both the moments. Additionally, on a semilog scale, the decays for all the moments appear to align along parallel lines at sufficiently large $N_B$ values. From the comparison between the Figs.\ref{fig:haar_vs_tran}a and \ref{fig:haar_vs_tran}b, it is evident that there are no noticeable differences when the generator states are chosen from different momentum sectors. We also consider the eigenbasis of the local translation operator $T_B$ for the measurements, of which a representative case for $k=0$ is shown in Fig. \ref{fig:haar_vs_tran}c. We observe the average trace distances deviate from the initial exponential decay and approach non-zero saturation values. In Fig. \ref{fig:haar_vs_tran}d, we consider the case where the translation symmetry is broken weakly by applying a single site Haar random unitary to the left of $T_B$, i.e., $T'_B=(u\otimes\mathbb{I}_{2^{N_B-1}})T_{B}$. We then consider the eigenbasis of the resultant operator $T'_B$ for the measurements on $B$. Surprisingly, the trace distance still saturates to a finite value for both the moments despite the broken translation symmetry. 
In the following subsection, we elaborate more on the interplay between the measurement bases and the sufficient condition derived in Result. \ref{sufficient}.

\subsection{Overview of the bases violating the sufficient condition} 
\label{violation}

\begin{figure}[ht!]
\includegraphics[scale=0.35]{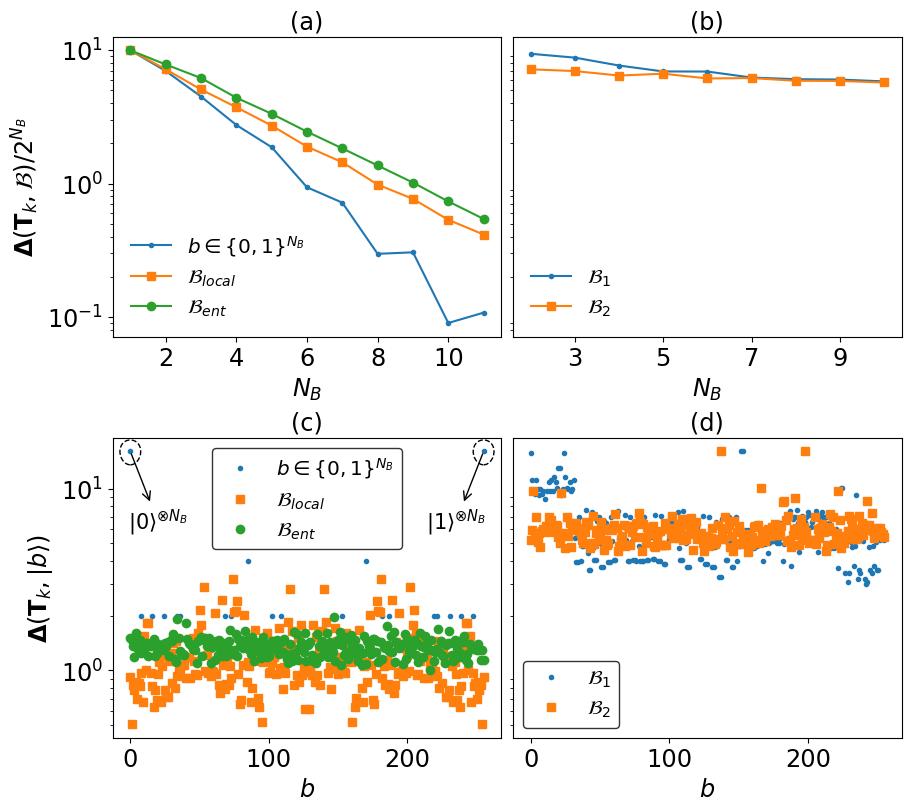}
\caption{\label{fig:viol}  The figure illustrates the average violation of the sufficient condition by different bases as quantified by $\bm{\Delta}(\mathbf{T}_{k}, \mathcal{B})/2^{N_B}$. {Here, we fix $N_{A}=3$.} In panel (a), $\bm{\Delta}(\mathbf{T}_{k}, \mathcal{B})/2^{N_B}$ versus $N_B$ is plotted for three different bases, namely, (i) the computational basis or $\sigma^z$ basis (blue), (ii) basis obtained by applying local Haar random unitaries on the computational basis (orange), and (iii) an entangled basis obtained by applying a Haar unitary of dimension $2^{N_B}$ on the computational basis (green). For all three bases, $\bm{\Delta}(\mathbf{T}_{k}, \mathcal{B})$ decays exponentially with $N_B$. In panel (b), the measurements are performed in the eigenbases of the operators $T_B$ and $u_{N_A+1}T_B$, where $T_B$ is the local translation operator supported over the subsystem $B$ and $u_{N_A+1}$ denotes a local Haar unitary acting on a site labeled with $N_A+1$. The violation stays nearly constant with $N_B$ for these two bases. In the bottom panels (c) and (d), the violation is quantified for each basis vector by fixing $N=11$. Here, we plot $\|\langle b| \mathbf{T}_{k} |b\rangle\|_1$ for each $|b\rangle\in\mathcal{B}$ for all the bases considered in the above panels. See the main text for more details. }   
\end{figure}

From Fig. \ref{fig:haar_vs_tran}, it is evident that not all measurement bases furnish higher-order state designs. Here, we analyze the degree of violation of the sufficient condition by different measurement bases. 
Some, like the computational basis, exhibit mild violations, while others significantly deviate from the condition. Given a measurement basis $\mathcal{B}$, we quantify the average violation of the sufficient condition using the quantity $\bm{\Delta}(\mathbf{T}_{k}, \mathcal{B})/2^{N_B}$, where 
\begin{align}
\bm{\Delta}(\mathbf{T}_{k}, \mathcal{B})=\sum_{|b\rangle\in\mathcal{B}}\left\|\langle b|\mathbf{T}_{k}|b\rangle -\mathbb{I}_{2^{N_A}}\right\|_{1}.     
\end{align}
In general, finding bases that fully satisfy the condition, implying $\bm{\Delta}(\mathbf{T}_{k}, \mathcal{B})=0$, is hard.
Depending upon $\mathcal{B}$, this quantity will display a multitude of behaviors. Also, note that for a single site unitary $u$, the local transformation of $|b\rangle$ to $|b'\rangle=u^{\otimes N_B}|b\rangle$ leaves $\bm{\Delta}(\mathbf{T}_{k}, \mathcal{B})$ invariant. To see the nature of the violation in a generic basis, we numerically examine $\bm{\Delta}(\mathbf{T}_{k}, \mathcal{B})/2^{N_B}$ versus $N_B$ for three different bases, namely, the computational basis, a Haar random product basis, and a Haar random entangling basis, all supported over $B$.

The results are shown in Fig. \ref{fig:viol}. In Fig. \ref{fig:viol}a, the blue curve represents the violation for the computational basis. Clearly, the decay of the violation is exponential and faster than the other cases considered. The orange curve represents the case of local random product basis. This can be obtained by applying a tensor product of Haar random local unitaries on the computational basis vectors. As the figure depicts, the violation still decays exponentially but slower than in the case of computational basis. Finally, we consider a random entangling basis by applying global Haar unitaries on the computational basis vectors. The violation still decays exponentially but slower than the previous two. In Fig. \ref{fig:viol}c, we plot the violation for each basis vector of the above bases considered while keeping $N_B$ fixed. We see that, except for a few vectors, the violation stays concentrated near a value of order $O(1)$. In the computational basis, the only vectors $|0\rangle^{\otimes N_B}$ and $|1\rangle^{\otimes N_B}$ show the maximal violation, which are encircled/marked in Fig. \ref{fig:viol}c. We consider the violation is significant if $\bm{\Delta}(\mathbf{T}_{k}, \mathcal{B})/2^{N_B}$ does not decay with $N_B$. If this happens, the projected ensembles may fail to converge to the state designs even in the large $N_B$ limit. Note that, while the exponential decay of the violation may appear generic, there exist bases that show nearly constant violation as $N_B$ increases, which are depicted in Fig. \ref{fig:haar_vs_tran}b and \ref{fig:haar_vs_tran}d.

To explore such measurement bases with significant violations, we consider the following.
\begin{align}\label{tria}
\bm{\Delta}(\mathbf{T}_{k}, \mathcal{B})&=\sum_{|b\rangle\in\mathcal{B}}\left\|\langle b|\sum_{j=0}^{N-1}e^{2\pi ijk/N}T^j|b\rangle -\mathbb{I}_{2^{N_A}}\right\|_{1}\nonumber\\     
&=\sum_{|b\rangle\in\mathcal{B}} \left\| \sum_{j=1}^{N-1}e^{2\pi ijk/N} \langle b| T^j |b\rangle \right\|_{1} \nonumber\\
&=\sum_{|b\rangle\in\mathcal{B}}\left(\left\|e^{2\pi ir/N} \langle b|T^{r}|b\rangle \right.\right.\nonumber\\
&\left.\left.+\sum_{j\neq r}e^{2\pi ijk/N} \langle b| T^j |b\rangle \right\|_{1}\right).
\end{align}
In the second line, we used the fact $\langle b|\mathbb{I}_{2^N}|b\rangle =\mathbb{I}_{2^{N_A}}$ and subtracted it from $e^{2\pi i 0/N}\langle b|T^0|b\rangle$. 
In the third equality, a term corresponding to an arbitrary integer $r$ in the summation, where $1 \leq r \leq N-1$, has been isolated from the remaining terms. This enables the analysis of violations with respect to each element of the translation group, facilitating the identification of the violating bases. To illustrate this, we consider the specific case where $r=1$:
\begin{eqnarray}
\langle b|T|b\rangle&=&\langle b|S_{1, 2}S_{2, 3}\cdots S_{N-1, N}|b\rangle\nonumber\\    
&=&\left(S_{1, 2}\cdots S_{N_A-1, N_A}\right)\nonumber\\
&&\hspace{2cm}\langle b|S_{N_A, N_A+1}\cdots S_{N-1, N}|b\rangle\nonumber\\
&=&\int_{u}d\mu(u)\left(T_Au^{\dagger}_{N_A}\right)\langle b|u_{N_A+1}T_B|b\rangle, 
\end{eqnarray}
where $S_{i,i+1}$ denotes the swap operator between $i$ and $i+1$ sites, $T_A$ and $T_B$ are translation operators locally supported over the subsystems $A$ and $B$. In the third equality, $S_{N_A, N_A+1}$ is replaced by the unitary Haar integral expression of the swap operator --- $S_{N_A, N_A+1}=\int_{u}d\mu(u)(u^{\dagger}_{N_A}\otimes u_{N_A+1})$, where $d\mu(u)$ represents the invariant Haar measure over the unitary group $U(2)$ \cite{zhang2014matrix}. Then, we heuristically argue that the measurements in the eigenbasis of the operator $u_{N_A+1}T_B$ for an arbitrary $u$ would lead to $\bm{\Delta}(\mathbf{T}_{k}, \mathcal{B})\sim O(2^{N})$. Consequently, the average violation $\bm{\Delta}(\mathbf{T}_{k}, \mathcal{B})/2^{N_B}$ stays nearly a constant of order $O(2^{N_A})$ even in the large $N_B$ limit. We illustrate this by considering the eigenbases of the operator $u_{N_A+1}T_B$ in Fig. \ref{fig:viol}b and \ref{fig:viol}d for two cases of $u_{N_B}$, namely, $u_{N_B}=\mathbb{I}_{2}$ and a Haar random $u_{N_B}$. Infact, we considered the same measurement bases in Figs. \ref{fig:haar_vs_tran}c and \ref{fig:haar_vs_tran}d and observed that the projected ensembles do not converge to the higher-order state designs.
It is interesting to notice that the measurement bases that do not respect the translation symmetry can also hinder the design formation [see Fig. \ref{fig:haar_vs_tran}d]. In Appendix \ref{vior2}, we elaborate this aspect further by considering $r=2$ in Eq. (\ref{tria}).

\section{Generalization to other symmetries}
\label{genn}
In the preceding sections, we examined the emergence of state designs from the translation symmetric generator states. In this section, we extend these findings to other symmetries, specifically considering $Z_{2}$ and reflection symmetries as representative examples.

\subsection{$Z_2$-symmetry}
\label{z2s}
The group associated with $Z_2$-symmetry consists of two elements $\{\mathbb{I}_{2^N}, \Sigma\}$, where $\Sigma=\otimes_{i=1}^{N}\sigma^{x}_{i}$. If a system is $Z_2$-symmetric, its Hamiltonian will be invariant under the action of $\Sigma$, i.e., $\Sigma  H\Sigma=H$. On the other hand, a quantum state $|\psi\rangle$ is considered $Z_2$-symmetric if it is an eigenstate of the operator $\Sigma$ with an eigenvalue $\pm 1$. Here, similar to the case of translation symmetry, we first construct an ensemble of $Z_2$-symmetric states by projecting the Haar random states onto a $Z_2$-symmetric subspace. We then follow the analysis of state designs using the projected ensemble framework.

Given a Haar random state $|\psi\rangle\in\mathcal{E}_{\text{Haar}}$ that has support over $N$-sites, then
\begin{eqnarray}
|\psi\rangle\rightarrow |\phi\rangle=\dfrac{1}{\mathcal{N}}\mathbf{Z}_{k}|\psi\rangle
\end{eqnarray}
is a $Z_2$-symmetric state with an eigenvalue $(-1)^k$ with $k\in\{0, 1\}$, 
where $\mathbf{Z}_k=\mathbb{I}+(-1)^k\Sigma\quad\text{and }\mathcal{N}=\sqrt{\langle\psi|\mathbf{Z}^{\dagger}_{k}\mathbf{Z}_{k}|\psi\rangle}=\sqrt{2\langle\psi|\mathbf{Z}_{k}|\psi\rangle}$. Introducing this symmetry reduces the randomness of the Haar random state by a factor of $2$. Here, the number of independent complex parameters needed to describe the state scales like $\sim O(2^{N-1})$, which can be contrasted with the $O(2^{N})$ variables required for Haar random states. Using similar techniques employed for the T-invariant states, the moments of $Z_2$-symmetric ensembles can be evaluated as
\begin{eqnarray}
 \mathbb{E}_{\phi\in\mathcal{E}^{k}_{Z_{2}}}\left[\left[ |\phi\rangle\langle\phi| \right]^{\otimes t}\right] = \dfrac{\mathbf{Z}^{\otimes t}_{k}\bm{\Pi}_{t}}{\Tr\left(\mathbf{Z}^{\otimes t}_{k}\bm{\Pi}_{t}\right)}.
\end{eqnarray}

The on-site nature of the $\mathbb{Z}_{2}$-symmetry allows us to write closed-form expressions for the moments of the projected ensembles irrespective of the choice of the measurement basis. Through detailed analytical derivation [see Appendix \ref{app-moments}], we show that the $t$-th order moment of the projected ensembles averaged over initial generator states takes the following form:
\begin{eqnarray}\label{z2moment}
\mathcal{M}^{t}_{\mathbb{Z}_{2}}= \dfrac{1}{\mathcal{N}}\sum_{|b\rangle\in\mathcal{B}}\langle b|\mathbf{Z}_{k}|b\rangle^{\otimes t}\bm{\Pi}^{A}_{t},    
\end{eqnarray}
where $\mathcal{N}$ denotes the normalization constant and is given by $\Tr\left( \sum_{|b\rangle\in\mathcal{B}}\langle b|\mathbf{Z}_{k}|b\rangle^{\otimes t}\bm{\Pi}^{A}_{t} \right)$. Whenever $\langle b|\mathbf{Z}_{2}|b\rangle =\mathbb{I}_{2^{N_A}}$ for all $|b\rangle\in \mathcal{B}$, as required by the sufficient condition, right-hand side of the Eq. (\ref{z2moment}) equates to the $t$-th order Haar moment. In the following, we examine the projected ensembles for a few different choices of measurement bases.

To analyze the projected ensembles, we first fix the measurements on $N_B$ sites in the computational basis, i.e., $\{|b\rangle\langle b|\}$ for all $|b\rangle\in\{0, 1\}^{N_B}$. To get approximate state designs, it is sufficient to have $\langle b|\mathbf{Z}_{k}|b\rangle\approx\mathbb{I}_{2^{N_A}}$ for a sufficiently large number of $b\in\{0, 1\}^{N_B}$. In the case of $Z_2$-symmetry, this is indeed satisfied as we have $\langle b|\mathbf{Z}_{k}|b\rangle=\mathbb{I}_{2^{N_A}}+(-1)^{k}\langle b|\Sigma|b\rangle$, where the second term can be simplified as 
\begin{equation}
\langle b|\Sigma|b\rangle =\left(\otimes_{j=1}^{N_A}\sigma^{x}_{j}\right)\left(\langle b_1|\sigma^x|b_1\rangle ...\langle b_{N_B}|\sigma^x|b_{N_B}\rangle\right).     
\end{equation}
The terms within the brackets are the diagonal elements of $\sigma^x$ operator, which are zeros in the computational basis, implying that $\langle b|\Sigma|b\rangle=0$. Therefore, in this case, the sufficient condition is exactly satisfied. Consequently, the projected ensembles converge to the state designs for large $N_B$. The numerical results for the average trace distance are shown in Fig. \ref{fig:ZZ2}. We notice that the results nearly coincide with the case of Haar random generator states [see Fig. \ref{fig:ZZ2}a and \ref{fig:ZZ2}b]. 

\begin{figure}[ht!]
\includegraphics[scale=0.35]{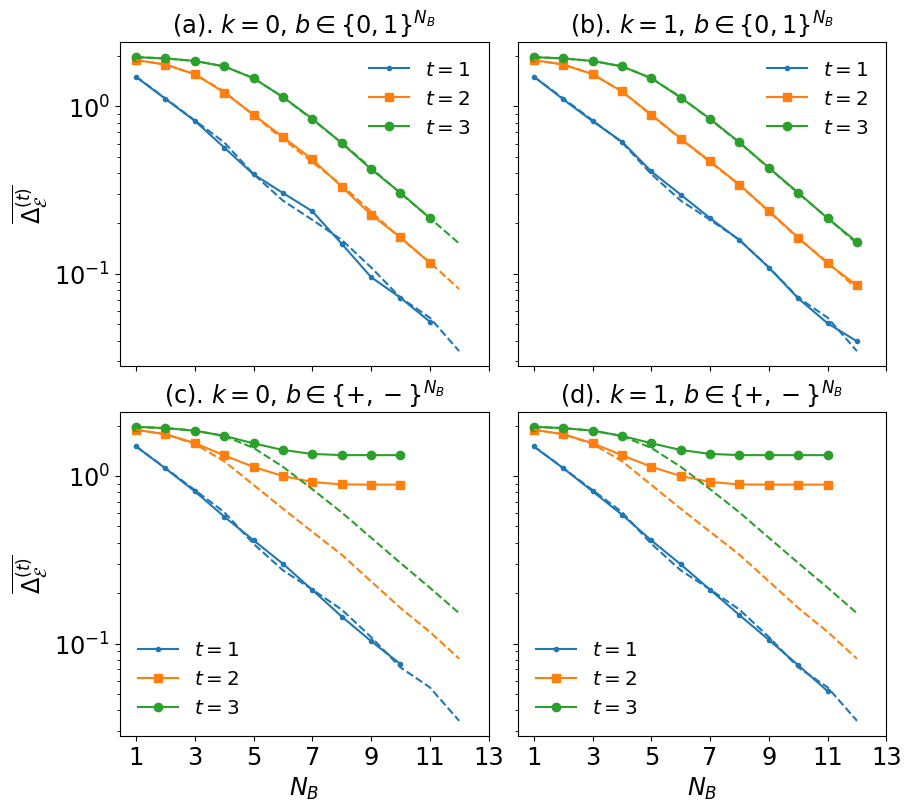}
\caption{\label{fig:ZZ2}  Average trace distance ($\overline{\Delta^{(t)}}$) between the moments of the projected ensembles and the moments of the Haar random states supported over $N_A$ sites, plotted against $N_B$. We fix $N_A=3$. The initial states are chosen uniformly at random from the ensemble of $Z_2$-symmetric quantum states. The averages are computed over $10$ samples of the initial generator states. In (a) and (b), the measurements are performed in the computational ($\sigma^{z}$) basis --- $\{|b\rangle\langle b|\}$ for all $b\in\{0, 1\}^{N_B}$. While the states chosen in (a) have the eigenvalue $1$, the other panel is plotted for the states with eigenvalue $-1$. We repeat the same calculation in the panels (c) and (d) with the measurement basis given by $\{|b\rangle\langle b|\}$ for all $b\in\{+, -\}^{N_B}$, where $|+\rangle$ and $|-\rangle$ represent the eigenstates of $\sigma^x$ and are connected to $|0\rangle$ and $|1\rangle$ through the Hadamard transform. 
} 
\end{figure}

However, if the measurements are performed in $\sigma^x$ basis, given by $\{|b\rangle\langle b|\}$ for all $b\in\{+, -\}^{N_B}$, where $|+\rangle=(|0\rangle + |1\rangle)/\sqrt{2}$ and $|-\rangle=(|0\rangle - |1\rangle)/\sqrt{2}$, then $\langle \pm|\sigma^x|\pm\rangle=\pm 1$. As a result, the projected ensembles deviate significantly from the quantum state designs. The violation from the sufficient condition in this case, as quantified by $\bm{\Delta}(\mathbf{Z}_{k}, \mathcal{B})/2^{N_B}$, remains a constant for any $N_B$: 
\begin{align}
\dfrac{\bm{\Delta}(\mathbf{Z}_{k}, \mathcal{B})}{2^{N_B}} &=\dfrac{1}{2^{N_B}}\sum_{b\in\{+, -\}^{N_B}}\left\| \langle b|\mathbf{Z}_{k}|b\rangle -\mathbb{I}_{2^{N_A}} \right\|_{1}\nonumber\\
&=\dfrac{1}{2^{N_B}}\sum_{b\in\{+, -\}^{N_B}}\nonumber\\
&\hspace{1.2cm}\left\| (-1)^{k+\sum_{i=1}^{N_B}\text{sgn}(b_i)}\otimes_{j=1}^{N_A} \sigma^{x}_{j} \right\|\nonumber\\
&=2^{N_A}.
\end{align}
We demonstrate the numerical results of the average distance in Figs. \ref{fig:ZZ2}c and \ref{fig:ZZ2}d. In contrast to the computational basis measurements, here, only the first moment coincides with the case of Haar random generator states, while the higher moments appear to saturate to a finite value of $\overline{\Delta^{(t)}}$. {Interestingly, in this case, the average $t$-th order moment of the projected ensembles as obtained in Eq. (\ref{z2moment}) admits the following simple form [see Appendix \ref{app-moments} for details]:
\begin{eqnarray}
\mathcal{{M}}^{t}_{\mathbb{Z}_{2}}=\dfrac{1}{\mathcal{N}}\left( \mathbf{Z}^{\otimes t}_{0, N_A} + \mathbf{Z}^{\otimes t}_{1, N_A}\right) \mathbf{\Pi}^{t}_{A}    
\end{eqnarray}
with an appropriate normalizing constant $\mathcal{N}$. Here, $\mathbf{Z}_{0(1), N_A}$ denotes the $\mathbb{Z}_{2}$-symmetric subspace projector with corresponding charge when the system size is $N_A$. 
}

We now consider a case where we perform the local measurements on $N_B$-th site in $\sigma^z$ basis while keeping $\sigma^x$ measurement basis for the remaining $N_B-1$ sites. We represent the resultant basis with $b'\in\{+, 1\}^{N-1}\times \{0, 1\}$. Then,
\begin{eqnarray}
\langle b'|\mathbf{Z}_k|b'\rangle=\mathbb{I}_{2^{N_A}}+(-1)^{k}\langle b'|\Sigma|b'\rangle,     
\end{eqnarray}
where the second term vanishes as $\langle b'_{N_B}|\sigma^{x}_{N_B}|b'_{N_B}\rangle =0$. Therefore, we have $\langle b'|\Sigma|b'\rangle =0$ for all $|b'\rangle$, implying the required condition for the convergence of projected ensembles to the quantum state designs.
The corresponding results for the average trace distances are shown in Fig. \ref{fig:ZZ2plustrans}a. In Fig. \ref{fig:ZZ2plustrans}b, we replace the $\sigma^x$ basis on $N_B$-th site with a local Haar random basis. We find no significant differences between \ref{fig:ZZ2plustrans}a and \ref{fig:ZZ2plustrans}b within the range of $N_B$ considered for the numerical simulations.  
This demonstrates that a mild modification of the measurement basis can retrieve the state design if the sufficient condition is satisfied. So, the sufficient condition allows us to infer suitable measurement bases for obtaining higher-order state designs, which might not be apparent otherwise. Moreover, by gradually switching the local measurement basis (over $N_B$-th site) from $\sigma^x$ to $\sigma^z$, we can observe a transition in the randomness of the projected ensembles as characterized by $\overline{\Delta^{(t)}}$. In particular, we can choose the eigenbasis of $\alpha \sigma^z+(1-\alpha)\sigma^x$ for the measurements on $N_B$-th site. We observe that as $\alpha$ varies, $\overline{\Delta^{(t)}}$ undergoes a transition from a system-size independent constant value to a value that is sensitive to the system size. For more details, we refer to the Appendix. \ref{Z2transition}, where we examine the violation of the sufficient condition as a function of $\alpha$.

For completeness, we also demonstrate the emergence of state designs from the generator states that respect both translation and $Z_2$-symmetry. Since $T$ and $\Sigma$ commute, we can easily construct the states that respect both the symmetries by applying corresponding projectors consecutively on an initial state. Let $|\psi\rangle$ denote a Haar random state, then $|\phi\rangle =\mathbf{T}_{k_1}\mathbf{Z}_{k_2}|\psi\rangle /\sqrt{\mathcal{N}}$, where $\mathcal{N}=\sqrt{2N\left(\mathbf{T}_{k_1}\mathbf{Z}_{k_2}\right)}$, is a random vector that is simultaneously an eigenvector of both $T$ and $\Sigma$ with respective charges $k_1$ and $k_2$. Then, the condition to get quantum state designs from the projected ensembles would be $\langle b|\mathbf{T}_{k_1}\mathbf{Z}_{k_2}|b\rangle =0$ for all $|b\rangle\in\mathcal{B}$. The projective measurements in the standard computational basis mildly violate this condition, which results in the convergence towards state designs. The numerical results are shown in \ref{fig:ZZ2plustrans}c and \ref{fig:ZZ2plustrans}d. While the former is plotted by taking the computational basis measurements, the latter represents the results for the local measurements in $\sigma^x$ basis, i.e., $b\in\{+, -\}^{N_B}$.

\begin{figure}[ht!]
\includegraphics[scale=0.35]{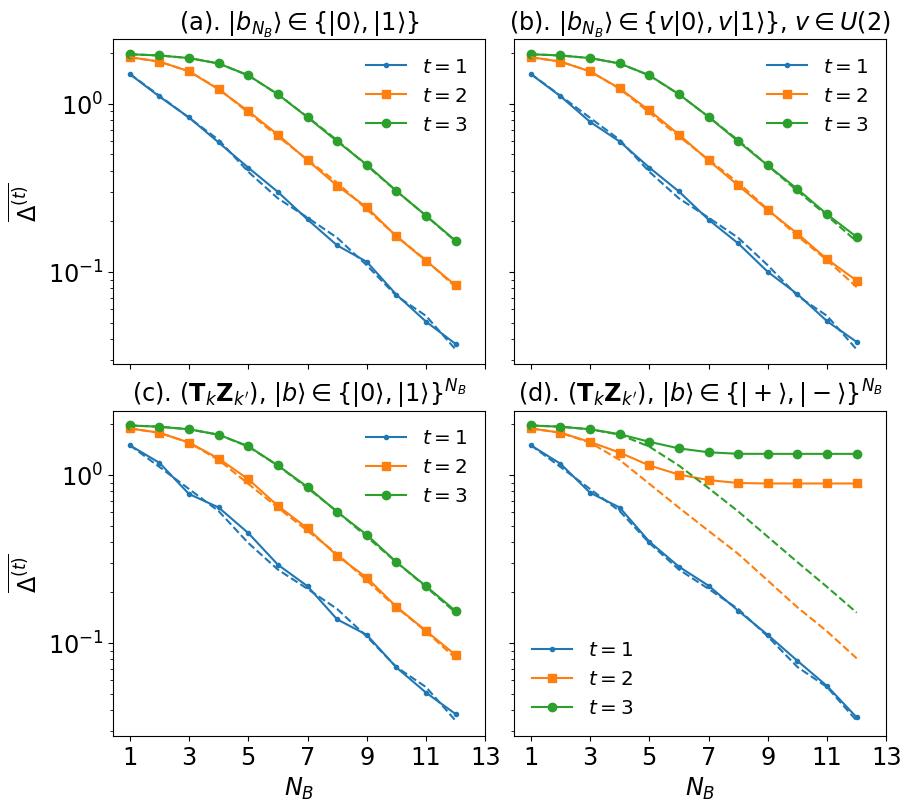}
\caption{\label{fig:ZZ2plustrans} The figure illustrates $\overline{\Delta^{(t)}}$ versus $N_B$. We fix $N_A=3$. In (a), the measurements are performed in the eigenbasis of the operator $\sigma^{x}\otimes\cdots\otimes\sigma^x\otimes \sigma^z$, where the tensor product of $\sigma^x$ operators have support over $N_B-1$ sites and $\sigma^z$ is supported over $N_B$-th site. In (b), we replace the eigenbasis of $\sigma^z$ on $N_B$-th site with an eigenbasis of single site Haar random unitary. Here, we fix the charge $k=0$. In (c) and (d), we take the random generator states that are simultaneous eigenstates of both $T$ and $\Sigma$. While the measurement basis in $c$ is the computational basis, the eigenbasis of $\sigma^x$ is considered for measurements in (d). Here, the charges are fixed at $k=k'=0$. In all the above panels, the averages are computed over $10$ realizations of the initial generator states.} 
\end{figure}

\subsection{Reflection symmetry}
Here, we employ the projected ensemble framework for the generator states having reflection or mirror symmetry. In a system exhibiting reflection symmetry, the Hamiltonian remains invariant under swapping of mirrored sites around the center. Let $R$ denote the reflection operation. Then $R$ generates a cyclic group of two elements, namely, the identity $\mathbb{I}$ and $R$ itself. In an $N$-qubit system, $R$ is defined as 
\begin{eqnarray}
R=
\begin{cases}
    S_{1, N}S_{2, N-1}....S_{N/2, N/2+1}      & \hspace{-0.4cm}\quad \textrm{if } N\textrm{ is even} \\
    S_{1, N}S_{2, N-1}....S_{(N-1)/2, (N+3)/2} & \hspace{-0.4cm}\quad \textrm{if } N\textrm{ is odd.}
\end{cases}\nonumber\\
\end{eqnarray}
The reflection operator has $\{-1, 1\}$ as its eigenvalues. Hence, the total Hilbert space admits a decomposition into two invariant sectors. Then, the Hermitian operators $\mathbf{R}_{\pm}=\mathbb{I}\pm R$ project arbitrary states onto the respective subspaces. Engineering state designs from these generator states would require $\langle b|\mathbf{R}_{\pm}|b\rangle =\mathbb{I}_{2^{N_A}}$ for all $|b\rangle\in\mathcal{B}$. 

We consider a product basis $\mathcal{B}\equiv\{u|0\rangle, u|1\rangle\}^{\otimes N_B}$ for the measurements, where $u$ is an arbitrary unitary operator. For some $|b\rangle\in\mathcal{B}$, we have $\langle b|\mathbf{R}_{\pm}|b\rangle =\mathbb{I}_{2^{N_A}}\pm \langle b|R|b\rangle$. Since $\mathcal{B}$ is assumed to be a local product basis, we can write $|b\rangle =|b_{N_A+1}b_{N_A+2}...b_{N}\rangle$. To be explicit in calculating $\langle b|R|b\rangle$, let us consider $N_A=3$ and  $N_B=N-3\geq N_A$. Then,  
\begin{align}
 \langle b|R|b\rangle= & \underbrace{|b_{N}b_{N-1}b_{N-2}\rangle\langle b_{N}b_{N-1}b_{N-2}|}_{\text{Supported on }A} \nonumber\\ 
& \underbrace{\delta_{b_4, b_{N-3}}  \delta_{b_5, b_{N-4}}...\delta_{b_{(N-1)/2}, b_{(N+3)/2}}}_{\text{palindrome condition}}.
\end{align}
Thus, $\langle b|R|b\rangle$ remains a non-zero operator only when the palindrome condition on the first $N_B-N_A$ bits of the string-$b$ is satisfied. If $N_B-N_A$ is even (odd), then we have a total of $n_{\text{even}}=2^{(N_B-N_A)/2}$ ($n_{\text{odd}}=2^{(N_B-N_A+1)/2}$) distinct palindromes. Then, the total number of violations of the condition for the given measurement basis will be $n_{\text{even}(\text{odd})}2^{N_A}$. For even-$N$, this number is exactly $2^{N/2}$. Moreover, the violation of the sufficient condition in the considered basis is
\begin{eqnarray}
\dfrac{\bm{\Delta}(\mathbf{R}_{k}, \mathcal{B})}{2^{N_B}}&=&\dfrac{1}{2^{N_B}}\sum_{|b\rangle\in\mathcal{B}}\| \langle b|R|b\rangle \|_{1}\nonumber\\
&=&
\begin{cases}
    2^{-N/2},& \text{if } N \text{ is even}\\
    2^{-(N-1)/2},              & \text{otherwise.}
\end{cases}
\end{eqnarray}
Since $\bm{\Delta}(\mathbf{R}_{k}, \mathcal{B})/2^{N_B}$ is exponentially suppressed as $N$ increases, the moments of the projected ensembles converge to the Haar moments.  
We plot the average $\Delta^{(t)}$ versus $N_B$ in Fig. \ref{fig:ref}. To illustrate, we consider the computational basis and random entangling basis for the measurement in  \ref{fig:ref}a and \ref{fig:ref}b, respectively. The later basis states can be obtained by the application of a fixed Haar random unitary supported over $B$ on the computational basis vectors. The results in both panels coincide with the case of Haar random generator states, implying the generation of higher-order state designs. 

\begin{figure}
\includegraphics[scale=0.35]{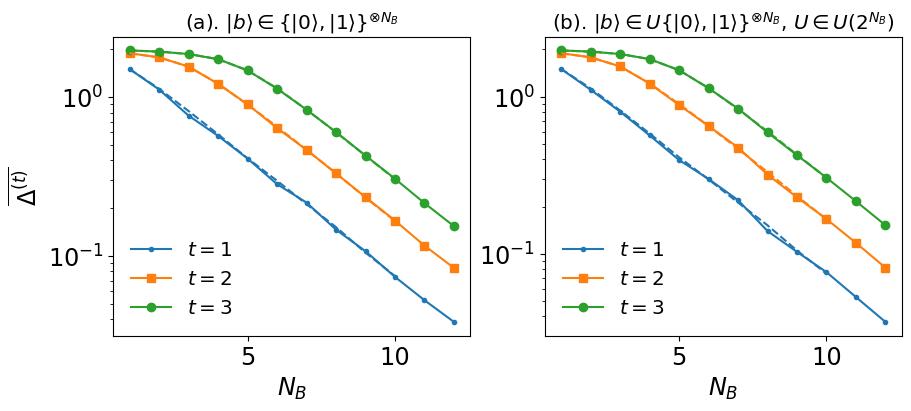}
\caption{\label{fig:ref} Illustration of $\overline{\Delta^{(t)}}$ vs $N_B$ for the random generator states with the reflection symmetry for the first three moments. We fix the charge $k=0$ {and $N_A=3$}. In (a), the computational basis measurements are considered. In (b), the measurements are performed in a random product basis. We find no noticeable differences between these two. Moreover, the decay nearly coincides with the case when the generator states are Haar random. In both (a) and (b) panels, the averages are taken over $10$ realizations of the initial generator states.
} 
\end{figure}

\subsection{Brief comment on the continuous symmetric cases}
So far, we have focused on the random generator quantum states with discrete symmetry group structures and examined the emergence of state designs with respect to various measurement bases. Constructing projectors onto the subspaces that conserve the charge of the symmetry operators lies at the heart of our formalism. We highlight that our formalism equally applies to the cases involving continuous symmetries, provided one can construct projectors onto the charge-conserving subspaces. For instance, the total magnetization conservation $Q=\sum_{j}\sigma^z_{j}$ (or equivalently $U(1)$ symmetry) generates continuous symmetry group. In this case, the projector onto the charge conserving sector with the total charge $s$ can be written in the computational basis as 
\begin{equation}
\mathbf{Q}_{s}= \sum_{f\in\{0, 1\}^{N}/|f|=s}|f\rangle\langle f|,\hspace{-0.2cm}\quad \text{where }|f|=\sum_{j=0}^{N-1}(-1)^{f_j+1}
\end{equation}
Therefore, the sufficient condition for the emergence of state designs from the random eigenstates of $Q$ with the charge $s$ is $\langle b| \mathbf{Q}_{s} |b\rangle =\mathbb{I}_{2^{N_A}}$ for all $|b\rangle\in\mathcal{B}$. In this case, the computational basis measurements are unsuitable for extracting state designs from the projected ensembles. One can then extract the state designs from the projected ensembles by carefully choosing the measurement basis.

\section{Deep thermalization in a chaotic Hamiltonian}
\label{sising}

\begin{figure*}
\includegraphics[scale=0.5]{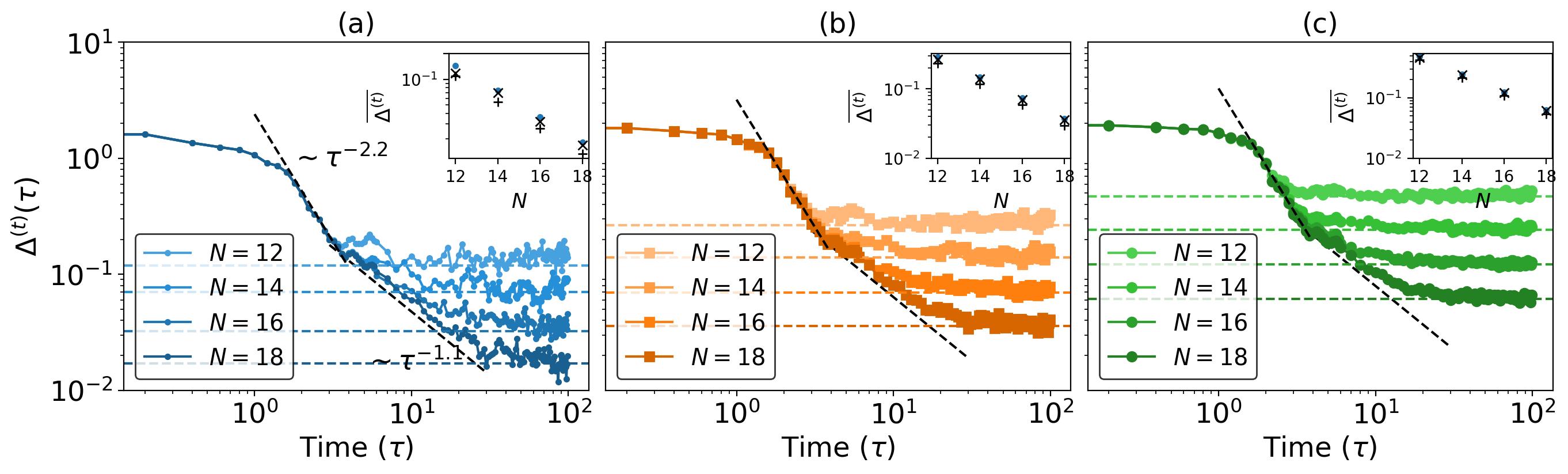}
\caption{\label{fig:ising-phy} Deep thermalization or dynamical generation of state designs (as characterized by $\Delta^{(t)}(\tau)$) of a quantum state $|\psi(\tau)\rangle$ evolved under the dynamics of a chaotic Ising Hamiltonian with periodic boundary conditions. Here, the evolution time is denoted with $\tau$, and the initial state is taken to be $|0\rangle^{\otimes N}$. The results are sequentially shown for the first three moments ($t=1, 2$, and $3$) in the panels along the row. The size of the projected ensembles $N_A$ is fixed at $3$. For each moment, the numerics are carried out for different system sizes varying from $N=12$ to $N=18$. Note the color scheme: 
darker to lighter shading of the colors represents larger to smaller system sizes. 
The dashed horizontal lines in all the panels represent the value attained on average by a typical random state, which is both translation symmetric (with momentum charge $k=0$) and a common eigenstate of the reflection operators about every site, with eigenvalue $1$ [see Appendix \ref{benchmark_latetim}]. From the numerical results, we observe a two-step relaxation of $\Delta^{(t)}(\tau)$ towards the saturation. (insets) {Comparison of the long-time averages of $\Delta^{t}(\tau)$ with the average trace distance when the generator states are random simultaneous eigenvectors of the translation and reflection operators (shown with $\cross$-markers) and random translation symmetric states (as shown with $+$-markers)}. The insets reveal that the long-time averages of $\Delta^{(t)}(\tau)$ coincide well with the former case while the latter case shows slight deviations. See the main text for more details. 
} 
\end{figure*}

In the preceding sections, our analysis focused primarily on obtaining state designs from Haar random states with symmetry. Here, we examine the dynamical generation of the state designs in a tilted field Ising chain with periodic boundary conditions (PBCs). The corresponding Hamiltonian is given by 
\begin{eqnarray}\label{ising}
H=\sum_{i=1}^{N}\sigma^{x}_{i}\sigma^{x}_{i+1}+h_x\sum_{i=1}^{N}\sigma^{x}_{i}+h_y\sum_{i=1}^{N}\sigma^{y}_{i}, 
\end{eqnarray}
where the PBCs correspond to $\sigma^{x, y, z}_{N+i}=\sigma^{x, y, z}_{i}$. {The periodicity, along with the homogeneity of the interactions and the magnetic fields, makes the system translation invariant, i.e., $[H, T]=0$.} In addition, the Hamiltonian is invariant under reflections about all the sites. For the parameters $h_x=(\sqrt{5}+1)/4$ and $h_y=(\sqrt{5}+5)/8$, the system is chaotic, and the ETH has been thoroughly verified in Ref. \cite{kim2014testing}. Furthermore, deep thermalization has also been investigated in this model with open boundary in Ref. \cite{cotler2023emergent} and \cite{bhore2023deep}, which does not respect the translation symmetry. Here, we explore this aspect for the Hamiltonian in Eq. (\ref{ising}) and contrast the results with those of the open boundary condition (OBC). To proceed, we consider a trivial product state $|\psi\rangle=|0\rangle^{\otimes N}$ as the initial state and evolve it under the  Hamiltonian. It is to be noted that the initial state is a common eigenstate of all the reflection operators (with eigenvalue $1$) and the translation operator (with eigenvalue $1$). Since the Hamiltonian commutes with $T$, the final state $|\psi(\tau)\rangle=e^{-i\tau H}|0\rangle^{\otimes N}$ will remain a common eigenstate of the translation and reflection operators with corresponding eigenvalues, where $\tau$ denotes the time of evolution. As this state evolves, we construct and examine its projected ensembles at various times. The computational basis is considered for the projective measurements on the subsystem-$B$. The corresponding numerical results are shown in Fig. \ref{fig:ising-phy}. 

Figures \ref{fig:ising-phy}a-\ref{fig:ising-phy}c demonstrate the decay of $\Delta^{(t)}$ as a function of evolution time ($\tau$) for the first three moments,  $t=1$, $2$, and $3$, respectively. We show this evolution for different system sizes $N$ by considering $N_A=3$ fixed and changing $N_B$.
The evolution of $\Delta^{(t)}(\tau)$ suggests a two-step relaxation towards the saturation. Initially, over a short period, the trace distance $\Delta^{(t)}(\tau)$ scales like $\sim \tau^{-2.2}$ for all three moments. For the largest considered system size, $N=18$, this power law behavior spans across the region $1\lesssim \tau\lesssim 4$. Moreover, this time scale appears to grow with $N_B$, which can be read off from the plots. Note that in Ref. \cite{cotler2023emergent}, the initial decay has been observed to be $\sim \tau^{-1.2}$ for the model with OBC [see also Appendix \ref{symbrok_comparision}]. In contrast, the present case exhibits a nearly doubled exponent of the power law behavior. Interestingly, similar behavior characterized by exponential decay has been observed in Ref. \cite{shrotriya2023nonlocality}, where the authors focused on dual unitary circuits and contrasted the case of PBCs with OBCs {while keeping the fields and interactions homogeneous}. The doubled exponent noticed in the current case can be attributed to the {periodicity as well as the homogeneity of interactions and fields in the Hamiltonian and, in turn, the translation symmetry.} {To elucidate it further, in Appendix \ref{symbrok_comparision}, we contrast these dynamics with the symmetry-broken cases obtained through the modification of boundary conditions and introduction of disorders}. In the present case, the entanglement across the bi-partition $AB$ grows at a rate twice the rate in the case of OBC \cite{mishra2015protocol, pal2018entangling}. To elucidate the role of entanglement growth at initial times, we examine the trace distance for the simplest case $t=1$, i.e., $\Delta^{(1)}(\tau)$. Since the first moment of the projected ensemble is simply the reduced density matrix of $A$, the trace distance can be written as 
\begin{equation}
\Delta^{(1)}(\tau)=\left\|\rho_{A}(\tau)-\dfrac{\mathbb{I}}{2^{N_A}}\right\|_{1}
=\sum_{j=0}^{2^{N_A}-1}\left| \gamma^{2}_{j}(\tau)-\dfrac{1}{2^{N_A}} \right|,    
\end{equation}
where $\rho_{A}(t)=\Tr_{B}(|\psi(t)\rangle\langle\psi(t)|)$ and $\{\gamma_{j}\}$ denote the Schmidt coefficients of $|\psi(t)\rangle$ across the bipartition. Above expression relates $\Delta^{(1)}(\tau)$ to the fluctuations of the Schmidt coefficients around $1/2^{N_A}$, corresponding to the maximally mixed value. At the time $\tau=0$, as the considered initial state is a product state, the only non-zero Schmidt coefficient is $\gamma_0(\tau=0)=1$. Hence, it is fair to say that $\Delta^{(1)}(0)$ is largely dominated by the decay of $\gamma_{0}$ during the early times. This regime usually witnesses inter-subsystem scrambling of the initial state mediated by the entanglement growth.  
Moreover, the initial decay appears across all three panels with the same power law scaling, indicating a similar early-time dependence of $\Delta^{t}(\tau)$ over $\gamma_{0}(\tau)$.

Beyond the initial power-law regime, we observe a crossover to an intermediate-time power-law decay regime with a smaller exponent, which is followed by saturation at a large time. The power-law exponent in the intermediate decay regime depends on the system size $N$, yielding a non-universal characteristic. For $N=18$, we obtain $\Delta^{(t)}\sim t^{-1.2}$, which is to be contrasted with the early-time decay exponent. At the onset of this scaling, the largest Schmidt coefficient $\gamma_0$ becomes comparable with the other $\gamma_j$s. The corresponding timescale is referred to as collision time \cite{vznidarivc2012subsystem}. At the collision time, $\gamma_0$ comes close to $\gamma_1$, the second largest Schmidt-coefficient. Hence, $\gamma_0$ does not solely determine the decay of $\Delta^{(1)}(\tau)$. The two-step relaxation of quantum systems has been recently studied in systems with two or more symmetries and also in quantum circuit models \cite{bensa2022two, vznidarivc2023two}. Finally, we benchmark the late time saturation values for each $N$ using the random matrix theory (RMT) predictions for the appropriate ensembles [see Appendix \ref{benchmark_latetim}]. We do this by plotting horizontal (dashed) lines corresponding to the RMT values. We notice that the saturation matches well with the corresponding RMT predictions. The same is also illustrated in the insets of Fig. \ref{fig:ising-phy}, where the long-time averages of $\Delta^{t}(\tau)$ for different system sizes are denoted with dots. Whereas the RMT values are shown with the marker-$\cross$. It is evident from the insets that the saturation values and the RMT values nearly coincide. On the other hand, it is interesting to note for the random translation symmetric states with no other symmetries present, the average trace distance $\overline{\Delta^{(t)}}$ deviates slightly from the former case. However, in the case of higher-order moments, these differences appear to become smaller. We intuitively expect the two-step relaxation observed in the present case to arise from its two competing features: initial faster decay and late time saturation above random matrix prediction.

\section{Summary and Discussion}\label{discussion}
In summary, we have investigated the role of symmetries on the choice of measurement basis for quantum state designs within the projected ensemble framework. By employing the tools from Lie groups and measure theory, we have evaluated the higher-order moments of the symmetry-restricted ensembles. Using these, we have derived a sufficient condition on the measurement basis for the emergence of higher-order state designs. The condition reads as follows: Given an arbitrary measurement basis $\mathcal{B}\equiv \{|b\rangle\}$ over a subsystem-$B$, for a typical $Q$-symmetric state $|\psi_{AB}\rangle\in\mathcal{E}^{k}_{\text{Q}}$ with a charge $k$, $\langle b|\mathbf{Q}_{k}|b\rangle=\mathbb{I}_{2^{N_A}}$ for all $|b\rangle\in\mathcal{B}$ implies that the projected ensembles approximate higher-order state designs. Moreover, the approximation improves exponentially with $N_B$, the bath size. While the condition is sufficient for the emergence of state designs, the necessity of it remains an open question. We demonstrate its versatility by considering measurement bases violating the condition mildly. Our analysis further suggests that a significant violation of the condition likely prevents the convergence of projected ensembles to the designs even in the limit of large $N_B$. To elucidate it, we have quantified the extent to which a basis violates the sufficient condition using the quantity $\sum_{|b\rangle\in\mathcal{B}}\|\langle b|\mathbf{Q}_{k}|b\rangle -\mathbb{I}_{2^{N_A}}\|_{1}/2^{N_B}$. This quantity allows us to identify the bases that violate the condition significantly. We have shown that the measurements in these bases result in a finite value for the trace distance $\Delta^{(t)}$ even when $N_B$ is large. Surprisingly, these include bases that do not adhere to the symmetry in the generator states.

To begin with, we have chosen random T-invariant states as the generator states. In constructing these states, we projected the Haar random states onto the momentum-conserving subspaces to reconcile both randomness and symmetry. This allows one to construct distinct ensembles of T-invariant states, each with a different momentum, and the states in the ensembles are uniformly distributed. Given a suitable measurement basis, Levy's lemma then ensures that the projected ensemble of a typical T-invariant state well approximates a state design. Equipped with this argument, we have numerically verified the emergence of designs for different measurement bases. These bases include the standard computational ($\sigma^z$) basis and the eigenbasis of $T_B$. While the former nearly satisfies the sufficient condition, the latter violates it significantly. Accordingly, the trace distance $\Delta^{(t)}$ decays exponentially with $N_B$ for the computational basis. Whereas, for the eigenbasis of $T_B$, $\Delta^{(t)}$ converges to a non-zero value. To further contextualize our results in a more physical setting, we have focused on deep thermalization in a tilted field Ising chain with PBCs, respecting the translation symmetry. The results indicate that the decay of the trace distance with time occurs in two steps. The initial decay is observed to be twice the rate of the case of the same model with OBCs. In the intermediate time, the decay trend is a system-dependent power law. Whereas, in a long time, the trace distance saturates to a value slightly larger than RMT prediction, a reminiscence of other symmetries.

Due to the generality of our formalism, the results can be extended to other discrete symmetries and are expected to hold for continuous symmetries as well. In particular, generalization to other cyclic groups is straightforward. To illustrate this, we have examined the projected ensembles from the generator states with $Z_2$ and reflection symmetries. A crucial implication of our results is that the sufficient condition plays a pivotal role in identifying appropriate measurement bases, even when their suitability is not immediately apparent.

If one considers two or more non-commuting symmetries, they do not share common eigenstates. In such systems, the equilibrium states have been shown to approximate non-abelian thermal states \cite{majidy2023noncommuting, yunger2016microcanonical}. These states have been experimentally realized recently in Ref. \cite{kranzl2023experimental}. Hence, an extensive study of deep thermalization and emergent state designs in these systems is a topic of our immediate future investigation. Additionally, measurement-induced phase transitions (MIPTs) occur due to an interplay between the measurements and the dynamics in many-body chaotic systems \cite{skinner2019measurement}. Our results can offer insights into the mechanism of the MIPTs whenever the dynamics and the measurements are chosen to respect symmetries \cite{majidy2023critical}.

\begin{acknowledgments}
We gratefully acknowledge useful discussions with Vaibhav Madhok, Arul Lakshminarayan, Philipp Hauke, and N. Ramadas.  We thank Andrea Legramandi for reading the manuscript and providing useful suggestions. N.D.V.\ acknowledges funding from the Department of Science and Technology, Govt of India, under Grant No. DST/ICPS/QusT/Theme-3/2019/Q69, and partial support by a grant from Mphasis to the Centre for Quantum Information, Communication, and Computing (CQuICC) at
IIT Madras.
S.B.\ acknowledges funding from the European Union under NextGenerationEU Prot. n. 2022ATM8FY (CUP: E53D23002240006), European Research Council (ERC) under the European Union’s Horizon 2020 research and innovation programme (grant agreement No 804305), Provincia Autonoma di Trento, Q$@$TN, the joint lab between University of Trento, FBK-Fondazione Bruno Kessler, INFN-National Institute for Nuclear Physics and CNR-National Research Council.
Views and opinions expressed are however those of the author(s) only and do not necessarily reflect those of the European Union or European Commission. Neither the European Union nor the granting authority can be held responsible for them.
S.B.\ acknowledges CINECA for the use of HPC resources under ISCRA-C projects ISSYK-2 (HP10CP8XXF) and DISYK (HP10CGNZG9).

\end{acknowledgments}

\onecolumngrid
\appendix

\section{Details on the construction of random T-invariant unitaries}
\label{polar}
The QR decomposition is traditionally used to generate Haar random unitary operators from the initial random Gaussian matrices. However, QR decomposition can not produce uniformly distributed unitaries from the subgroups such as $U_{\text{TI}}(d^N)$ as the decomposition does not preserve the symmetries of the initial operator. Here, we use polar decomposition as an alternative to the QR decomposition to generate random unitary operators. For a given initial operator $Z$, the polar decomposition is given by $Z=UP$, where $P$ is a positive semi-definite operator, $P=\sqrt{Z^{\dagger}Z}$. If $Z$ is a full rank matrix, $U=Z(Z^{\dagger}Z)^{-1/2}$ can be  uniquely computed. If $Z$ is a complex Gaussian matrix with mean $\mu=0$ and standard deviation $\sigma=1$, the polar decomposition will yield the ensemble of unitaries whose moments match those of the Haar random unitaries. To see this, consider $A$, an arbitrary operator acting on $t$-replicas of the same Hilbert space $\mathcal{H}^{d}$. Then, we have 
\begin{eqnarray}
\langle U^{\dagger\otimes t}AU^{\otimes t}\rangle &=& \int_{Z}d\mu(Z) \left(Z(Z^{\dagger}Z)^{-1/2}\right)^{\dagger\otimes t}A\left(Z(Z^{\dagger}Z)^{-1/2}\right)^{\otimes t}\nonumber\\
&=&\int_{Z}d\mu(Z) \left((Z^{\dagger}Z)^{-1/2} Z^{\dagger}\right)^{\otimes t}A\left(Z(Z^{\dagger}Z)^{-1/2}\right)^{\otimes t},  
\end{eqnarray}
where $d\mu(Z)$ denotes the invariant measure over the Ginebre ensemble. 
Since the Ginibre ensemble is unitarily invariant, we replace $Z$ with $VZ$ for some $V\in U(d^N)$ and perform Haar integral over $V$. This action keeps the overall integral in the above equation invariant. 
\begin{eqnarray}
\langle U^{\dagger\otimes t}AU^{\otimes t}\rangle &=&\int_{Z}d\mu(Z)\int_{V\in U(d^N)}d\mu(V) \left((Z^{\dagger}Z)^{-1/2} Z^{\dagger}V^{\dagger}\right)^{\otimes t}A\left(VZ(Z^{\dagger}Z)^{-1/2}\right)^{\otimes t}\nonumber\\
&=&\int_{Z}d\mu(Z)\left((Z^{\dagger}Z)^{-1/2} Z^{\dagger}\right)^{\otimes t}\left(\int_{V\in U(d^N)}d\mu(V)V^{\dagger\otimes t}AV^{\otimes t}\right)\left(Z(Z^{\dagger}Z)^{-1/2}\right)^{\otimes t}. 
\end{eqnarray}
By the Schur-Weyl duality, $\int_{V\in U(d^N)}d\mu(V)V^{\dagger\otimes t}AV^{\otimes t}=\sum_{i=1}^{t!}c_i \pi_i$, where $\{\pi_i\}$s are permutation operators acting on $t$-replicas of the Hilbert space. It then follows that 
\begin{eqnarray}
\langle U^{\dagger\otimes t}AU^{\otimes t}\rangle =\left(\int_{V\in U(d^N)}d\mu(V)V^{\dagger\otimes t}AV^{\otimes t}\right)\left(\int_{Z}d\mu(Z)\left((Z^{\dagger}Z)^{-1/2} Z^{\dagger}\right)^{\otimes t}\left(Z(Z^{\dagger}Z)^{-1/2}\right)^{\otimes t}\right).  
\end{eqnarray}
Since $Z(Z^{\dagger}Z)^{-1/2}=U$ is a unitary operator, the integrand of the second integral becomes the Identity operator. Therefore, 
\begin{eqnarray}
\langle U^{\dagger\otimes t}AU^{\otimes t}\rangle= \int_{V\in U(d^N)}d\mu(V)V^{\dagger\otimes t}AV^{\otimes t}. 
\end{eqnarray}
This equation implies that the moments of the ensemble of unitaries from the polar decomposition are identical to those of the Haar ensemble of unitaries. 

We now show that if the initial operator commutes with an arbitrary unitary operator, then the resulting unitary from the polar decomposition necessarily commutes with the same. For our purpose, we take the commuting unitary to be $T$, the translation operator. Let $Z$ be randomly drawn from the Ginibre ensemble. Then, the operator $Z'=\sum_{j=0}^{N-1}T^{\dagger j}ZT^{j}$ is translation invariant as $T^{\dagger}Z'T=Z'$. Moreover, $P'=\sqrt{Z^{'\dagger}Z'}$ is also $T$-invariant whenever $Z'$ is a full rank matrix. Consequently, the resulting unitary $U'$ commutes with $T$. One can also show that the distribution of $Z'$ is invariant under the action of elements of $U_{\text{TI}}(d^N)$. Therefore, the resulting ensemble of unitaries has the same moments as those of the unitary subgroup $U_{\text{TI}}(d^N)$.

\section{Proof of Result \ref{Tran_designs}}\label{trans_designs}
\begin{proof}
We first note that given a Haar random pure state $|\psi\rangle$, under the map $\mathbb{T}_{k}$, becomes an eigenstate of $T$ with the eigenvalue $e^{-2\pi ik/N}$, i.e., $|\phi\rangle =\mathbf{T}_{k}|\psi\rangle/\sqrt{\langle\psi|\mathbf{T}^{\dagger}_{k}\mathbf{T}_{k}|\psi\rangle}$. This generates an ensemble $\{|\phi\rangle \}$, denoted with $\mathcal{E}^{k}_{\text{TI}}$, when $|\psi\rangle\in\mathcal{E}_{\text{Haar}}$. We are interested in finding the moments associated with $\mathcal{E}^{k}_{\text{TI}}$. For any $k$, the $t$-th moment can be evaluated as follows: 
\begin{eqnarray}
\mathbb{E}_{\phi\in\mathcal{E}^{k}}\left[\left[|\phi\rangle\langle\phi|\right]^{\otimes t}\right]=\int_{\psi\in\mathcal{E}_{\text{Haar}}}d\psi\dfrac{\mathbf{T}^{\otimes t}_{k}\left[|\psi\rangle\langle\psi|\right]^{\otimes t}\mathbf{T}^{\dagger\otimes t}_{k}}{\langle|\psi|\mathbf{T}^{\dagger}_{k}\mathbf{T}_{k}|\psi\rangle^{t}},
\end{eqnarray}
where the integral on the right-hand side is performed over the Haar random pure states. Since the Haar random states can be generated through the action of Haar random unitaries on a fixed fiducial quantum state, we can write
\begin{eqnarray}
\mathbb{E}_{\phi\in\mathcal{E}^{k}}\left[\left[|\phi\rangle\langle\phi|\right]^{\otimes t}\right]=\int_{u\in U(d^N)}d\mu(u)\dfrac{\mathbf{T}^{\otimes t}_{k}\left[u|0\rangle\langle 0|u^{\dagger}\right]^{\otimes t}\mathbf{T}^{\dagger\otimes t}_{k}}{\langle 0|u^{\dagger}\mathbf{T}^{\dagger}_{k}\mathbf{T}_{k}u|0\rangle^{t}}.
\end{eqnarray}
Since the integrand in the above equation has $u$-dependence in both the numerator and the denominator, direct evaluation of the Haar integral is challenging. To circumvent it, let us now consider the following ensemble average:
\begin{eqnarray}\label{A3}
\mathbb{E}\left[\left[|\phi\rangle\langle\phi|\right]^{\otimes t} \langle\psi|\mathbf{T}^{\dagger}_{k}\mathbf{T}_{k}|\psi\rangle^{t}\right]=\int_{u\in U(d^N)}d\mu(u)\dfrac{\mathbf{T}^{\otimes t}_{k}\left[u|0\rangle\langle 0|u^{\dagger}\right]^{\otimes t}\mathbf{T}^{\dagger\otimes t}_{k}}{\langle 0|u^{\dagger}\mathbf{T}^{\dagger}_{k}\mathbf{T}_{k}u|0\rangle^{t}}\langle 0|u^{\dagger}\mathbf{T}^{\dagger}_{k}\mathbf{T}_{k}u|0\rangle^{t}.\end{eqnarray}
By taking advantage of the left and the right invariance of the Haar measure over the unitary group $U(d^N)$, we replace $u$ in the above equation with $vu$, where $v\in U_{\text{TI}}(d^N)\subset U(d^N)$. Under this action, the term $\langle 0|u^{\dagger}\mathbf{T}^{\dagger}_{k}\mathbf{T}_{k}u|0\rangle^{t}$ remains independent of $v$ as $[v, \mathbf{T}_{k}]=0$ for all $v\in U_{\text{TI}}(d^N)$ and $k$.
We then perform the Haar integration over $U_{\text{TI}}(d^N)$, which corresponds to the following:
\begin{align}\label{uni_inv}
\mathbb{E}\left[\left[|\phi\rangle\langle\phi|\right]^{\otimes t} \langle\psi|\mathbf{T}^{\dagger}_{k}\mathbf{T}_{k}|\psi\rangle^{t}\right]
&=\int_{u\in U(d^N)}d\mu(u)\langle 0|u^{\dagger}\mathbf{T}^{\dagger}_{k}\mathbf{T}_{k}u|0\rangle^{t}
\underbrace{\int_{v\in U_{\text{TI}}(d^N)}d\mu_{\text{TI}}(v)\dfrac{\mathbf{T}^{\otimes t}_{k}\left[vu|0\rangle\langle 0|u^{\dagger}v^{\dagger}\right]^{\otimes t}\mathbf{T}^{\dagger\otimes t}_{k}}{\langle 0|u^{\dagger}\mathbf{T}^{\dagger}_{k}\mathbf{T}_{k}u|0\rangle^{t}}}_{=\mathbb{E}_{\phi\in\mathcal{E}^{k}_{\text{TI}}}\left[\left[|\phi\rangle\langle\phi|\right]^{\otimes t}\right]}\nonumber\\
&=\mathbb{E}_{\phi\in\mathcal{E}^{k}_{\text{TI}}}\left[\left[|\phi\rangle\langle\phi|\right]^{\otimes t}\right] \int_{u\in U(d^N)}d\mu(u)\langle 0| u^{\dagger}\mathbf{T}^{\dagger}_{k}\mathbf{T}_{k}u|0\rangle^{t}\nonumber\\
&=\mathbb{E}_{\phi\in\mathcal{E}^{k}_{\text{TI}}}\left[\left[|\phi\rangle\langle\phi|\right]^{\otimes t}\right] \int_{|\psi\rangle\in \mathcal{E}_{\text{Haar}}}d\psi
\langle\psi|\mathbf{T}^{\dagger}_{k}\mathbf{T}_{k}|\psi\rangle^{t}.
\end{align}
where $d\mu_{\text{TI}}$ denotes the Haar measure over the subgroup $U_{\text{TI}}(d^N)$. Since $v$ is uniformly random in $U_{\text{TI}}(d^N)$, $v|\phi\rangle$ is also uniformly random in $\mathcal{E}^{k}_{\text{TI}}$
for any $|\phi\rangle\in\mathcal{E}^{k}_{\text{TI}}$. Therefore, $\mathbb{E}_{|\phi\rangle\in\mathcal{E}^{k}_{\text{TI}}}\left[\left(|\phi\rangle\langle\phi|\right)^{\otimes t}\right]=E_{v\in U_{\text{TI}}(d^N)}\left[\left(v|\phi\rangle\langle\phi|v^{\dagger}\right)^{\otimes t}\right]$. This is substituted in the second equality above. Equation (\ref{uni_inv}) implies that $\left[|\phi\rangle\langle\phi|\right]^{\otimes t}$ and $\langle\psi|\mathbf{T}^{\dagger}_{k}\mathbf{T}_{k}|\psi\rangle^{t}$ are independent random variables. Then, combining Eq. (\ref{A3}) and (\ref{uni_inv}), we get
\begin{eqnarray}
\mathbb{E}_{\phi\in\mathcal{E}^{k}_{\text{TI}}}\left[\left[|\phi\rangle\langle\phi|\right]^{\otimes t}\right]&=&\dfrac{\int_{u\in U(d^N)}d\mu(u)\mathbf{T}^{\otimes t}_{k}\left[u|0\rangle\langle 0|u^{\dagger}\right]^{\otimes t}\mathbf{T}^{\dagger\otimes t}_{k} }{\int_{u\in U(d^N)}d\mu(u)\langle 0| u^{\dagger}\mathbf{T}^{\dagger}_{k}\mathbf{T}_{k}u|0\rangle^{t}} \nonumber\\
&=&\dfrac{\mathbf{T}^{\otimes t}_{k}\bm{\Pi}_{t}}{\Tr\left( \mathbf{T}^{\otimes t}_{k}\bm{\Pi}_{t} \right)},
\end{eqnarray}
implying the result. Our analysis does not require an explicit form for $\alpha^{t}_{k}=\Tr\left( \mathbf{T}^{\otimes t}_{k}\bm{\Pi}_{t} \right)$. Hence, we leave it unchanged.
\end{proof}

\subsection{Alternative proof}
{Here, we provide an alternative proof for the moments of the ensembles of translation invariant states. The proof relies on identifying the commutant of the moment operator. 
\begin{proof}
We intend to compute the moments of random ensembles of translation symmetric states given by
\begin{eqnarray}\label{comtnt}
 \mathbb{E}_{|\phi\rangle\in\mathcal{E}^{k}_{\text{TI}}}\left[\left[ |\phi\rangle\langle\phi| \right]^{\otimes t}\right]&=&\int_{u\in U_{\text{TI}}(2^N)}d\mu(u) \left[ u|\phi\rangle\langle\phi|u^{\dagger} \right]^{\otimes t}.
\end{eqnarray}
The next step is to identify that the $t$-th order commutant of the unitary subgroup $U_{\text{TI}}(2^N)$ is given by the following Cartesian product:
\begin{eqnarray}
 \text{Comm}\left(U_{\text{TI}}(2^N), t\right)=\left\{\{\mathcal{T}^j\}^{N}_{j=1}\right\}^{\otimes t}\cross \{\pi_{l}\}^{t!}_{l=1}   
\end{eqnarray}
where $\mathcal{T}=e^{2\pi ik/N}T$ acts on $\mathcal{H}^{\otimes N}$ and $\pi_{l}$s denote the permutation operators acting on $t$-replicas of $\mathcal{H}^{\otimes N}$. Given the commutant, the right-hand side of Eq. (\ref{comtnt}) can be written as the linear combination of the elements of the commutant. 
\begin{eqnarray}\label{expan}
\int_{u\in U_{\text{TI}}(2^N)}d\mu(u) \left[ u|\phi\rangle\langle\phi|u^{\dagger} \right]^{\otimes t}=\sum_{j_1, j_2, \cdots J_t=1}^{N}\sum_{l=1}^{t!} \alpha_{j_1, j_2, \cdots j_t, l} \left(\mathcal{T}^{j_1}\otimes\mathcal{T}^{j_2}\otimes\cdots \otimes \mathcal{T}^{j_t} \right)\pi_{l}  .
\end{eqnarray}
In the next step, we multiply the above expression with $\left(\mathcal{T}^{j_1}\otimes\mathcal{T}^{j_2}\otimes\cdots \otimes \mathcal{T}^{j_t} \right)\pi_{l}$ for some $j_1, j_2, \cdots, j_t$ and $l$. The term on the left-hand side remains unaffected by this action, i.e., 
\begin{eqnarray}
\left(\mathcal{T}^{j_1}\otimes\mathcal{T}^{j_2}\otimes\cdots \otimes \mathcal{T}^{j_t} \right)\pi_{l}\int_{u\in U_{\text{TI}}(2^N)}d\mu(u) \left[ u|\phi\rangle\langle\phi|u^{\dagger} \right]^{\otimes t}=\int_{u\in U_{\text{TI}}(2^N)}d\mu(u) \left[ u|\phi\rangle\langle\phi|u^{\dagger} \right]^{\otimes t} .   
\end{eqnarray}
Also, note that the right-hand side of Eq. (\ref{expan}) still remains as a linear combination of the elements of the commutant, but the coefficients $\alpha_{j_1, j_2, \cdots , j_{t}, l}$ are now permuted. The resultant expression can be made into a linear equation of the coefficients by taking traces on both sides. It is then straightforward to verify that all the coefficients turn out to be equal. It then follows that 
\begin{eqnarray}
\int_{u\in U_{\text{TI}}(2^N)}d\mu(u) \left[ u|\phi\rangle\langle\phi|u^{\dagger} \right]^{\otimes t}&=&\alpha\sum_{j_1, j_2, \cdots, j_t=1}^{N}\sum_{l=1}^{t!}  \left(\mathcal{T}^{j_1}\otimes\mathcal{T}^{j_2}\otimes\cdots \otimes \mathcal{T}^{j_t} \right)\pi_{l}\nonumber\\
&=&\alpha \left(\sum_{j=1}^{N}e^{2\pi ijk/N} T^j\right)^{\otimes t}\left(\sum_{l=1}^{t!}\pi_{l}  \right)\nonumber\\
&=&\alpha \mathbf{T}^{\otimes t}_{k}\bm{\Pi}_{t}. 
\end{eqnarray}
The normalization constant is given by $\alpha=1/\Tr\left( \mathbf{T}^{\otimes t}_{k}\bm{\Pi}_{t} \right)$. 
\end{proof}}

\section{Partial trace of $T^j$}
\label{ptrace}
This appendix shows that the particle trace of $T^j$ results in some permutation operator whenever $N_A\geq \gcd(N, j)$. We first consider $j=1$. Then, $\Tr_{B}(T)$ is still a translation operator, acting on the subsystem-$A$ as shown in the following:
\begin{eqnarray}
\Tr_{B}\left(T\right)&=&\sum_{b\in\{0, 1\}^{N_B}}\langle b|T|b\rangle    \nonumber\\
&=&\sum_{b\in\{0, 1\}^{N_B}}\sum_{a\in\{0, 1\}^{N_A}}\sum_{a'\in\{0, 1\}^{N_A}}\langle a_1...a_{N_A}b_1...b_{N_B}|T|a'_1...a'_{N_A}b_1...b_{N_B}\rangle |a_1...a_{N_A}\rangle\langle a'_1...a'_{N_A}|\nonumber\\
&=&\sum_{b\in\{0, 1\}^{N_B}}\sum_{a\in\{0, 1\}^{N_A}}\sum_{a'\in\{0, 1\}^{N_A}}\left(\langle a_1...a_{N_A}b_1...b_{N_B}|b_{N_B}a'_1...a'_{N_A}b_1...b_{N_B-1}\rangle\right) |a_1...a_{N_A}\rangle\langle a'_1...a'_{N_A}|\nonumber\\
&=&\sum_{b\in\{0, 1\}^{N_B}}\sum_{a\in\{0, 1\}^{N_A}}\sum_{a'\in\{0, 1\}^{N_A}} \delta_{a_1, b_{N_B}}\delta_{a_2, a'_1}\delta_{a_3, a'_2}...\delta_{a_{N_A}, a'_{N_A-1}}\delta_{b_1, a'_{N_A}}\delta_{b_2, b_1}...\delta_{b_{N_B}, b_{N_B-1}}|a_1...a_{N_A}\rangle\langle a'_1...a'_{N_A}|\nonumber\\
&=&\sum_{a\in\{0, 1\}^{N_A}}|a_1...a_{N_A}\rangle\langle a_2...a_{N_A}a_1|\nonumber\\
&=&T_A.
\end{eqnarray}
In the fourth equality, on the right-hand side, the product of Kronecker deltas results in the following chains of equalities:
\begin{align}
a_1=b_{N_B}=b_{N_B-1}=...=b_{2}=b_{1}=a'_{N_A} \quad\text{and}\quad a_{i}=a'_{i-1}\quad\text{for all }N\leq i\leq 2.
\end{align}
The first chain contains equalities of all the bits of the $b$-strings. Therefore, the summation over $b\in\{0, 1\}^{N_B}$ disappears. Besides, the sum involving $a'$ strings disappears due to the remaining equalities, finally leading to the translation operator on $A$. 

For any $j>1$, the partial trace of $T^j$ also forms the product of Kronecker deltas. Every equality chain starting with $a_i$ of the string $a$ must end with $a'_j$ of $a'$ for some $i, j \leq N{A}$. We denote this chain as $[a_i-a'_j]$. Constructing a sequence of distinct chains $[a_i-a'_j][a_j-a'_k]...[a_l-a'_i]$, where subscripts of the last and first elements of consecutive chains match, forms a complete cycle if it covers all bits of $a$, $a'$, and $b$. Then, for any $j>1$, we observe the following implications:
\begin{enumerate}[i]
    \item If $N_A < \gcd(N, j)$, chains starting with $a_i$s always end with $a'_i$s, preventing a complete cycle. However, there will be exactly $\gcd(N, j)$ number of equality chains, each forming an incomplete cycle. Since the endpoints of the chains share the same subscripts, the resulting operator is a constant multiple of $\mathbb{I}_{2^{N_A}}$. 
    
    \item The second possibility is that all the chains can be stacked together to form a full cycle, mapping all $a'_i$s to distinct $a_j$s. Consequently, the resulting operator becomes a permutation operator on subsystem-$A$. A complete cycle can only be formed if $N_A\geq \gcd(N, j)$ (See also Lemma 3.8 in Ref. \cite{sugimoto2023eigenstate}). 
\end{enumerate}       
For $N_A>1$, a full cycle will always form if $N$ assumes a prime number as $N_A\geq \gcd(N, j)=1$ for any $j$.

\section{Proof of Result \ref{sufficient}}\label{verification}
\begin{proof}
Here, we seek to obtain a sufficient condition for the emergence of state designs from randomly chosen T-invariant generator states from $\mathcal{E}^{k}_{\text{TI}}$. In particular, for a randomly chosen $|\phi_{AB}\rangle\in\mathcal{E}^{k}_{\text{TI}}$, we aim to establish a condition on the measurement basis $|\mathcal{B}\rangle\equiv\{|b\rangle\}$ for the following identity:
\begin{align}
\mathbb{E}_{|\phi_{AB}\rangle\in\mathcal{E}^{k}_{\text{TI}}} \left(\sum_{|b\rangle\in\mathcal{B}}\dfrac{\left[\langle b|\phi_{AB}\rangle\langle\phi_{AB}|b\rangle\right]^{\otimes t}}{\left(\langle\phi_{AB}|b\rangle\langle b|\phi_{AB}\rangle\right)^{t-1}}\right)=\dfrac{\bm{\Pi}^{A}_{t}}{d_A(d_A+1)...(d_A+t-1)},
\end{align}
where $d_A=2^{N_A}$, the total Hilbert space dimension of the susbsyetm $A$. Since the expectation ($\mathbb{E}_{|\phi\rangle\in\mathcal{E}^{k}_{\text{TI}}}$) commutes with the summation ($\sum_{|b\rangle\in\mathcal{B}}$), we write 
\begin{align}
\mathbb{E}_{|\phi_{AB}\rangle\in\mathcal{E}^{k}_{\text{TI}}} \left(\sum_{|b\rangle\in\mathcal{B}}\dfrac{\left[\langle b|\phi_{AB}\rangle\langle\phi_{AB}|b\rangle\right]^{\otimes t}}{\left(\langle\phi_{AB}|b\rangle\langle b|\phi_{AB}\rangle\right)^{t-1}}\right)=\sum_{|b\rangle\in\mathcal{B}}\mathbb{E}_{|\phi_{AB}\rangle\in\mathcal{E}^{k}_{\text{TI}}}\left( \dfrac{\left[\langle b|\phi_{AB}\rangle\langle\phi_{AB}|b\rangle\right]^{\otimes t}}{\left(\langle\phi_{AB}|b\rangle\langle b|\phi_{AB}\rangle\right)^{t-1}} \right) .   
\end{align}
We note that for any $|\phi_{AB}\rangle\in\mathcal{H}^{\otimes N}$ and any $|b\rangle\in\mathcal{H}^{\otimes N_B}$, the scalar quantity $\langle\phi_{AB}|b\rangle\langle b|\phi_{AB}\rangle$ is always less than or equal to $1$, i.e., $\langle\phi_{AB}|b\rangle\langle b|\phi_{AB}\rangle\leq 1$. Thus, by writing $(1-(1-\langle\phi_{AB}|b\rangle\langle b|\phi_{AB}\rangle))$ in the denominator, we make use of the infinite series expansion of $1/(1-x)^{t-1}$ to evaluate the above expression. It then follows that 
\begin{align}\label{B3}
\dfrac{\left(\langle b|\phi_{AB}\rangle\langle\phi_{AB}|b\rangle\right)^{\otimes t}}{\left(\langle\phi_{AB}|b\rangle\langle b|\phi_{AB}\rangle\right)^{t-1}}= \left(\langle b|\phi_{AB}\rangle\langle\phi_{AB}|b\rangle\right)^{\otimes t}\sum_{n=0}^{\infty}\binom{n+t-2}{t-2}\sum_{r=0}^n \binom{n}{r} (-1)^r \Tr\left[\left(\langle b|\phi_{AB}\rangle\langle\phi_{AB}|b\rangle\right)^{\otimes r}\right]
\end{align}
Note that the (unnormalized) state $\langle b|\phi_{AB}\rangle\langle\phi_{AB}|b\rangle$ has support solely over the subsystem $A$. For computational convenience, we write it as follows:
\begin{eqnarray}\label{unn}
\langle b|\phi_{AB}\rangle\langle\phi_{AB}|b\rangle &=&\left(\sum_{m_i=0}^{2^{N_A}-1}|m_i\rangle \langle m_i|\right)\left(\langle b|\phi_{AB}\rangle\langle\phi_{AB}|b\rangle\right)\left(\sum_{n_i=0}^{2^{N_A}-1}|n_i\rangle \langle n_i|\right)\nonumber\\
&=&\sum_{m_i=0}^{2^{N_A}-1}\sum_{n_i=0}^{2^{N_A}-1}|m_i\rangle\langle n_i| \Tr\left[ |n_i\rangle\langle m_i|\left(\langle b|\phi_{AB}\rangle\langle\phi_{AB}|b\rangle\right) \right],
\end{eqnarray}
where 
\begin{equation*}
\sum_{m_i=0}^{2^{N_A}-1}|m_i\rangle \langle m_i|=\sum_{n_i=0}^{2^{N_A}-1}|n_i\rangle \langle n_i|= \mathbb{I}_{2^{N_A}}.  
\end{equation*}
Incorporating Eq. (\ref{unn}) into Eq. (\ref{B3}) gives
\begin{align}
\dfrac{\left[\langle b|\phi_{AB}\rangle\langle\phi_{AB}|b\rangle\right]^{\otimes t}}{\left(\langle\phi_{AB}|b\rangle\langle b|\phi_{AB}\rangle\right)^{t-1}}&=\sum_{\substack{m_1, m_2, ..., m_t \\ n_1, n_2, ..., n_t}}|m_1m_2..., m_t\rangle\langle n_1, n_2, ...n_t|\sum_{n=0}^{\infty}\binom{n+t-2}{t-2}\sum_{r=0}^n \binom{n}{r} (-1)^r \nonumber\\
&\hspace{3cm}\Tr\left[|n_1, n_2, ..., n_t\rangle\langle m_1, m_2, ..., m_t|\left(\langle b|\phi_{AB}\rangle\langle\phi_{AB}|b\rangle\right)^{\otimes (t+r)}\right],
\end{align}
In this expression, all the replicas of $\langle b|\phi_{AB}\rangle\langle\phi_{AB}|b\rangle$ are stacked together within the trace operation. This allows us to perform the invariant integration over the states $|\phi_{AB}\rangle\in\mathcal{E}^{k}_{\text{TI}}$, which is evaluated as 
\begin{eqnarray}\label{B6}
\mathbb{E}_{|\phi_{AB}\rangle\in\mathcal{E}^{k}_{\text{TI}}}\left[\left(\langle b|\phi_{AB}\rangle\langle\phi_{AB}|b\rangle\right)^{\otimes (t+r)}\right]=\dfrac{\langle b|\mathbf{T}_{k}|b\rangle^{\otimes (t+r)}\bm{\Pi}^{A}_{t+r}}{\Tr(\mathbf{T}^{\otimes (t+r)}_{k}\bm{\Pi}^{AB}_{t+r})},    
\end{eqnarray}
where $\bm{\Pi}^{A}_{t+r}$ and $\bm{\Pi}^{AB}_{t+r}$ denote projectors onto the permutation symmetric subspaces of $t+r$ copies. While the former acts only on the replicas of the subsystem $A$, the latter acts on the replicas of the entire system $AB$. It is now useful to write $\bm{\Pi}^{A}_{t+r}$ as follows:
\begin{eqnarray}
\bm{\Pi}^{A}_{t+r}=\mathcal{D}_{A, t+r}\int_{|\psi\rangle\in\mathcal{E}_{\text{Haar}}}d\psi \left(|\psi\rangle\langle\psi |\right)^{\otimes (t+r)},\quad\text{where }|\psi\rangle\in\mathcal{H}^{\otimes N_A}
\end{eqnarray}
Where, $\mathcal{D}_{A, t+r}=d_A(d_A+1)...(d_A+t+r-1)$ and $d_A=2^{N_A}$. It then follows that 
\begin{align}\label{sing}
\sum_{|b\rangle\in\mathcal{B}}\mathbb{E}_{\phi_{AB}\in\mathcal{E}^{k}}\left( \dfrac{\left[\langle b|\phi_{AB}\rangle\langle\phi_{AB}|b\rangle\right]^{\otimes t}}{\left(\langle\phi_{AB}|b\rangle\langle b|\phi_{AB}\rangle\right)^{t-1}} \right)=&
\sum_{|b\rangle\in\mathcal{B}}\langle b|\mathbf{T}_{k}|b\rangle^{\otimes t}\int_{|\psi\rangle\in\mathcal{E}_{\text{Haar}}}d\psi \left(|\psi\rangle\langle\psi |\right)^{\otimes t} \sum_{n=0}^{\infty}\binom{n+t-2}{t-2}\nonumber\\
&\hspace{2cm}\sum_{r=0}^n \binom{n}{r} (-1)^r \mathcal{D}_{A, t+r}\dfrac{\langle\psi b|\mathbf{T}_{k}|\psi b\rangle^{r}}{\Tr(\mathbf{T}^{\otimes (t+r)}_{k}\bm{\Pi}^{AB}_{t+r})}.
\end{align}
{The above expression characterizes the moments of the projected ensembles and, thus, the underlying distribution of the projected ensembles.}
A sufficient condition, $\langle b|\mathbf{T}_{k}|b\rangle=\mathbb{I}_{2^{N_A}}$ for all $|b\rangle\in\mathcal{B}$, ensures the convergence of the right-hand side of the above expression to the Haar moments. If satisfied, for all $|b\rangle\in\mathcal{B}$, we will have $\langle \psi b|\mathbf{T}_{k}|\psi b\rangle=1$. As a result, both the integral and the integrand can be decoupled from the infinite series. Then, the infinite series can be understood as the normalizing factor, which necessarily converges to $1/2^{N_B}$. Therefore, we get
\begin{align}\label{suf}
\mathbb{E}_{|\phi_{AB}\rangle\in\mathcal{E}^{k}}\left(\sum_{|b\rangle\in\mathcal{B}} \dfrac{\left[\langle b|\phi_{AB}\rangle\langle\phi_{AB}|b\rangle\right]^{\otimes t}}{\left(\langle\phi_{AB}|b\rangle\langle b|\phi_{AB}\rangle\right)^{t-1}} \right)= \int_{|\psi\rangle\in\mathcal{E}^{A}_{\text{Haar}}}d\psi \left(|\psi\rangle\langle\psi |\right)^{\otimes t}=\dfrac{\bm{\Pi^{A}_{t}}}{2^{N_A}(2^{N_A}+1)...(2^{N_A}+t-1)},
\end{align}
implying that the moments of the projected ensembles, on average, converge towards the Haar moments. 
\end{proof}

It's often challenging to find a basis fully satisfying the condition. The approximate state designs can still be obtained even when the given basis moderately violates the condition. In the computational basis where $b\in\equiv\{0, 1\}^{N_B}$, the number of violations are exponentially suppressed in $N_B$. To further elucidate, we examine $\langle b|\mathbf{T}_{k}|b\rangle$:
\begin{eqnarray}
\langle b|\mathbf{T}_{k}|b\rangle=\mathbb{I}_{2^{N_A}}+\sum_{j=1}^{N-1}e^{2\pi ijk/N}\langle b|T^{j}|b\rangle , 
\end{eqnarray}
where the operators $\langle b|T^j|b\rangle$ for all $j\geq 1$ are sparse matrices with the elements either being zeros or ones. To quantify the violation, in the main text, we studied the quantity $\bm{\Delta}(\mathbf{T}_{k}, \mathcal{B})/2^N$. We found that the violation decays exponentially with $N_B$, where the exponent depends on the particular basis under consideration.

\section{Applicability of Levy's lemma for the moments of the projected ensembles with typical generator states}
{Equation (\ref{suf}) tells us that the $t$-th moment operator of a projected ensemble, when averaged over many random symmetric generator states, equates to the $t$-th Haar moment whenever the sufficient condition holds. Then, Levy's lemma can be used to argue that the moment operator for a typical generator state also approximates the Haar moments, and the distance between them shrinks exponentially with the Hilbert space dimension. In this appendix, we provide some basic details concerning Levy's lemma and its applicability in the context of the projected ensembles.} 

{\textbf{Definition} (Lipschitz continuous functions). A function $f: X\rightarrow Y$ is Lipschitz continuous with Lipschitz constant $\eta$, if for any $x_1, x_2\in X$, it holds that 
\begin{eqnarray}
d_y(f(x_1), f(x_2))\leq \eta d_x(x_1, x_2),     
\end{eqnarray}
where $d_x$ and $d_y$ indicate the distance metrics associated with the spaces $X$ and $Y$,  respectively. The Lipschitz continuity is a stronger form of the uniform continuity of $f$ \cite{o2006metric}, and $\eta$ upper bounds the slope of $f$ in $X$ \cite{milman1986asymptotic, ledoux2001concentration}}. 

\textbf{Levy's lemma} \cite{milman1986asymptotic, ledoux2001concentration}
{Let $f: \mathbb{S}^{d-1} \rightarrow \mathbb{R}$ be a Lipschitz function defined over a $(d-1)$-sphere $\mathbb{S}^{d-1}$, equipped with a natural Haar measure. Suppose a point $x \in \mathbb{S}^{d-1}$ is drawn uniformly at random from $\mathbb{S}^{d-1}$. Then, for any $\varepsilon > 0$, the following concentration inequality holds:
\begin{eqnarray}
\text{Pr}\left[ \left|f(x)-\mathbb{E}_{x\in \mathbb{S}^{d-1}}(f(x))\right|\geq \varepsilon \right]\leq 2\exp\left\{ \dfrac{-d\varepsilon^2}{9\pi^3\eta^2} \right\}, 
\end{eqnarray}
where $\eta$ is the Lipschitz constant of $f$ and $c$ is a positive constant.}

{Proof of Levy's lemma can be found in Ref. \cite{gerken2013measure}. Levy's lemma guarantees that the value of a Lipschitz continuous function at a typical $x\in \mathbb{S}^{d-1}$ is always close to its mean value, as given by $\mathbb{E}_{x\in \mathbb{S}^{d-1}}(f(x))$. The difference between the mean and a typical value is exponentially suppressed with the Hilbert space dimension}.

{Let us now consider the moments of a projected ensemble for an arbitrary generator state ($|\phi\rangle$) with symmetry: 
\begin{eqnarray}\label{levylemma}
\mathcal{M}^{t}(|\phi\rangle)=\sum_{|b\rangle\in\mathcal{B}} \dfrac{\left[\langle b|\phi\rangle\langle\phi|b\rangle\right]^{\otimes t}}{\left(\langle\phi|b\rangle\langle b|\phi\rangle\right)^{t-1}} .   
\end{eqnarray}
For Haar random generator state, $[\mathcal{M}^{t}(|\phi\rangle)]_{ij}$, the elements of the moment operator, have been shown to be Lipschitz continuous functions of $|\phi\rangle$ with the Lipschitz constant $\eta\leq 2(2t-1)$. Then, the Levy's lemma as given in Eq. (\ref{levylemma}) implies that the differences between the elements of the projected ensemble moments and their respective means are exponentially suppressed with the total Hilbert space dimension $2^{N}$. A detailed explanation for the same can be found in Ref. \cite{cotler2023emergent}.
Since the ensembles of uniform random states with symmetry are subsets of the Haar ensemble, the corresponding $[\mathcal{M}^{t}(|\phi\rangle)]_{ij}$s remain Lipschitz continuous with the same Lipschitz constant $\eta$ (or smaller than $\eta$)}.

\section{Violation of the sufficient condition for $r=2$}
\label{vior2}
In this appendix, we examine Eq. (\ref{tria}) for $r=2$ and seek to identify the measurement bases that strongly violate the sufficient condition. In the main text, we have carried out the analysis for $r=1$ and we observed that the eigenbases of the operators $u_{N_A+1}T_B$ for all $u_{N_A+1}\in U(d)$ significantly violate the condition, where $T_B$ denotes the translation operator supported over $B$ and the subscript $N_A+1$ denotes that the unitary acts on the site labeled $N_A+1$. When $u_{N_A+1}=\mathbb{I}_{2}$, the operator is simply a translation operator over $B$. The eigenbasis of this operator is locally translation invariant. However, for a random $u_{N_A+1}$, the translation symmetry gets weakly broken. As $r$ is increased further, the local translation symmetry of the measurement basis gradually disappears. To see this for $r=2$, we consider the following equality: 
\begin{eqnarray}
\bm{\Delta}(\mathbf{T}_{k}, \mathcal{B})=\sum_{|b\rangle\in\mathcal{B}}\left(\left\|e^{2\pi ir/N} \langle b|T^{2}|b\rangle +\sum_{j\neq 2}e^{2\pi ijk/N} \langle b| T^j |b\rangle \right\|_{1}\right).  
\end{eqnarray}
We now write $T^2$ as 
\begin{eqnarray}
 T^2&=&TT\nonumber\\
 &=&\left(S_{12}S_{23}\cdots S_{N-1, N}\right)\left(S_{12}S_{23}\cdots S_{N-1, N}\right)
\end{eqnarray}
For simplicity, we take $N_A=3$. The partial expectation of $T^2$ with respect to a basis vector $|b\rangle\in\mathcal{B}$ can be written as 
\begin{eqnarray}
\langle b|T^2|b\rangle &=& S_{12}S_{23}S_{12}\langle b| S_{34}S_{23}\left( S_{45}S_{34}S_{56}\cdots S_{N-1, N} \right)\left( S_{45}S_{56}\cdots S_{N-1, N} \right) |b\rangle
\end{eqnarray}
We now substitute the integral expression of the swap operators corresponding to $S_{34}$ and $S_{23}$. It then follows that 
\begin{eqnarray}
\langle b|T^2|b\rangle&=&S_{13}\int_{u\in U(d)}d\mu(u)\int_{v\in U(d)} d\mu(v)\int_{w\in U(d)}d\mu(w)u_{2}v_{3}u_{3}w_{3}\langle b| v_{4} \left( S_{45}w_{4}S_{56}\cdots S_{N-1, N} \right)\left( S_{45}S_{56}\cdots S_{N-1, N} \right)|b\rangle \nonumber\\
&=&S_{13}\int_{u\in U(d)}d\mu(u)\int_{v\in U(d)} d\mu(v)\int_{w\in U(d)}d\mu(w)u_{2}v_{3}u_{3}w_{3}\langle b| v_{4}T_Bw_{4}T_{B} |b\rangle
\end{eqnarray}
If the measurement basis is the eigenbasis of the operator $( v_4T_Bw_4T_B)$ for some $v$ and $w$ being local Haar random unitaries, one can expect that $\sum_{|b\rangle\mathcal{B}}\|\langle b| T^2 |b\rangle\|_{1}\sim O(2^{N_B})$. For $v=w=\mathbb{I}_{2}$, the above operator becomes $T^2_B$. The eigenbasis of this operator is invariant under translations by two sites. Likewise, one can show that for an arbitrary integer $r$, the eigenbasis of $T^r$ strongly violates the condition. Figures \ref{r2r3}a-\ref{r2r3}c demonstrate the decay of trace distance by considering the measurements in the eigenbases of $T^2_{B}$ and $T^3_{B}$ operators. From the figure, it is evident that the design formation is obstructed. This suggests that the sufficient condition we derived could potentially be necessary as well for the emergence of higher-order state designs.  

\begin{figure*}
\includegraphics[scale=0.5]{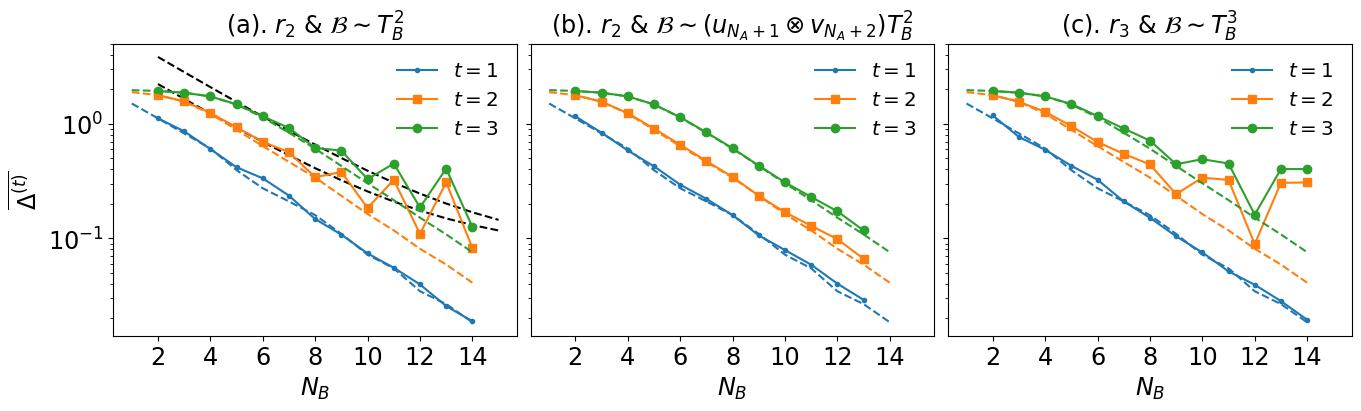}
\caption{\label{r2r3} Illustration of $\overline{\Delta^{(t)}}$ versus $N_B$ for the first three moments when the measurements are performed in bases largely violating the sufficient condition. {Here, we fix $N_A=3$}. In panels (a) and (b), the measurement bases are the eigenbases of $T^2_{B}$ and $(u_{N_A+1}\otimes v_{N_A+1})T^2_{B}$, where $T_B$ denotes the local translation operator over the subsystem-$B$. We observe that in the former case, the trace distances for higher order moments initially decay and acquire oscillatory behavior around exponential curves decaying to finite non-zero values for larger values of $N_B$. Due to the applications of local Haar random unitaries, the behavior in the latter case appears to decay for smaller $N_B$ values. However, for larger values of $N_B$, we anticipate that the trace distance approaches a finite non-zero value. Panel (c) corresponds to the measurements in the eigenbasis of $T^3_{B}$. {In all the above panels, the averages are taken over $10$ samples of the initial generator states from the respective state ensembles}.   
} 
\end{figure*}

\section{Characterization of the distribution of the projected ensembles }
\label{app-moments}
{In Appendix \ref{verification}, Eq. (\ref{sing}) provides an expression for the moments of the projected ensembles with respect to a given measurement basis, averaged over the initial generator states from the symmetry-restricted ensembles. These moments characterize the underlying distribution of the ensembles to some extent. The resultant distribution depends both on the initial symmetry and the measurement basis considered. When the sufficient condition is satisfied, the resultant distribution displays the moments of the Haar ensemble. It's interesting to examine the distribution of projected ensembles when the sufficient condition does not hold. In this appendix, we provide closed-form expressions for the moments of the projected ensembles. To illustrate, we take the generator states from $\mathbb{Z}_2$-symmetric ensembles and perform the measurements in the local product basis given by $\mathcal{B}\equiv\{+, -\}^{N_B}$. 
For the $\mathbb{Z}_2$-symmtric case, the resulting distribution of the projected ensembles from the measurements in $\sigma^x$ basis indeed displays the moments, which are linear combinations of the symmetry-restricted moments corresponding to both parities. In the following, we shall show this analytically. }

{To begin with, let us recall that the moments of the projected ensembles averaged over the ensemble of $\mathbb{Z}_2$-symmetric states are given by 
\begin{eqnarray}\label{moment}
\mathcal{M}^{t}_{\mathbb{Z}_{2}}=\mathbb{E}_{|\phi\rangle\in\mathcal{E}^{k}_{\mathbb{Z}_{2}}}\left(\sum_{|b\rangle\in\mathcal{B}}\dfrac{\left[\langle b|\phi\rangle\langle\phi|b\rangle\right]^{\otimes t}}{\left(\langle\phi|b\rangle\langle b|\phi\rangle\right)^{t-1}}\right) = \int_{u\in U_{\mathbb{Z}_2}(2^N)} d\mu(u) \left(\sum_{|b\rangle\in\mathcal{B}}\dfrac{\left[\langle b|u|\phi\rangle\langle \phi|u^{\dagger}|b\rangle\right]^{\otimes t}}{\left(\langle \phi|u^{\dagger}|b\rangle\langle b|u|\phi\rangle\right)^{t-1}}\right), 
\end{eqnarray}
where $U_{\mathbb{Z}_{2}}(2^N)\subset U(2^N)$ denotes the set of all unitaries that commute with the operator $\otimes_{i=1}^{N}\sigma^x_{i}$. One can easily show that $U_{\mathbb{Z}_{2}}(2^N)$ is a compact subgroup of $U(2^N)$.} 

\begin{figure}
    \centering
    \includegraphics[width=0.5\linewidth]{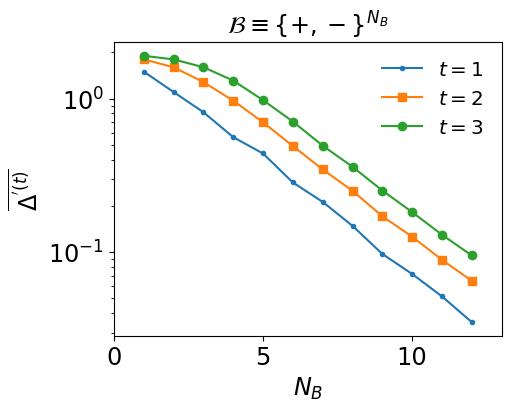}
    \caption{{Illustration of the average trace distance $\overline{\Delta^{'(t)}}$ vs $N_B$ for the random generator states with the $\mathbb{Z}_{2}$-symmetry for the first three moments. We fix the charge $k=0$ and $N_A=3$. The measurements are performed in local product basis $\mathcal{B}\equiv \{+, -\}^{N_B}$. For numerical purposes, the average trace distance is evaluated by considering $10$ samples of the initial generator states.}}
    \label{fig1}
\end{figure}

{Let $U_{\mathbb{Z}_{2}}(2^{N_A})$ denotes the set of all unitaries that commute with $\bigotimes_{i=1}^{N_A}\sigma^{x}_{i}$. Then, one can verify that the elements of $U_{\mathbb{Z}_{2}}(2^{N_A})$ also commute with the symmetry operator \( \bigotimes_{i=1}^{N} \sigma^x_i \). This will imply that $U_{\mathbb{Z}_{2}}(2^{N_{A}})\subset U_{\mathbb{Z}_{2}}(2^{N})$ \footnote{Note that for non-on-site symmetries, such as translation and reflection symmetries, this statement does not hold, so the following analysis cannot be extended straightforwardly.}. By making use of invariance of Haar measure associated with $U_{\mathbb{Z}_{2}}(2^{N})$, we replace $u$ in Eq. (\ref{moment}) with $vu$, where $v\in U_{\mathbb{Z}_{2}}(2^{N_{A}})$. It then follows that
\begin{eqnarray}
\mathcal{M}^{t}_{\mathbb{Z}_{2}}=\sum_{|b\rangle\in\mathcal{B}} \int_{u\in U_{\mathbb{Z}_2}(2^N)} d\mu(u)  \left(\dfrac{\left[\langle b|vu|\phi\rangle\langle \phi|u^{\dagger}v^{\dagger}|b\rangle\right]^{\otimes t}}{\left(\langle \phi|u^{\dagger}|b\rangle\langle b|u|\phi\rangle\right)^{t-1}}\right)=v^{\otimes t}\mathcal{M}^{t}_{\mathbb{Z}_{2}} v^{\dagger \otimes t}. 
\end{eqnarray}
The unitary freedom in the above equation can be used to show that the integrand and the denominator are independent random variables. This can be done by using arguments similar to those used in Appendix \ref{trans_designs}. It is then straightforward to write the moments of the projected ensemble as
\begin{eqnarray}\label{ind-meas}
\mathcal{M}^{t}_{\mathbb{Z}_{2}} \propto \sum_{|b\rangle\in\mathcal{B}}\langle b| \mathbf{Z}_{k} |b\rangle^{\otimes t}\mathbf{\Pi}^{t}_{A}, 
\end{eqnarray}
where $\mathbf{\Pi}^{t}_{A}=\sum_{j}\pi_{j}$ is the projector onto the permutation symmetric subspace of $t$-copies of the Hilbert space spanning $N_A$ sites. Note that Eq. (\ref{ind-meas}) holds for any measurement basis. We now fix the measurement basis to be $\mathcal{B}\equiv\{+, -\}^{N_B}$. For any $|b\rangle\in \mathcal{B}$, the partial inner product $\langle b| \mathbf{Z}_{k} |b\rangle$ can be evaluated as follows: 
\begin{eqnarray}
\langle b| \mathbf{Z}_{k} |b\rangle = \mathbb{I}_{2^{N_A}}+(-1)^{k+\sum_{i=1}^{N_B}\text{sgn}(b_i)}\left(\otimes_{i=1}^{N_A}\sigma^{x}_{i}\right)=\mathbf{Z}_{k', N_A},   
\end{eqnarray}
where $k'=k+\sum_{i=1}^{N_B}\text{sgn}(b_i)$. The subscript $N_A$ in $\mathbf{Z}_{k', N_A}$ distinguishes the subspace projector acting on all the sites ($\mathbf{Z}_{k}$) from the one that acts only on $N_A$-number of sites ($\mathbf{Z}_{k', N_A}$). Depending on the values taken by $k$ and the parity of the basis vector $|b\rangle$, r.h.s of the above expression becomes a subspace projector onto one of the eigenspaces of the $\mathbb{Z}_{2}$-symmetry operator. Moreover, one can verify that half of the basis vectors have even parity and the remaining half have odd parity. It then follows that 
\begin{eqnarray}
\mathcal{M}^{t}_{\mathbb{Z}_{2}}=\dfrac{1}{\mathcal{N}} \left( \mathbf{Z}^{\otimes t}_{0, N_A} + \mathbf{Z}^{\otimes t}_{1, N_A}\right) \mathbf{\Pi}^{t}_{A}, 
\end{eqnarray}
where $\mathcal{N}$ denotes the normalization constant and is given by $\mathcal{N}=\text{Tr}\left[ \left( \mathbf{Z}^{\otimes t}_{0, N_A} + \mathbf{Z}^{\otimes t}_{1, N_A}\right) \mathbf{\Pi}^{t}_{A}\right]$.}

{We now numerically calculate $\overline{\Delta^{'(t)}}$, the average trace distance between the moments of the projected ensembles of a typical $\mathbb{Z}_{2}$-symmetric generator state and $\mathcal{M}_{\mathbb{Z}^{t}_{2}}$, as the system size varies. 
\begin{eqnarray}
    \Delta'^{(t)}=\left\| \left(\sum_{|b\rangle\in\mathcal{B}}\dfrac{\left[\langle b|\phi_{AB}\rangle\langle\phi_{AB}|b\rangle\right]^{\otimes t}}{\left(\langle\phi_{AB}|b\rangle\langle b|\phi_{AB}\rangle\right)^{t-1}}\right) - \mathcal{M}^{t}_{\mathbb{Z}_{2}}\right\|_{1}.
\end{eqnarray}
We evaluate $\Delta^{'(t)}$ and average it over many samples of the initial generator states, which is denoted by $\overline{\Delta^{'t}}$. Supporting numerical results are shown in Fig. \ref{fig1}. We observe that the average trace distance $\overline{\Delta^{'t}}$ exponentially converges to zero as the measured subsystem size $N_B$ increases. }

\section{Transition in the randomness of the projected ensemble}
\label{Z2transition}
In the main text, while analyzing the emergence of state designs from $Z_2$, we have observed that the measurements in the eigenbasis of $\otimes_{N_B}\sigma^x$, i.e.,  ($\mathcal{B}\equiv\{|b\rangle\}$, $b\in\{+, -\}^{N_B}$) results in constant violation of the sufficient condition. However, changing measurements on a single arbitrary site to the $\sigma^z$ basis while keeping $\sigma^x$ measurements on other sites results in zero violation of the sufficient condition. Consequently, the trace distance $\overline{\Delta^{(t)}}$ converges exponentially to zero with the size of the system. In this appendix, we analyze the crossover from non-convergence to convergence in the projected ensembles towards the designs. In particular, we fix $\sigma^x$ measurement basis for $N_B-1$ sites and take the eigenbasis of $\alpha\sigma^x+(1-\alpha)\sigma^z$ for the local measurement over $N_B$-th site. As the parameter varies, we observe a transition in $\overline{\Delta^{(t)}}$ from a finite constant value toward a system size-dependent value. We show the corresponding results in Fig. \ref{fig:z2trans}. From the figure, we observe that near $\alpha=0$, the trace distance $\overline{\Delta^{(t)}}$ remains system size independent. In contrast, as $\alpha$ approaches $1$, the trace distance becomes sensitive to the system size $N$.

The violation of the sufficient condition, in this case, can be quantified as follows:
\begin{align}\label{G1}
\dfrac{\bm{\Delta}(\mathbf{Z}_{k}, \mathcal{B})}{2^{N_B}} &=\dfrac{1}{2^{N_B}}\sum_{b\in\mathcal{B}}\left\| \langle b|\mathbf{Z}_{k}|b\rangle -\mathbb{I}_{2^{N_A}} \right\|_{1}\nonumber\\
&=\dfrac{1}{2^{N_B}}\sum_{b_{1}\cdots b_{N_B-1}\in\{+, -\}^{N_B-1}}\sum_{b_{N_B}}\left\| (-1)^{k+\sum_{i=1}^{N_B-1}\text{sgn}(b_i)}\langle b_{N_B}| \sigma^x |b_{N_B}\rangle\otimes_{j=1}^{N_A} \sigma^{x}_{j} \right\|\nonumber\\
&=\dfrac{1}{2^{N_B}}\left\|\otimes_{j=1}^{N_A} \sigma^{x}_{j} \right\|\sum_{b_{1}\cdots b_{N_B-1}\in\{+, -\}^{N_B-1}}\sum_{b_{N_B}} |\langle b_{N_B}| \sigma^x |b_{N_B}\rangle|\nonumber\\
&=2^{N_A-1}\sum_{b_{N_B}} \left|\langle b_{N_B}| \sigma^x |b_{N_B}\rangle\right|,
\end{align}
where $\{|b_{N_B}\rangle\}$ denotes the eigenbasis of the operator $\alpha\sigma^z+(1-\alpha)\sigma^x$. From Eq. (\ref{G1}), we notice that the violation remains independent of the system size ($N$) and depends only on the parameter $\alpha$. Near $\alpha=0$, the violation stays nearly as constant ($\approx 1$) as depicted in Fig. \ref{fig:z2trans}c. Since $\dfrac{\bm{\Delta}(\mathbf{Z}_{k}, \mathcal{B})}{2^{N_B}}$ remains independent of $N$, the projected ensembles do not converge to the designs even in the limit of large $N$ when $\alpha$ is close to $0$. Hence, $\Delta^{(t)}$ remains nearly constant for all $N$. On the contrary, as $\alpha$ approaches $1$, the violation decays to zero, implying the convergence of the projected ensembles to the designs in the large $N$ limit. This may be understood as the transition of the projected ensemble from a localized distribution to a uniform distribution over the Hilbert space. 

\begin{figure*}
\includegraphics[scale=0.5]{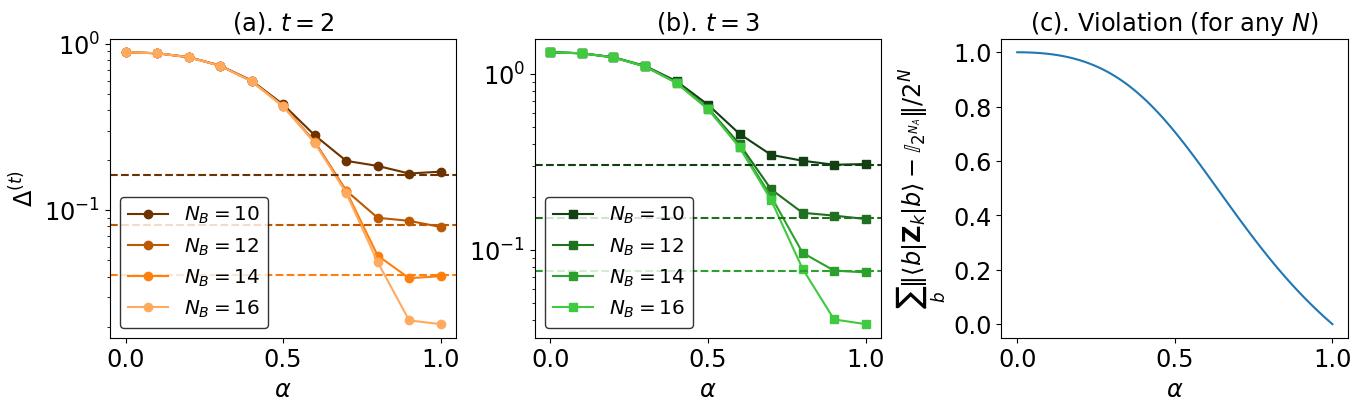}
\caption{\label{fig:z2trans} The figure illustrates the transition in the randomness of the projected ensemble when the initial generator states are generic states with $Z_2$ symmetry. Local $\sigma^x$ basis measurements are fixed for $N_B-1$ sites. The measurements on $N_B$-th site are performed in the eigenbasis of $\alpha\sigma^z+(1-\alpha)\sigma^x$. The trace distance between the moments of the Haar ensemble and the projected ensemble, $\overline{\Delta^{t}}$, is plotted against the parameter $\alpha$ for $t=2$ and $3$ for different system sizes. The dashed lines correspond to $\overline{\Delta^{t}}$ of that of Fig. \ref{fig:ZZ2plustrans}a. Note that the case of $t=1$ is trivial and stays nearly a constant for any $\alpha$, as it is independent of the measurement basis considered. } 
\end{figure*}

\section{Further details on deep thermalization in Ising chain}\label{comparision}
\subsection{Contrasting deep thermalization with and without translation symmetry}\label{symbrok_comparision}
{In the main text, we have examined the deep thermalization in the Ising chain having homogeneous model parameters and with periodic boundary conditions (PBCs). Recall that PBCs, together with homogeneous interaction and fields, imply translation symmetry in the considered system. Due to this symmetry, the entanglement builds up in the time-evolved state at a rate double that of the case with open boundary. As we notice here, the same is reflected in the phenomenon of emergent state designs. 
In this appendix, we study deep thermalization in the absence of translation symmetry and contrast it with the translation-symmetric case. }

\subsubsection{Through modification of boundary conditions}
{Here, we break the translation symmetry of the model by considering (i) open boundary condition (OBC) where interaction between the last and first spins of the chain is absent and (ii) inhomogeneous interaction between a pair of spins with closed boundary. Specifically, we contrast the initial decay of the trace distance measure for these two cases with the translation symmetric case. The corresponding results are shown in Fig. \ref{fig:ising}.  
The OBC implies $J_{1,N} = 0$, where $J_{1, N}$ represents the interaction strength between the first and $N$-th spins. In this case, \(\Delta^{(t)}(\tau)\) displays a power-law decay \(\sim \tau^{-1.2}\) in the early time regime. This has nearly half the exponent of $\sim \tau^{-2.2}$ scaling observed in the translation-symmetric case. In order to interpolate between these two cases, we consider the system with a closed boundary but with an inhomogeneous interaction between a pair of spins. This is implemented by considering $J_{1, N}=0.5$. Notice that though the system is now closed, the translation symmetry still remains absent, and we refer to this case as moderately broken translation symmetry. Here, we observe that the power-law decay $\sim \tau^{-1.8}$ faster than in the OBC case, however, it is still slower than in the PBCs case with homogeneous model parameters. In all the cases, the trace distance measure saturates at late times to appropriate RMT predicted values, which are marked by the horizontal dashed lines. }

\begin{figure}[ht!]
\includegraphics[scale=0.5]{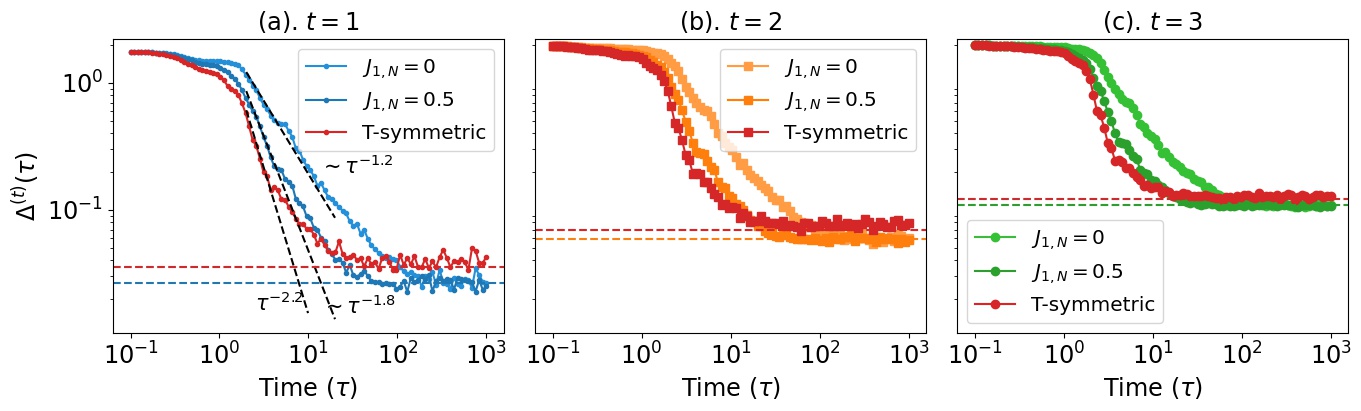}
\caption{\label{fig:ising} The figure contrasts the decay of $\Delta^{(t)}$ for the state $|0\rangle^{\otimes N}$ evolved under the Hamiltonian of the Ising chain with translation symmetry (red curves) with that of moderately broken translation symmetry and the one with open ends (blue, orange, and green curves for $t=1$, $2$, and $3$). The latter cases are characterized by the interaction strength between the first and $N$-th spins, $J_{1, N} = 0.5$ and $0$, respectively. The data is presented for $N=16$ and $N_A=3$. The panels correspond to the first three moments $t=1$, $2$, and $3$, respectively. We show the power-law scaling to compare the convergence rate in the early time regime. At late times, the data saturate to the predicted RMT values with and without symmetry, which are marked by the horizontal dashed lines.} 
\end{figure}

\subsubsection{Through introduction of disorder}
\label{comparision_disord}
{In this section, we consider the Ising chain with inhomogeneous model parameters. We employ this by introducing diagonal and off-diagonal disorders. In particular, we consider two cases by randomizing the strengths of the (i) interactions and (ii) transverse fields. Thus, the Hamiltonian of the Ising model with these disorders can be written as  
\begin{eqnarray}\label{inhomo_ising}
H=\sum_{i=1}^{N} (J+\eta_{i})\sigma^{x}_{i} \sigma^{x}_{i+1}+\sum_{i=1}^{N}h_x\sigma^{x}_{i}+\sum_{i=1}^{N}(h_y+\xi_i)\sigma^{y}_{i}, 
\end{eqnarray}
where $\eta_{i}$ and $\xi_i$ are independent and identically distributed random variables chosen from Gaussian distribution $\mathcal{N}(0, v)$ with zero mean and variance $v$.  Similar to the main text, here we study the model with a closed boundary. The case with $\eta_{i} = \xi_i = 0$ for all $i=1, 2, \cdots, N$, and $J = 1$ corresponds to the Hamiltonian considered in Eq. \eqref{ising} of the main text. Note that the model needs to be chaotic in order to obtain the deep-thermalization characteristics. It is known that the Ising model with both transverse and longitudinal fields is nonintegrable for non-zero values of $h_x$, $h_y$, and $J$ \cite{kim2023nonintIsing, kim2014testing, sharma2015quenches, cotler2023emergent}. But, in order to stay close to the point where the model is robustly chaotic, we choose $\left[J, h_{x}, h_y\right] = \left[1, (\sqrt{5}+1)/4, (\sqrt{5}+5)/8\right]$, as also considered in the main text. Then, the random variables $\eta_{i}$ and $\xi_i$ make the model inhomogeneous, thereby breaking the translation symmetry. The variance $v$ of the random variables controls the strength of the disorder introduced in the otherwise translation symmetric system. In other words, $v$ quantifies the degree of translation symmetry breaking in the considered model.}

{In Fig. \ref{fig:ising-disorder}, we illustrate the decay of $\Delta^{(t)}(\tau)$ for the disordered Ising Hamiltonian presented in Eq. (\ref{inhomo_ising}), and contrast it with the clean, homogeneous case. Figures \ref{fig:ising-disorder}a-\ref{fig:ising-disorder}c show the evolution when the disorder is introduced by randomizing the interaction strengths. We consider a few values of $v$ to demonstrate the effects of disorder with increasing strength. However, we keep $v$ considerably small such that the model remains chaotic, which can also be inferred from the decreasing trend and long-time saturation of $\Delta^{(t)}(\tau)$. We notice that in the early time regime, the trace distance measure shows slower convergence in comparison to the clean case, and this rate decreases as the strength of the disorder is enhanced. In particular, the numerical results depict that the disorder ensemble averaged $\Delta^{(t)}(\tau)$ has the power-law scaling as $\sim\tau^{-1.4},$ $\sim\tau^{-1.1},$ and $\sim\tau^{-0.8}$ for $v = 0.3$, $0.4$ and $0.5$, respectively. This can be contrasted with the clean case, where the decay follows $\sim \tau^{-2.2}$ law. For all the cases, $\Delta^{(t)}(\tau)$ saturate at late time to appropriate RMT predicted values, which are marked by the horizontal dashed lines. Similarly, Figs. \ref{fig:ising-disorder}d-\ref{fig:ising-disorder}f illustrate the evolution when the disorder is present only in the transverse fields. As previous, we observe that in the early time-regime the decay of trace distance measure gets slower with increasing disorder strength. Consequently, the above analysis unveils that the Ising model with translation symmetry exhibits faster convergence of the trace distance measure during the early time evolution in comparison to the cases when the symmetry is broken.}

\begin{figure}[ht!]
\includegraphics[scale=0.5]{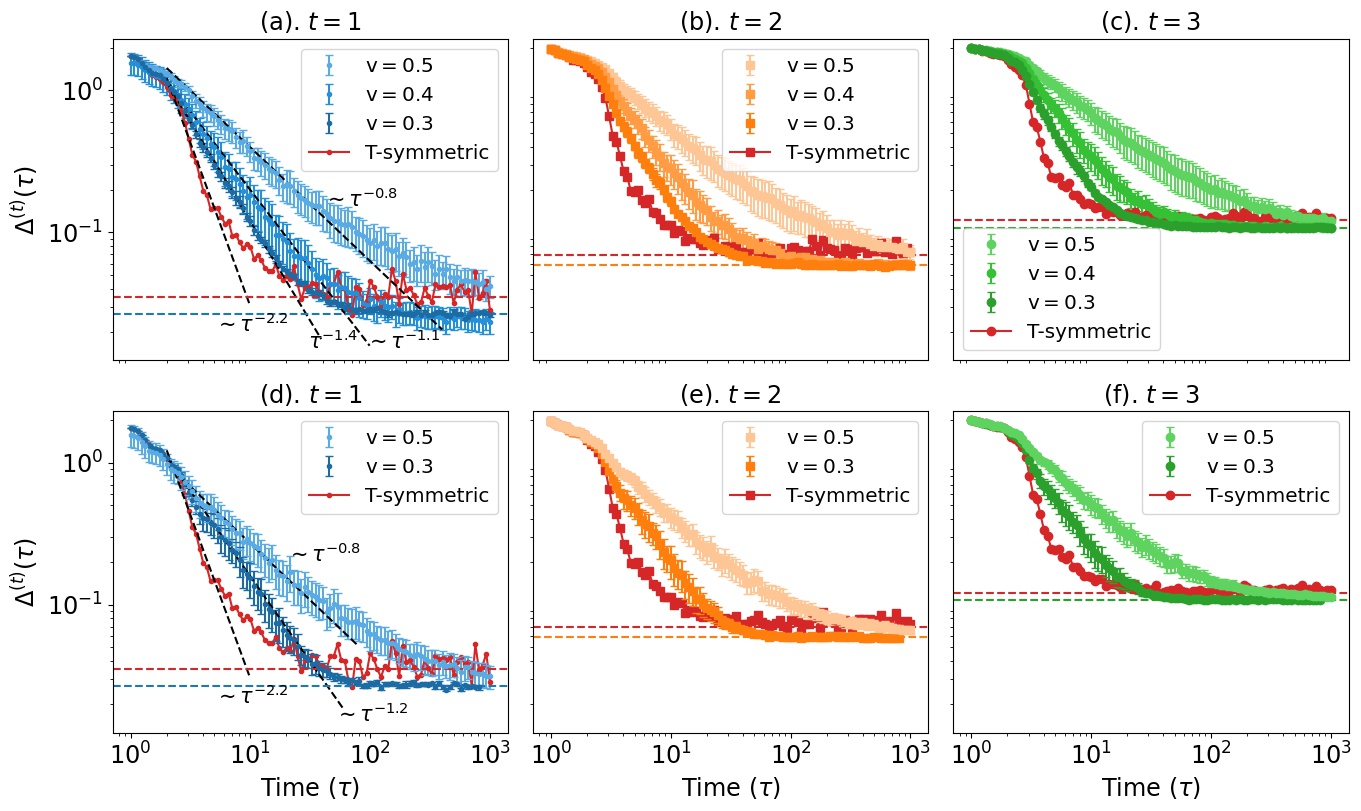}
\caption{\label{fig:ising-disorder} 
The figure contrasts the decay of $\Delta^{(t)}$ for the state $|0\rangle^{\otimes N}$ evolved under the dynamics of translation-symmetric Ising model (red curves) and disordered Ising model given in Eq. \eqref{inhomo_ising} (blue, orange, and green curves for $t=1$, $2$, and $3$).  Like in the main text, we employ PBCs, and the data is presented for $N = 16$ and $N_A = 3$. (a)-(c) and (d)-(f) illustrate the results when the disorder is introduced by randomizing the strengths of the interactions and transverse fields, respectively. The variance $v$ of the randomized parameters controls the strength of the disorder (see Appendix \ref{comparision_disord}), and lighter shading corresponds to a stronger disorder. 
For the disordered cases, ensemble-averaged (over $10$ realizations) data with standard ensemble error (shown as error bars) is presented. We show the power-law scaling to compare the convergence rate in the early time regime. 
At late times, the data saturate to the predicted RMT values with and without symmetry, which are marked by the horizontal dashed lines.} 
\end{figure}

\subsection{Benchmarking the late time evolution with translation plus reflection symmetric random states}\label{benchmark_latetim}
\begin{figure}[ht!]
    \centering
    \includegraphics[width=0.4\linewidth]{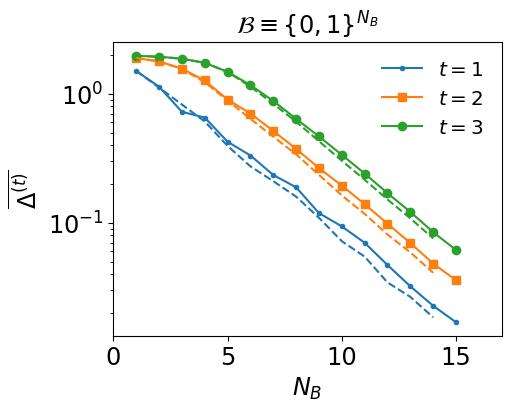}
    \caption{Comparison of the average trace distance $\overline{\Delta^{t}}$ versus $N_B$ for the first three moments when the initial generator states are simultaneous eigenvectors of the translation and reflection symmetries (shown with thick lines) with the case of Haar random generator states (shown with dashed lines). The measurements are performed in $\sigma^z$ basis, and $N_A$ is fixed at $3$. Additionally, note that we have considered $10$ samples of the initial generator states to numerically evaluate the average trace distance. }
    \label{ref-tran}
\end{figure}
{The Ising Hamiltonian considered in this work displays reflection symmetries about every site in addition to the translation symmetry. Moreover, the initial state $|0\rangle^{\otimes N}$ is a common eigenvector of all the aforementioned symmetry operators. Hence, to obtain the RMT predictions, it is imperative to evaluate the trace distance $\overline{\Delta^{t}}$ for the ensemble of states that are common eigenvectors of the translation and reflection symmetries with the eigenvalue $1$. It is to be noted that the symmetry operators corresponding to distinct reflection operations do not commute with each other, nor do they commute with the translation operator. Nevertheless, they all share an overlapping eigenspace. Therefore, under the dynamics of the chaotic Ising chain, the initial state evolves and equilibrates to one of the states belonging to the ensemble spanned by the common eigenspace of all the symmetry operators. Starting from a global Haar random state $|\psi\rangle$, one can construct uniform random states with the aforementioned symmetries as follows:
\begin{eqnarray}\label{tran_ref}
   |\phi\rangle=\dfrac{1}{\mathcal{N}}\mathbf{R}^{N-1}_{0}\mathbf{R}^{N-2}_{0}\cdots \mathbf{R}^{0}_{0}\mathbf{T}_{0}|\psi\rangle,  
\end{eqnarray}
where $\mathbf{R}^{j}_{0}$ denotes the projector onto the reflection (around $j$-th site) symmetric subspace with the eigenvalue $(-1)^{0}=1$. Its worth noting that if a state $|\psi\rangle$ is a simultaneous eigenvector of a reflection operator about an arbitrary site and also the translation operator, then it will also be an eigenvector of other reflection operators corresponding to remaining $N-1$ sites. Having constructed the ensemble of states as given in Eq. (\ref{tran_ref}), one can proceed with the computation of $\overline{\Delta^{t}}$. The corresponding results for the first three moments are shown In Fig. \ref{ref-tran}. These results are compared with the case when the initial states are completely Haar random (shown with dashed lines). We notice slight differences between both cases. These differences may be attributed to the non-commuting nature of the translational and reflection symmetries. The RMT values obtained in this figure, by taking all the symmetries into account, can be used to compare the saturation values in Fig. \ref{fig:ising-phy} of the main text.}

\vspace{1cm}
\twocolumngrid 
\bibliographystyle{myunsrtnat}
\bibliography{references}

\begin{thebibliography}{81}
\providecommand{\natexlab}[1]{#1}
\providecommand{\url}[1]{\texttt{#1}}
\expandafter\ifx\csname urlstyle\endcsname\relax
  \providecommand{\doi}[1]{doi: #1}\else
  \providecommand{\doi}{doi: \begingroup \urlstyle{rm}\Url}\fi

\bibitem[Emerson et~al.(2005)Emerson, Alicki, and Życzkowski]{benchmarking1}
Joseph Emerson, Robert Alicki, and Karol Życzkowski.
\newblock Scalable noise estimation with random unitary operators.
\newblock J. Opt. B: Quantum and Semiclass. Opt., 7:\penalty0 S347, 2005.
\newblock \doi{10.1088/1464-4266/7/10/021}.

\bibitem[Knill et~al.(2008)Knill, Leibfried, Reichle, Britton, Blakestad, Jost, Langer, Ozeri, Seidelin, and Wineland]{knill2008randomized}
E.~Knill, D.~Leibfried, R.~Reichle, J.~Britton, R.~B. Blakestad, J.~D. Jost, C.~Langer, R.~Ozeri, S.~Seidelin, and D.~J. Wineland.
\newblock Randomized benchmarking of quantum gates.
\newblock Phys. Rev. A, 77:\penalty0 012307, 2008.
\newblock \doi{10.1103/PhysRevA.77.012307}.

\bibitem[Dankert et~al.(2009)Dankert, Cleve, Emerson, and Livine]{benchmarking2}
Christoph Dankert, Richard Cleve, Joseph Emerson, and Etera Livine.
\newblock Exact and approximate unitary 2-designs and their application to fidelity estimation.
\newblock Phys. Rev. A, 80:\penalty0 012304, 2009.
\newblock \doi{10.1103/PhysRevA.80.012304}.

\bibitem[Vermersch et~al.(2019)Vermersch, Elben, Sieberer, Yao, and Zoller]{vermersch2019probing}
B.~Vermersch, A.~Elben, L.~M. Sieberer, N.~Y. Yao, and P.~Zoller.
\newblock Probing scrambling using statistical correlations between randomized measurements.
\newblock Phys. Rev. X, 9:\penalty0 021061, 2019.
\newblock \doi{10.1103/PhysRevX.9.021061}.

\bibitem[Elben et~al.(2023)Elben, Flammia, Huang, Kueng, Preskill, Vermersch, and Zoller]{elben2023randomized}
Andreas Elben, Steven~T. Flammia, Hsin-Yuan Huang, Richard Kueng, John Preskill, Beno{\^i}t Vermersch, and Peter Zoller.
\newblock The randomized measurement toolbox.
\newblock Nat. Rev. Phys., 5\penalty0 (1):\penalty0 9--24, 2023.
\newblock \doi{10.1038/s42254-022-00535-2}.

\bibitem[Harrow and Low(2009)]{harrow2009random}
Aram~W. Harrow and Richard~A. Low.
\newblock Random quantum circuits are approximate 2-designs.
\newblock Commun. Math. Phys., 291\penalty0 (1):\penalty0 257--302, 2009.
\newblock \doi{10.1007/s00220-009-0873-6}.

\bibitem[Brown and Viola(2010)]{brown2010convergence}
Winton~G. Brown and Lorenza Viola.
\newblock Convergence rates for arbitrary statistical moments of random quantum circuits.
\newblock Phys. Rev. Lett., 104:\penalty0 250501, 2010.
\newblock \doi{10.1103/PhysRevLett.104.250501}.

\bibitem[Smith et~al.(2013)Smith, Riofr\'{\i}o, Anderson, Sosa-Martinez, Deutsch, and Jessen]{smith2013quantum}
A.~Smith, C.~A. Riofr\'{\i}o, B.~E. Anderson, H.~Sosa-Martinez, I.~H. Deutsch, and P.~S. Jessen.
\newblock Quantum state tomography by continuous measurement and compressed sensing.
\newblock Phys. Rev. A, 87:\penalty0 030102, 2013.
\newblock \doi{10.1103/PhysRevA.87.030102}.

\bibitem[Merkel et~al.(2010)Merkel, Riofr\'{\i}o, Flammia, and Deutsch]{merkel2010random}
Seth~T. Merkel, Carlos~A. Riofr\'{\i}o, Steven~T. Flammia, and Ivan~H. Deutsch.
\newblock Random unitary maps for quantum state reconstruction.
\newblock Phys. Rev. A, 81:\penalty0 032126, 2010.
\newblock \doi{10.1103/PhysRevA.81.032126}.

\bibitem[Sekino and Susskind(2008)]{sekino2008fast}
Yasuhiro Sekino and L.~Susskind.
\newblock Fast scramblers.
\newblock J. High Energ. Phys., 2008\penalty0 (10):\penalty0 065, 2008.
\newblock \doi{10.1088/1126-6708/2008/10/065}.

\bibitem[Styliaris et~al.(2021)Styliaris, Anand, and Zanardi]{styliaris2021information}
Georgios Styliaris, Namit Anand, and Paolo Zanardi.
\newblock Information scrambling over bipartitions: Equilibration, entropy production, and typicality.
\newblock Phys. Rev. Lett., 126:\penalty0 030601, 2021.
\newblock \doi{10.1103/PhysRevLett.126.030601}.

\bibitem[Hosur et~al.(2016)Hosur, Qi, Roberts, and Yoshida]{pawan}
Pavan Hosur, Xiao-Liang Qi, Daniel~A. Roberts, and Beni Yoshida.
\newblock Chaos in quantum channels.
\newblock J. High Energ. Phys., 2016\penalty0 (2):\penalty0 4, 2016.
\newblock \doi{10.1007/JHEP02(2016)004}.

\bibitem[Haake et~al.(2018)Haake, Gnutzmann, and Ku{\'s}]{haake1991quantum}
Fritz Haake, Sven Gnutzmann, and Marek Ku{\'s}.
\newblock Quantum Signatures of Chaos.
\newblock Springer Series in Synergetics. Springer, 4 edition, 2018.
\newblock \doi{10.1007/978-3-319-97580-1}.

\bibitem[Hayden and Preskill(2007)]{hayden2007black}
Patrick Hayden and John Preskill.
\newblock Black holes as mirrors: quantum information in random subsystems.
\newblock J. High Energ. Phys., 2007\penalty0 (09):\penalty0 120, 2007.
\newblock \doi{10.1088/1126-6708/2007/09/120}.

\bibitem[Yoshida and Kitaev(2017)]{yoshida2017efficient}
Beni Yoshida and Alexei Kitaev.
\newblock Efficient decoding for the {H}ayden-{P}reskill protocol.
\newblock arXiv:1710.03363 [hep-th], 2017.

\bibitem[Huang et~al.(2020)Huang, Kueng, and Preskill]{huang2020predicting}
Hsin-Yuan Huang, Richard Kueng, and John Preskill.
\newblock Predicting many properties of a quantum system from very few measurements.
\newblock Nat. Phys., 16:\penalty0 1050--1057, 2020.
\newblock \doi{10.1038/s41567-020-0932-7}.

\bibitem[Huang et~al.(2022)Huang, Broughton, Cotler, Chen, Li, Mohseni, Neven, Babbush, Kueng, Preskill, and McClean]{huang2022quantum}
Hsin-Yuan Huang, Michael Broughton, Jordan Cotler, Sitan Chen, Jerry Li, Masoud Mohseni, Hartmut Neven, Ryan Babbush, Richard Kueng, John Preskill, and Jarrod~R. McClean.
\newblock Quantum advantage in learning from experiments.
\newblock Science, 376:\penalty0 1182--1186, 2022.
\newblock \doi{10.1126/science.abn7293}.

\bibitem[Holmes et~al.(2021)Holmes, Arrasmith, Yan, Coles, Albrecht, and Sornborger]{holmes2021barren}
Zo\"e Holmes, Andrew Arrasmith, Bin Yan, Patrick~J. Coles, Andreas Albrecht, and Andrew~T. Sornborger.
\newblock Barren plateaus preclude learning scramblers.
\newblock Phys. Rev. Lett., 126:\penalty0 190501, 2021.
\newblock \doi{10.1103/PhysRevLett.126.190501}.

\bibitem[Tilly et~al.(2022)Tilly, Chen, Cao, Picozzi, Setia, Li, Grant, Wossnig, Rungger, Booth, and Tennyson]{tilly2022variational}
Jules Tilly, Hongxiang Chen, Shuxiang Cao, Dario Picozzi, Kanav Setia, Ying Li, Edward Grant, Leonard Wossnig, Ivan Rungger, George~H. Booth, and Jonathan Tennyson.
\newblock The variational quantum eigensolver: A review of methods and best practices.
\newblock Phys. Rep., 986:\penalty0 1--128, 2022.
\newblock \doi{10.1016/j.physrep.2022.08.003}.

\bibitem[Renes et~al.(2004)Renes, Blume-Kohout, Scott, and Caves]{renes2004symmetric}
Joseph~M. Renes, Robin Blume-Kohout, A.~J. Scott, and Carlton~M. Caves.
\newblock {Symmetric informationally complete quantum measurements}.
\newblock J. Math. Phys., 45\penalty0 (6):\penalty0 2171--2180, 2004.
\newblock \doi{10.1063/1.1737053}.

\bibitem[Klappenecker and Rotteler(2005)]{klappenecker2005mutually}
Andreas Klappenecker and Martin Rotteler.
\newblock Mutually unbiased bases are complex projective 2-designs.
\newblock In Proceedings. International Symposium on Information Theory, 2005. ISIT 2005., pages 1740--1744. IEEE, 2005.
\newblock \doi{10.1109/ISIT.2005.1523643}.

\bibitem[Morvan et~al.(2021)Morvan, Ramasesh, Blok, Kreikebaum, O'Brien, Chen, Mitchell, Naik, Santiago, and Siddiqi]{morvan2021qutrit}
A.~Morvan, V.~V. Ramasesh, M.~S. Blok, J.~M. Kreikebaum, K.~O'Brien, L.~Chen, B.~K. Mitchell, R.~K. Naik, D.~I. Santiago, and I.~Siddiqi.
\newblock Qutrit randomized benchmarking.
\newblock Phys. Rev. Lett., 126:\penalty0 210504, 2021.
\newblock \doi{10.1103/PhysRevLett.126.210504}.

\bibitem[Proctor et~al.(2022)Proctor, Seritan, Rudinger, Nielsen, Blume-Kohout, and Young]{proctor2022scalable}
Timothy Proctor, Stefan Seritan, Kenneth Rudinger, Erik Nielsen, Robin Blume-Kohout, and Kevin Young.
\newblock Scalable randomized benchmarking of quantum computers using mirror circuits.
\newblock Phys. Rev. Lett., 129:\penalty0 150502, 2022.
\newblock \doi{10.1103/PhysRevLett.129.150502}.

\bibitem[Boixo et~al.(2018)Boixo, Isakov, Smelyanskiy, Babbush, Ding, Jiang, Bremner, Martinis, and Neven]{boixo2018characterizing}
Sergio Boixo, Sergei~V. Isakov, Vadim~N. Smelyanskiy, Ryan Babbush, Nan Ding, Zhang Jiang, Michael~J. Bremner, John~M. Martinis, and Hartmut Neven.
\newblock Characterizing quantum supremacy in near-term devices.
\newblock Nat. Phys., 14:\penalty0 595--600, 2018.
\newblock \doi{10.1038/s41567-018-0124-x}.

\bibitem[Gross and Bloch(2017)]{gross2017quantum}
Christian Gross and Immanuel Bloch.
\newblock Quantum simulations with ultracold atoms in optical lattices.
\newblock Science, 357:\penalty0 995--1001, 2017.
\newblock \doi{10.1126/science.aal3837}.

\bibitem[Blatt and Roos(2012)]{blatt2012quantum}
R.~Blatt and C.~F. Roos.
\newblock Quantum simulations with trapped ions.
\newblock Nat. Phys., 8:\penalty0 277--284, 2012.
\newblock \doi{10.1038/nphys2252}.

\bibitem[Browaeys et~al.(2016)Browaeys, Barredo, and Lahaye]{browaeys2016experimental}
Antoine Browaeys, Daniel Barredo, and Thierry Lahaye.
\newblock Experimental investigations of dipole–dipole interactions between a few rydberg atoms.
\newblock J. Phys. B: At. Mol. Opt. Phys., 49\penalty0 (15):\penalty0 152001, 2016.
\newblock \doi{10.1088/0953-4075/49/15/152001}.

\bibitem[Gambetta et~al.(2017)Gambetta, Chow, and Steffen]{gambetta2017building}
Jay~M. Gambetta, Jerry~M. Chow, and Matthias Steffen.
\newblock Building logical qubits in a superconducting quantum computing system.
\newblock npj Quantum Inf, 3:\penalty0 2, 2017.
\newblock \doi{10.1038/s41534-016-0004-0}.

\bibitem[Cotler et~al.(2023)Cotler, Mark, Huang, Hern\'andez, Choi, Shaw, Endres, and Choi]{cotler2023emergent}
Jordan~S. Cotler, Daniel~K. Mark, Hsin-Yuan Huang, Felipe Hern\'andez, Joonhee Choi, Adam~L. Shaw, Manuel Endres, and Soonwon Choi.
\newblock Emergent quantum state designs from individual many-body wave functions.
\newblock PRX Quantum, 4:\penalty0 010311, 2023.
\newblock \doi{10.1103/PRXQuantum.4.010311}.

\bibitem[Choi et~al.(2023)Choi, Shaw, Madjarov, Xie, Finkelstein, Covey, Cotler, Mark, Huang, Kale, Pichler, Brand{\~a}o, Choi, and Endres]{choi2023preparing}
Joonhee Choi, Adam~L. Shaw, Ivaylo~S. Madjarov, Xin Xie, Ran Finkelstein, Jacob~P. Covey, Jordan~S. Cotler, Daniel~K. Mark, Hsin-Yuan Huang, Anant Kale, Hannes Pichler, Fernando G. S.~L. Brand{\~a}o, Soonwon Choi, and Manuel Endres.
\newblock Preparing random states and benchmarking with many-body quantum chaos.
\newblock Nature, 613:\penalty0 468--473, 2023.
\newblock \doi{10.1038/s41586-022-05442-1}.

\bibitem[Deutsch(1991)]{deutsch1991quantum}
J.~M. Deutsch.
\newblock Quantum statistical mechanics in a closed system.
\newblock Phys. Rev. A, 43:\penalty0 2046, 1991.
\newblock \doi{10.1103/PhysRevA.43.2046}.

\bibitem[Srednicki(1994)]{srednicki1994chaos}
Mark Srednicki.
\newblock Chaos and quantum thermalization.
\newblock Phys. Rev. E, 50:\penalty0 888, 1994.
\newblock \doi{10.1103/PhysRevE.50.888}.

\bibitem[Rigol et~al.(2008)Rigol, Dunjko, and Olshanii]{ETH_ansatz_expt}
Marcos Rigol, Vanja Dunjko, and Maxim Olshanii.
\newblock Thermalization and its mechanism for generic isolated quantum systems.
\newblock Nature, 452:\penalty0 854--858, 2008.
\newblock \doi{10.1038/nature06838}.

\bibitem[D'Alessio et~al.(2016)D'Alessio, Kafri, Polkovnikov, and Rigol]{d2016quantum}
Luca D'Alessio, Yariv Kafri, Anatoli Polkovnikov, and Marcos Rigol.
\newblock From quantum chaos and eigenstate thermalization to statistical mechanics and thermodynamics.
\newblock Advances in Physics, 65\penalty0 (3):\penalty0 239--362, 2016.
\newblock \doi{10.1080/00018732.2016.1198134}.

\bibitem[Deutsch(2018)]{deutsch_18}
J.~M. Deutsch.
\newblock Eigenstate thermalization hypothesis.
\newblock Rep. Prog. Phys., 81:\penalty0 082001, 2018.
\newblock \doi{10.1088/1361-6633/aac9f1}.

\bibitem[Roy et~al.(2023)Roy, Bandyopadhyay, de~Almeida, and Hauke]{eth_nonherm}
Sudipto~Singha Roy, Soumik Bandyopadhyay, Ricardo~Costa de~Almeida, and Philipp Hauke.
\newblock Unveiling eigenstate thermalization for non-hermitian systems.
\newblock arXiv:2309.00049 [quant-ph], 2023.

\bibitem[Ho and Choi(2022)]{ho2022exact}
Wen~Wei Ho and Soonwon Choi.
\newblock Exact emergent quantum state designs from quantum chaotic dynamics.
\newblock Phys. Rev. Lett., 128:\penalty0 060601, 2022.
\newblock \doi{10.1103/PhysRevLett.128.060601}.

\bibitem[Ippoliti and Ho(2022)]{ippoliti2022solvable}
Matteo Ippoliti and Wen~Wei Ho.
\newblock Solvable model of deep thermalization with distinct design times.
\newblock {Quantum}, 6:\penalty0 886, 2022.
\newblock \doi{10.22331/q-2022-12-29-886}.

\bibitem[Ippoliti and Ho(2023)]{ippoliti2023dynamical}
Matteo Ippoliti and Wen~Wei Ho.
\newblock Dynamical purification and the emergence of quantum state designs from the projected ensemble.
\newblock PRX Quantum, 4:\penalty0 030322, 2023.
\newblock \doi{10.1103/PRXQuantum.4.030322}.

\bibitem[Lucas et~al.(2023)Lucas, Piroli, De~Nardis, and De~Luca]{lucas2023generalized}
Maxime Lucas, Lorenzo Piroli, Jacopo De~Nardis, and Andrea De~Luca.
\newblock Generalized deep thermalization for free fermions.
\newblock Phys. Rev. A, 107:\penalty0 032215, 2023.
\newblock \doi{10.1103/PhysRevA.107.032215}.

\bibitem[Shrotriya and Ho(2023)]{shrotriya2023nonlocality}
Harshank Shrotriya and Wen~Wei Ho.
\newblock Nonlocality of deep thermalization.
\newblock arXiv:2305.08437 [quant-ph], 2023.

\bibitem[Versini et~al.(2023)Versini, El-Din, Mintert, and Mukherjee]{versini2023efficient}
Lorenzo Versini, Karim~Alaa El-Din, Florian Mintert, and Rick Mukherjee.
\newblock Efficient estimation of quantum state k-designs with randomized measurements.
\newblock arXiv:2305.01465 [quant-ph], 2023.

\bibitem[Claeys and Lamacraft(2022)]{claeys2022emergent}
Pieter~W. Claeys and Austen Lamacraft.
\newblock Emergent quantum state designs and biunitarity in dual-unitary circuit dynamics.
\newblock {Quantum}, 6:\penalty0 738, 2022.
\newblock \doi{10.22331/q-2022-06-15-738}.

\bibitem[Bhore et~al.(2023)Bhore, Desaules, and Papi\ifmmode~\acute{c}\else \'{c}\fi{}]{bhore2023deep}
Tanmay Bhore, Jean-Yves Desaules, and Zlatko Papi\ifmmode~\acute{c}\else \'{c}\fi{}.
\newblock Deep thermalization in constrained quantum systems.
\newblock Phys. Rev. B, 108:\penalty0 104317, 2023.
\newblock \doi{10.1103/PhysRevB.108.104317}.

\bibitem[McGinley and Fava(2023)]{mcginley2023shadow}
Max McGinley and Michele Fava.
\newblock Shadow tomography from emergent state designs in analog quantum simulators.
\newblock Phys. Rev. Lett., 131:\penalty0 160601, 2023.
\newblock \doi{10.1103/PhysRevLett.131.160601}.

\bibitem[Khemani et~al.(2018)Khemani, Vishwanath, and Huse]{ope4}
Vedika Khemani, Ashvin Vishwanath, and David~A. Huse.
\newblock Operator spreading and the emergence of dissipative hydrodynamics under unitary evolution with conservation laws.
\newblock Phys. Rev. X, 8:\penalty0 031057, 2018.
\newblock \doi{10.1103/PhysRevX.8.031057}.

\bibitem[Rakovszky et~al.(2018)Rakovszky, Pollmann, and von Keyserlingk]{ope5}
Tibor Rakovszky, Frank Pollmann, and C.~W. von Keyserlingk.
\newblock Diffusive hydrodynamics of out-of-time-ordered correlators with charge conservation.
\newblock Phys. Rev. X, 8:\penalty0 031058, 2018.
\newblock \doi{10.1103/PhysRevX.8.031058}.

\bibitem[Friedman et~al.(2019)Friedman, Chan, De~Luca, and Chalker]{friedman2019spectral}
Aaron~J. Friedman, Amos Chan, Andrea De~Luca, and J.~T. Chalker.
\newblock Spectral statistics and many-body quantum chaos with conserved charge.
\newblock Phys. Rev. Lett., 123:\penalty0 210603, 2019.
\newblock \doi{10.1103/PhysRevLett.123.210603}.

\bibitem[Yunger~Halpern et~al.(2016)Yunger~Halpern, Faist, Oppenheim, and Winter]{yunger2016microcanonical}
Nicole Yunger~Halpern, Philippe Faist, Jonathan Oppenheim, and Andreas Winter.
\newblock Microcanonical and resource-theoretic derivations of the thermal state of a quantum system with noncommuting charges.
\newblock Nat Commun, 7:\penalty0 12051, 2016.
\newblock \doi{10.1038/ncomms12051}.

\bibitem[Nakata et~al.(2023)Nakata, Wakakuwa, and Koashi]{nakata2023black}
Yoshifumi Nakata, Eyuri Wakakuwa, and Masato Koashi.
\newblock Black holes as clouded mirrors: the {H}ayden-{P}reskill protocol with symmetry.
\newblock {Quantum}, 7:\penalty0 928, 2023.
\newblock \doi{10.22331/q-2023-02-21-928}.

\bibitem[Bhattacharya et~al.(2017)Bhattacharya, Chakrabarti, Jatkar, and Kundu]{bhattacharya2017syk}
Ritabrata Bhattacharya, Subhroneel Chakrabarti, Dileep~P. Jatkar, and Arnab Kundu.
\newblock Syk model, chaos and conserved charge.
\newblock J. High Energ. Phys., 2017\penalty0 (11):\penalty0 180, 2017.
\newblock \doi{10.1007/JHEP11(2017)180}.

\bibitem[Balachandran et~al.(2021)Balachandran, Benenti, Casati, and Poletti]{balachandran2021eigenstate}
Vinitha Balachandran, Giuliano Benenti, Giulio Casati, and Dario Poletti.
\newblock From the eigenstate thermalization hypothesis to algebraic relaxation of otocs in systems with conserved quantities.
\newblock Phys. Rev. B, 104:\penalty0 104306, 2021.
\newblock \doi{10.1103/PhysRevB.104.104306}.

\bibitem[Kudler-Flam et~al.(2022)Kudler-Flam, Sohal, and Nie]{kudler2022information}
Jonah Kudler-Flam, Ramanjit Sohal, and Laimei Nie.
\newblock {Information Scrambling with Conservation Laws}.
\newblock SciPost Phys., 12:\penalty0 117, 2022.
\newblock \doi{10.21468/SciPostPhys.12.4.117}.

\bibitem[Chen et~al.(2020)Chen, Gu, and Lucas]{chen2020many}
Xiao Chen, Yingfei Gu, and Andrew Lucas.
\newblock {Many-body quantum dynamics slows down at low density}.
\newblock SciPost Phys., 9:\penalty0 071, 2020.
\newblock \doi{10.21468/SciPostPhys.9.5.071}.

\bibitem[Paviglianiti et~al.(2023)Paviglianiti, Bandyopadhyay, Uhrich, and Hauke]{paviglianiti2023absence}
Alessio Paviglianiti, Soumik Bandyopadhyay, Philipp Uhrich, and Philipp Hauke.
\newblock Absence of operator growth for average equal-time observables in charge-conserved sectors of the {S}achdev-{Y}e-{K}itaev model.
\newblock J. High Energ. Phys., 2023:\penalty0 126, 2023.
\newblock \doi{10.1007/JHEP03(2023)126}.

\bibitem[Agarwal et~al.(2023)Agarwal, Sahu, and Xu]{agarwal2023charge}
Lakshya Agarwal, Subhayan Sahu, and Shenglong Xu.
\newblock Charge transport, information scrambling and quantum operator-coherence in a many-body system with u(1) symmetry.
\newblock J. High Energ. Phys., 2023:\penalty0 37, 2023.
\newblock \doi{10.1007/JHEP05(2023)037}.

\bibitem[Varikuti and Madhok(2024)]{varikuti2022out}
Naga~Dileep Varikuti and Vaibhav Madhok.
\newblock {Out-of-time ordered correlators in kicked coupled tops: Information scrambling in mixed phase space and the role of conserved quantities}.
\newblock Chaos, 34:\penalty0 063124, 2024.
\newblock \doi{10.1063/5.0191140}.

\bibitem[Gioia and Wang(2022)]{gioia2022nonzero}
Lei Gioia and Chong Wang.
\newblock Nonzero momentum requires long-range entanglement.
\newblock Phys. Rev. X, 12:\penalty0 031007, 2022.
\newblock \doi{10.1103/PhysRevX.12.031007}.

\bibitem[Santos and Rigol(2010)]{santos2010localization}
Lea~F. Santos and Marcos Rigol.
\newblock Localization and the effects of symmetries in the thermalization properties of one-dimensional quantum systems.
\newblock Phys. Rev. E, 82:\penalty0 031130, 2010.
\newblock \doi{10.1103/PhysRevE.82.031130}.

\bibitem[Mori et~al.(2018)Mori, Ikeda, Kaminishi, and Ueda]{mori2018thermalization}
Takashi Mori, Tatsuhiko~N Ikeda, Eriko Kaminishi, and Masahito Ueda.
\newblock Thermalization and prethermalization in isolated quantum systems: a theoretical overview.
\newblock J. Phys. B: At. Mol. Opt. Phys., 51:\penalty0 112001, 2018.
\newblock \doi{10.1088/1361-6455/aabcdf}.

\bibitem[Sugimoto et~al.(2023)Sugimoto, Henheik, Riabov, and Erd{\H{o}}s]{sugimoto2023eigenstate}
Shoki Sugimoto, Joscha Henheik, Volodymyr Riabov, and L{\'a}szl{\'o} Erd{\H{o}}s.
\newblock Eigenstate thermalisation hypothesis for translation invariant spin systems.
\newblock J Stat Phys, 190:\penalty0 128, 2023.
\newblock \doi{10.1007/s10955-023-03132-4}.

\bibitem[Linden et~al.(1998)Linden, Popescu, and Popescu]{linden1998multi}
N.~Linden, S.~Popescu, and S.~Popescu.
\newblock On multi-particle entanglement.
\newblock Fortschritte der Physik, 46\penalty0 (4-5):\penalty0 567--578, 1998.
\newblock \doi{10.1002/(SICI)1521-3978(199806)46:4/5<567::AID-PROP567>3.0.CO;2-H}.

\bibitem[Nakata and Murao(2020)]{nakata2020generic}
Yoshifumi Nakata and Mio Murao.
\newblock Generic entanglement entropy for quantum states with symmetry.
\newblock Entropy, 22\penalty0 (6), 2020.
\newblock \doi{10.3390/e22060684}.

\bibitem[Mark et~al.(2024)Mark, Surace, Elben, Shaw, Choi, Refael, Endres, and Choi]{mark2024maximum}
Daniel~K. Mark, Federica Surace, Andreas Elben, Adam~L. Shaw, Joonhee Choi, Gil Refael, Manuel Endres, and Soonwon Choi.
\newblock A maximum entropy principle in deep thermalization and in hilbert-space ergodicity.
\newblock arXiv:2403.11970 [quant-ph], 2024.

\bibitem[Zhang(2014)]{zhang2014matrix}
Lin Zhang.
\newblock Matrix integrals over unitary groups: An application of {S}chur-{W}eyl duality.
\newblock arXiv:1408.3782 [quant-ph], 2014.

\bibitem[Kim et~al.(2014)Kim, Ikeda, and Huse]{kim2014testing}
Hyungwon Kim, Tatsuhiko~N. Ikeda, and David~A. Huse.
\newblock Testing whether all eigenstates obey the eigenstate thermalization hypothesis.
\newblock Phys. Rev. E, 90:\penalty0 052105, 2014.
\newblock \doi{10.1103/PhysRevE.90.052105}.

\bibitem[Mishra et~al.(2015)Mishra, Lakshminarayan, and Subrahmanyam]{mishra2015protocol}
Sunil~K. Mishra, Arul Lakshminarayan, and V.~Subrahmanyam.
\newblock Protocol using kicked {I}sing dynamics for generating states with maximal multipartite entanglement.
\newblock Phys. Rev. A, 91:\penalty0 022318, 2015.
\newblock \doi{10.1103/PhysRevA.91.022318}.

\bibitem[Pal and Lakshminarayan(2018)]{pal2018entangling}
Rajarshi Pal and Arul Lakshminarayan.
\newblock Entangling power of time-evolution operators in integrable and nonintegrable many-body systems.
\newblock Phys. Rev. B, 98:\penalty0 174304, 2018.
\newblock \doi{10.1103/PhysRevB.98.174304}.

\bibitem[Vinayak and Žnidarič(2012)]{vznidarivc2012subsystem}
Vinayak and Marko Žnidarič.
\newblock Subsystem dynamics under random hamiltonian evolution.
\newblock J. Phys. A: Math. Theor., 45:\penalty0 125204, 2012.
\newblock \doi{10.1088/1751-8113/45/12/125204}.

\bibitem[Bensa and \ifmmode \check{Z}\else \v{Z}\fi{}nidari\ifmmode~\check{c}\else \v{c}\fi{}(2022)]{bensa2022two}
Ja\ifmmode \check{s}\else~\v{s}\fi{} Bensa and Marko \ifmmode \check{Z}\else \v{Z}\fi{}nidari\ifmmode~\check{c}\else \v{c}\fi{}.
\newblock Two-step phantom relaxation of out-of-time-ordered correlations in random circuits.
\newblock Phys. Rev. Res., 4:\penalty0 013228, 2022.
\newblock \doi{10.1103/PhysRevResearch.4.013228}.

\bibitem[Žnidarič(2023)]{vznidarivc2023two}
Marko Žnidarič.
\newblock Two-step relaxation in local many-body floquet systems.
\newblock J. Phys. A: Math. Theor., 56\penalty0 (43):\penalty0 434001, 2023.
\newblock \doi{10.1088/1751-8121/acfc05}.

\bibitem[Majidy et~al.(2023{\natexlab{a}})Majidy, Braasch, Lasek, Upadhyaya, Kalev, and Yunger~Halpern]{majidy2023noncommuting}
Shayan Majidy, William~F. Braasch, Aleksander Lasek, Twesh Upadhyaya, Amir Kalev, and Nicole Yunger~Halpern.
\newblock Noncommuting conserved charges in quantum thermodynamics and beyond.
\newblock Nat Rev Phys, 5:\penalty0 689--698, 2023{\natexlab{a}}.
\newblock \doi{10.1038/s42254-023-00641-9}.

\bibitem[Kranzl et~al.(2023)Kranzl, Lasek, Joshi, Kalev, Blatt, Roos, and Yunger~Halpern]{kranzl2023experimental}
Florian Kranzl, Aleksander Lasek, Manoj~K. Joshi, Amir Kalev, Rainer Blatt, Christian~F. Roos, and Nicole Yunger~Halpern.
\newblock Experimental observation of thermalization with noncommuting charges.
\newblock PRX Quantum, 4:\penalty0 020318, 2023.
\newblock \doi{10.1103/PRXQuantum.4.020318}.

\bibitem[Skinner et~al.(2019)Skinner, Ruhman, and Nahum]{skinner2019measurement}
Brian Skinner, Jonathan Ruhman, and Adam Nahum.
\newblock Measurement-induced phase transitions in the dynamics of entanglement.
\newblock Phys. Rev. X, 9:\penalty0 031009, 2019.
\newblock \doi{10.1103/PhysRevX.9.031009}.

\bibitem[Majidy et~al.(2023{\natexlab{b}})Majidy, Agrawal, Gopalakrishnan, Potter, Vasseur, and Halpern]{majidy2023critical}
Shayan Majidy, Utkarsh Agrawal, Sarang Gopalakrishnan, Andrew~C. Potter, Romain Vasseur, and Nicole~Yunger Halpern.
\newblock Critical phase and spin sharpening in su(2)-symmetric monitored quantum circuits.
\newblock Phys. Rev. B, 108:\penalty0 054307, 2023{\natexlab{b}}.
\newblock \doi{10.1103/PhysRevB.108.054307}.

\bibitem[O'Searcoid(2007)]{o2006metric}
M{\'\i}che{\'a}l O'Searcoid.
\newblock Metric spaces.
\newblock Springer Undergraduate Mathematics Series. Springer London, 1 edition, 2007.
\newblock \doi{https://doi.org/10.1007/978-1-84628-627-8}.

\bibitem[Milman and Schechtman(1986)]{milman1986asymptotic}
Vitali~D Milman and Gideon Schechtman.
\newblock Asymptotic theory of finite dimensional normed spaces: Isoperimetric inequalities in riemannian manifolds.
\newblock Lecture Notes in Mathematics. Springer Berlin, Heidelberg, 1 edition, 1986.
\newblock \doi{https://doi.org/10.1007/978-3-540-38822-7}.

\bibitem[Ledoux(2001)]{ledoux2001concentration}
Michel Ledoux.
\newblock The concentration of measure phenomenon.
\newblock Number~89. American Mathematical Soc., 2001.

\bibitem[Gerken(2013)]{gerken2013measure}
Manuel Gerken.
\newblock Measure concentration: Levy’s lemma.
\newblock Lecture Notes for Talk, 6, 2013.

\bibitem[Kim and Huse(2013)]{kim2023nonintIsing}
Hyungwon Kim and David~A. Huse.
\newblock Ballistic spreading of entanglement in a diffusive nonintegrable system.
\newblock Phys. Rev. Lett., 111:\penalty0 127205, 2013.
\newblock \doi{10.1103/PhysRevLett.111.127205}.

\bibitem[Sharma et~al.(2015)Sharma, Suzuki, and Dutta]{sharma2015quenches}
Shraddha Sharma, Sei Suzuki, and Amit Dutta.
\newblock Quenches and dynamical phase transitions in a nonintegrable quantum ising model.
\newblock Phys. Rev. B, 92:\penalty0 104306, 2015.
\newblock \doi{10.1103/PhysRevB.92.104306}.

\end{thebibliography}

\end{document}